\documentclass[11pt]{article}
\usepackage{verbatim} 
\usepackage[usenames, dvipsnames]{color}

\usepackage{bbm}
\usepackage{mathtools}
\usepackage{amsmath}
\usepackage{amsthm, thmtools}
\usepackage{amstext,amssymb, amsfonts}

\usepackage{dsfont} 
 
\usepackage{empheq} 

\usepackage{lmodern} 

\usepackage{float} 

\usepackage{array,multirow} 

\usepackage{algorithmic}
\usepackage[section]{algorithm} 
\usepackage{todonotes}

\newcounter{mynotes}
\usepackage{hyperref} 

\usepackage[symbol]{footmisc}

\declaretheorem[within=section]{theorem}
\declaretheorem[sibling=theorem]{corollary}

\declaretheorem[sibling=theorem]{lemma}
\declaretheorem[sibling=theorem]{claim}

\declaretheorem[sibling=theorem]{definition}

\declaretheorem[sibling=theorem]{proposition}
\declaretheorem[sibling=theorem]{remark}

\declaretheorem[sibling=theorem]{conjecture}

\declaretheorem[sibling=theorem]{fact}
\def\ba{{\mathbf a}}
\def\bb{{\mathbf b}}

\def\bu{{\mathbf u}}
\def\bv{{\mathbf v}}

\def\bx{{\mathbf x}}
\def\by{{\mathbf y}}

\newcommand{\R}{\mathbb{R}} 
\newcommand{\Z}{\mathbb{Z}} 

\newcommand{\F}{\mathbb{F}}

\newcommand{\cD}{\mathcal D}

\newcommand{\cF}{\mathcal F}

\DeclarePairedDelimiter\ceil{\lceil}{\rceil} 
\newcommand{\bigexp}[1]{\exp\left(#1\right)} 
\newcommand{\indicator}{\mathds{1}} 

\newcommand{\inpro}[2]{\left\langle #1,#2 \right\rangle} 
\newcommand{\linearspan}[1]{\mathsf{span}\{#1\}} 

\newcommand{\E}{\mathbb{E}} 

\newcommand{\poly}{\mathrm{poly}}

\newcommand{\set}[1]{\{#1\}}

\newcommand{\bits}{\set{0,1}}
\newcommand{\sbits}{\set{-1,1}}

\renewcommand{\epsilon}{\varepsilon}
\newcommand{\eps}{\epsilon}

  \newcommand{\beq}{\begin{equation}}
  \newcommand{\eeq}{\end{equation}}
  \newcommand{\beqn}{\begin{equation*}}
  \newcommand{\eeqn}{\end{equation*}}
  \newcommand{\beqr}{\begin{eqnarray}}
  \newcommand{\eeqr}{\end{eqnarray}}
  \newcommand{\beqrn}{\begin{eqnarray*}}
  \newcommand{\eeqrn}{\end{eqnarray*}}
  \newcommand{\bmline}{\begin{multline}}
  \newcommand{\emline}{\end{multline}}
  \newcommand{\bmlinen}{\begin{multline*}}
  \newcommand{\emlinen}{\end{multline*}}

\usepackage[margin=1in]{geometry}
\usepackage{url}

\renewcommand{\le}{\leqslant}

\renewcommand{\ge}{\geqslant}

\newcommand{\modcsp}{\operatorname{\mathsf{MOD-CSP}}}
\newcommand{\Pol}{\operatorname{Pol}}
\newcommand{\FindMinimal}{\operatorname{\mathsf{FindMinimal}}}
\newcommand{\ar}{\mathsf{arity}}
\newcommand{\CSP}{\mathsf{CSP}}
\newcommand{\OR}{\mathsf{OR}}
\newcommand{\MAJ}{\mathsf{MAJ}}
\newcommand{\AND}{\mathsf{AND}}
\newcommand{\NAND}{\mathsf{NAND}}
\newcommand{\ETH}{\mathsf{ETH}}

\newcommand{\ot}{\leftarrow}
\newcommand{\Ham}{\mathsf{Ham}}

\newcommand{\allones}{\mathbf{1}}

\newcommand{\cov}{\operatorname{\mathsf{cov}}}

\newcommand{\ind}{\indicator}
\newcommand{\bset}{\bits}
\newcommand{\ones}{\allones}
\newcommand{\MOD}{\mathsf{MOD}}
\newcommand{\linmod}{\operatorname{\mathsf{LIN-2-MOD}}}
\newcommand{\hornsatmod}{\operatorname{\mathsf{HORN-SAT-MOD}}}
\newcommand{\dualhornsatmod}{\operatorname{\mathsf{DUAL-HORN-SAT-MOD}}}
\newcommand{\twosat}{\operatorname{\mathsf{2-SAT}}}
\newcommand{\threesat}{\operatorname{\mathsf{3-SAT}}}
\newcommand{\lintwo}{\operatorname{\mathsf{LIN-2}}}
\newcommand{\hornsat}{\operatorname{\mathsf{HORN-SAT}}}
\newcommand{\dualhornsat}{\operatorname{\mathsf{DUAL-HORN-SAT}}}
\newcommand{\threexor}{\operatorname{\mathsf{3-XOR}}}
\newcommand{\twosatmod}{\operatorname{\mathsf{2-SAT-MOD}}}
\newcommand{\submod}{\operatorname{\mathsf{SUBMOD-MIN-MOD}}}
\newcommand{\hornsatnoconst}{\operatorname{\mathsf{HORN-SAT-NO-CONSTANTS}}}
\newcommand{\dualhornsatnoconst}{\operatorname{\mathsf{DUAL-HORN-SAT-NO-CONSTANTS}}}
\newcommand{\lintwoevenzero}{\operatorname{\mathsf{LIN-2-EVEN-ZERO}}}
\newcommand{\coeffnorm}[1]{\left|#1\right|}
\newcommand{\ORdim}{\cD}

\title{CSPs with Global Modular Constraints: \\
	Algorithms and Hardness via Polynomial Representations}

\author{Joshua Brakensiek\thanks{Department of Computer Science, Stanford University, Stanford, CA. Email: {\tt jbrakens@stanford.edu}. Portions of this work were done while at Carnegie Mellon University and during a visit to Microsoft Research, Redmond. Research supported in part by NSF CCF-1526092 and an NSF Graduate Research Fellowship.} \and Sivakanth Gopi\thanks{Microsoft Research, Redmond, WA. Email: {\tt sigopi@microsoft.com}. Part of the research was done while the author was a student at Princeton University where he was supported by NSF CAREER award 1451191 and NSF grant CCF-1523816.} \and Venkatesan Guruswami\thanks{Computer Science Department, Carnegie Mellon University, Pittsburgh, PA 15213. Email: {\tt venkatg@cs.cmu.edu}. Research supported in part by NSF grants CCF-1422045 and CCF-1526092.} }

\date{}
\begin{document}

\maketitle
\thispagestyle{empty}

\begin{abstract}
We study the complexity of Boolean constraint satisfaction problems (CSPs) when the assignment must have Hamming weight in some congruence class modulo $M$, for various choices of the modulus $M$. Due to the known classification of tractable Boolean CSPs, this mainly reduces to the study of three cases: $\twosat$, $\hornsat$, and $\lintwo$ (linear equations mod $2$). We classify the moduli $M$ for which these respective problems are polynomial time solvable, and when they are not (assuming the ETH). Our study reveals that this modular constraint lends a surprising richness to these classic, well-studied problems, with interesting broader connections to complexity theory and coding theory. The $\hornsat$ case is connected to the covering complexity of polynomials representing the $\NAND$ function mod $M$. The $\lintwo$ case is tied to the sparsity of polynomials representing the $\OR$ function mod $M$, which in turn has connections to modular weight distribution properties of linear codes and locally decodable codes. In both cases, the analysis of our algorithm as well as the hardness reduction rely on these polynomial representations, highlighting an interesting algebraic common ground between hard cases for our algorithms and the gadgets which show hardness. These new complexity measures of polynomial representations merit further study.

The inspiration for our study comes from a recent work by N\"{a}gele, Sudakov, and Zenklusen on submodular minimization with a global congruence constraint. Our algorithm for $\hornsat$ has strong similarities to their algorithm, and in particular identical kind of set systems arise in both cases. Our connection to polynomial representations leads to a simpler analysis of such set systems, and also sheds light on (but does not resolve) the complexity of submodular minimization with a congruency requirement modulo a composite $M$.
\end{abstract}

\vspace{-3ex}
\tableofcontents
\thispagestyle{empty}

\pagebreak

\section{Introduction}\label{sec:introduction}

We study how the complexity of tractable cases of Boolean constraint satisfaction problems, namely $\twosat$, $\hornsat$ and $\lintwo$ (linear equations over $\mathbb F_2$), is affected when we seek a solution that obeys a global modular constraint, such as Hamming weight being divisible by some modulus $M$. As our work reveals, this seemingly simple twist lends a remarkable amount of richness to these classic problems, raising new questions concerning polynomial representations of simple Boolean functions that form the common meeting ground of both algorithmic and hardness results.

The inspiration for our study comes from a beautiful recent work on minimizing a submodular function in the presence of a global modular constraint~\cite{NageleSZ18}. This framework captures questions such as:  Given a graph $G$, find the minimum cut one of whose sides has size divisible by $6$. The complexity of this basic question remains open. Surprisingly, the same combinatorial set system that governed the complexity of the algorithms in \cite{NageleSZ18} arises in our study of $\hornsat$ with a modular constraint. We connect such set systems to polynomial representations of the NAND function, thereby shedding light on submodular minimization with a global constraint involving a composite modulus, a case not handled by \cite{NageleSZ18}. We describe this  connection, as well as the relation to integer programs of bounded modularity that partly motivated \cite{NageleSZ18} in Section~\ref{subsec:intro-IP}.

\subsection{CSPs and modular CSPs}

We now describe some of our other, intrinsic motivations to study the modular variant of constraint satisfaction problems (CSPs). CSPs have a storied place in computational complexity, spurring several of its most influential developments such as NP-completeness, Schaefer's dichotomy theorem~\cite{Schaefer:1978}, the PCP theorem and the Unique Games conjecture (which together have led to the very rich field of inapproximability), and the algebraic program for studying CSPs inspired by the Feder-Vardi~\cite{DBLP:journals/siamcomp/FederV98} and crystallized by~Bulatov, Jeavons and Krokhin~\cite{Bulatov2005}, which recently culminated in the resolution of the CSP dichotomy conjecture~\cite{DBLP:conf/focs/Bulatov17,DBLP:conf/focs/Zhuk17}.

One reason that CSPs receive so much attention is that the local nature  of their constraints offers just the right amount of structure to aid the development of novel algorithmic and hardness techniques, which then often extend to more general settings. For instance, semidefinite programming which was first used in approximating the Max-Cut CSP, has been one of the most influential algorithmic tools in approximating a whole variety of problems. On the hardness side, the PCP theorem, which is a statement about hardness of approximating Max-CSP, in combination with clever gadgets has led to inapproximability results for covering, packing, cut, routing, and other classes of problems. Further, for the problem of satisfiability of CSPs, we have a precise understanding of the interplay between mathematical structure and tractability: efficient algorithms exist iff the problem admits non-trivial ``polymorphisms" which are operations under which the solution space is closed.

One enhancement to the CSP framework would be to impose some \emph{global} constraint on the solution. For example, one could impose a global cardinality constraint, such as an equal number of $0$'s and $1$'s in the solution, or more generally a specified frequency for each value in the domain.
This global condition is quite strict, often making many tractable CSPs \textsf{NP}-complete once these constraints are added. In fact, in the Boolean case the ``hardest'' problem that can be solved in polynomial time is (weighted)graph $2$-coloring (by doing a simple dynamic program on the connected components)~\cite{Creignou:2010:NBC:1805950.1805954}.

In the case of approximating a $\twosat$ instance with a balanced cardinality constraint (an equal number of $0$s and $1$s), it is NP-hard to solve \cite{Guruswami2SAT2016}. In fact, the authors show that it is NP-hard to find an assignment satisfying a $(1-\epsilon_0)$ fraction of clauses for some absolute constant $\epsilon_0 > 0$. Further, this inapproximability holds for $\twosat$ instances with Horn clauses, and thus also implies hardness of $\hornsat$ with a balanced cardinality constraint.

This strictness of the constraints allowed for a full dichotomy (for any sized domain) to be proved long before the Feder-Vardi dichotomy was resolved, as most problems become NP-hard~\cite{lmcs:1025}. Bulatov and Marx showed that the only tractable problems are those that are ``non-crossing decomposable.'' Although the formal definition is a bit technical, informally such problems need to both be ``convex'' in that they are tractable in the second round of the Sherali-Adams hierarchy as well as ``linear'' in that they are solvable using a variant of Gaussian elimination (c.f., \cite{DBLP:conf/dagstuhl/BartoKW17}). As those two types of algorithms typically solve quite different problems (think $\twosat$ versus $\threexor$), the family of tractable cardinality CSPs is far less diverse than for ``ordinary'' CSPs.

The main focus of this work is to investigate CSPs with a much less strict global constraint, which we refer to as a \emph{modular} (or \emph{congruency}) constraint.  We will restrict ourselves to the Boolean case in this work, and impose the requirement that the number of $1$'s in the solution be congruent to $\ell$ modulo $M$, for some integers $\ell,M$.\footnote{Our actual setup is a bit more general, associating a weight from an abelian group for each variable-value pair, and requiring that the sum of the weights equals some value.}
We refer to this class of problems as (Boolean) Mod-CSPs.

 Informally, it is easy to see that any such Mod-CSP is at least as hard as the corresponding local CSP, as we can take a local CSP instance and add $M$ dummy variables not part of any clauses, so that the modular constraint is now trivially satisfiable.  Conversely, these Mod-CSPs are at least as easy (up to polynomial factors) as the corresponding cardinality problem, because we can brute force the cardinality of $1$s by trying all $c_1 =\ell\mod M$.

 By Schaefer's Boolean CSP dichotomy theorem~\cite{Schaefer:1978}, there are only three essentially different non-trivial tractable cases of Boolean CSP: $\twosat$, $\hornsat$, and $\lintwo$. We thus study each of these problems when we seek a solution of Hamming weight in some congruence class modulo $M$, for a fixed $M$. (When $M$ can grow with the input, these problems become NP-hard as one can encode a global cardinality constraint~\cite{Creignou:2010:NBC:1805950.1805954}.) Our goal is to classify the cases which are polynomial time solvable, as a function of $M$. In order to better appreciate the difficulty of this endeavor, we encourage the reader to not peek ahead and write down a guess for each of the cases listed in the table below. (Note that $\twosatmod_M$ refers to $\twosat$ with the global constraint modulo $M$, etc.)

\begin{center}
  \begin{tabular}{r|r|r|r|r|r}
    Name & Constraint types & $M = 3$ & $M = 4$ & $M = 6$ & $M=15$\\\hline
    $\twosatmod_M$ & $x \vee y$; $x = \neg y$ & & & & \\\hline
    $\hornsatmod_M$ & $x \wedge y \to z$ & & & &  \\\hline
    $\linmod_M$ & $x \oplus y \oplus z = 0\text{ or }1$ & & & &
  \end{tabular}
\end{center}

Our work hinges on several connections which makes our investigation interesting beyond the specific application to modular CSPs. The complexity of the problems (except $\twosatmod_M$) are tied to the parameters of certain polynomial representations --- lower bounds on such representations yields our algorithmic guarantees, and at the same the existence of efficient representations implies  hardness results. Thus the work illustrates an interesting duality between algorithms and hardness as originating from the same object. The relevant complexity measures for polynomial representations are novel and deserve further study. As the particular choice of complexity measure for each problem seems closely linked to the underlying polymorphisms of the CSP, we hope that initiating such a study could help bring together computational complexity theorists and specialists in the algebraic theory of CSPs.

The result for linear equations has interesting connections to coding theory, namely the extremal dimension of binary linear codes whose codewords have modular restrictions on their weights (which relates to concepts like doubly even codes that have been studied in coding theory), as well as to locally decodable codes, via relationship between polynomial representations and matching vector families and the Polynomial Freiman-Ruzsa (PFR) conjecture in additive combinatorics.

\subsection{Our Results}

We resolve the complexity of $\twosatmod_M$, $\hornsatmod_M$, and $\linmod_M$, namely whether it belongs to or is unlikely to be polynomial time solvable, for all moduli.\footnote{Unfortunately, one uncovered case is $\linmod_M$ for moduli $M=2^\ell p^s$ for an odd prime $p$ and $\ell \ge 3$.}
We will denote by $\twosatmod_M(n)$ an instance of $\twosat$ with $n$ variables and a global constraint modulo $M$, $\hornsatmod_M(n)$ and $\linmod_M(n)$ are similarly defined.

\begin{theorem}[Informal statement of main results]
	\label{thm:main-intro}
Suppose we have a single global modular constraint with a fixed modulus $M$.
	\begin{enumerate}
		\item $\twosatmod_M(n)$ is polytime solvable for all moduli $M$.
		\item $\hornsatmod_M(n)$ is polytime solvable when the modulus $M$ equals a prime power, and assuming $\ETH$, cannot be solved in $n^{o((\log n/\log \log n)^{r-1})}$ time when $M$ has $r$ distinct prime factors.
		\item $\linmod_M(n)$ is randomized polytime solvable when the modulus $M=2^\ell p^s$ for some odd prime $p$ and $\ell \le 2$, and assuming $\ETH$, cannot be solved in $n^{o((\log n/\log \log n)^{r-1})}$ when $M$ has $r$ distinct \emph{odd} prime factors. Further, assuming the Polynomial Freiman-Ruzsa conjecture, there is an algorithm with runtime $\exp(O_M(n/\log n))$ for all $M$.\footnote{We write $A=O_M(B)$ or $A\lesssim_M B$ to say that $A\le C_M B$ for some constant $C_M>0$ which depends only on $M$ and not on any other parameters. $A=\Omega_M(B), A\gtrsim_M B, A=o_M(B), A=\omega_M(B)$ are similarly defined.}
\end{enumerate}

\noindent
The polynomial runtime grows like $n^{O(M)}$, and assuming $\ETH$ this cannot be improved to $n^{o(M/\log M)}$ for the case of $\hornsatmod_M$ and $\linmod_M$. Moreover, for the $\linmod_M$ problem, the randomized polytime algorithms have quasipolynomial derandomizations.

\end{theorem}

For the one uncovered case of $\linmod_M(n)$, when $M=2^\ell p^s$ for an odd prime $p$ and $\ell > 2$, we give a quasi-polynomial time algorithm with runtime $\exp(O_M((\log n)^{2^{\ell-1}-1}))$. We are not sure if there should be a polynomial time algorithm also in this case.

The following table summarizes our results. We will assume $M\le cn$ for some sufficiently small constant $c>0$. All the hardness results are assuming $\ETH.$ For simplicity, the results for $\linmod_M$ are only stated for odd $M$ and the running times are for randomized algorithms, the algorithm for general $M$ is assuming PFR conjecture.
\renewcommand{\arraystretch}{1.5}
\begin{center}
\label{tab:MainResults}
\begin{tabular}{|c|c|c|c|c|}
\hline
\multirow{2}{*}{}&\multicolumn{2}{c|}{$M$ is a prime power}
& \multicolumn{2}{c|}{$M$ has $r$ distinct prime factors ($r\ge 2$)}\\
\cline{2-5}
&Algorithm & Hardness & Algorithm & Hardness\\
\hline
$\twosatmod_M(n)$&$n^{M+O(1)}$&-&$n^{M+O(1)}$&-\\
\hline
$\hornsatmod_M(n)$&$n^{M+O(1)}$&$n^{\Omega\left(\frac{M}{\log M}\right)}$&-&$\exp\left(\Omega_M\left(\left(\frac{\log n}{\log\log n}\right)^r\right)\right)$\\
\hline
$\linmod_M(n)\footnotemark$&$n^{M+O(1)}$&$n^{\Omega\left(\frac{M}{\log M}\right)}$&$\exp\left(O_M\left(\frac{n}{\log n}\right)\right)$&$\exp\left(\Omega_M\left(\left(\frac{\log n}{\log\log n}\right)^r\right)\right)$\\
\hline
\end{tabular}
\end{center}
\renewcommand{\arraystretch}{1}
\footnotetext{For $\linmod_M$, we are only stating the results for odd $M$ and the algorithms are randomized. The running time of the algorithm for non-prime-power $M$ is conditioned on Polynomial Freiman-Ruzsa conjecture.}

\smallskip\noindent\textbf{Extensions to multiple modular constraints.} We also consider natural extensions of the three problems, where we allow a small number of mod $M$ constraints or the more general version where we allow $\ell$ constraints with different moduli $M_1,M_2,\dots,M_\ell$. Our algorithm our $\twosatmod_M$ is presented in this more general model. For $\hornsat$ and $\lintwo$, we show that these general versions can be reduced to the basic version with a single mod $M$ constraint, without increasing the size of the instance too much. For example we can reduce an instance of $\hornsat$ over $n$ variables with $\ell$ constraints modulo $M$ to an instance of $\hornsatmod_M$ over $n^{O_M(\ell)}$ variables. Note that once $\ell$ becomes linear in $n$, these problems are \textsf{NP}-hard, so an exponential dependence in $\ell$ is necessary. A similar statement also holds for $\lintwo$ with multiple modular constraints.

\smallskip \noindent \textbf{Completing a classification of Boolean Mod-CSPs.}  Although $\hornsat$, $\lintwo$, and $\twosat$ are the most important CSPs that need to be analyzed, some additional care is needed to extend these classifications to all Boolean CSPs. This work is done in Appendix~\ref{app:dicot}, where we classify the computational complexity. The main observation is to use a classification by Post~\cite{post} of the polymorphisms of Boolean CSPs. Using this classification, we can show that when $M$ is a prime power, the Mod-CSP problem is always either in $\mathsf{RP}$ or is $\mathsf{NP}$-complete. On the other hand, when $M$ is divisible by distinct primes, the classification becomes a bit more difficult. The main difficulty is that we have an upper bound for $\twosatmod_M$, but lower bounds on $\hornsatmod_M$ and $\linmod_M$ (when $M$ is divisible by distinct odd primes). Because of the confluence of lower and upper bounds, there are some additional tractable cases that show up, which are similar in structure to $\twosat$. See Section~\ref{subsec:noprime} for more details.

\subsection{Our Techniques and Connections}

We now give brief overviews of our approach to establish Theorem~\ref{thm:main-intro}. We discuss each of the three constraint types in turn.

\subsubsection{$\twosatmod_M$ and recursive methods}
For $\twosatmod_M$, our algorithm is recursive. The one key idea is that we work with a more general form of modular constraint to make the recursion work, one that allows the Hamming weight to belong to a subset $S$ of congruence classes (rather than a single value). Standard methods, with some care to update the modular constraint, allow us to reduce to the case when the ``implication graph" on the literals of the $\twosat$ instance is a DAG. We then select a literal $y$ with no outgoing edge, and first set it to $1$ (which doesn't impact any other literal), and solve the $\twosat$ instance on the remaining variables with an updated modular constraint (that takes into account the setting of $y$).
If this succeeds, we can output this assignment and be done.

 Otherwise we set $y=0$, which forces all literals which have a path to $y$ in the DAG also to $0$. We can update the modular constraint accordingly, but note that in the end we are allowed to flip $y$ to $1$ and the $\twosat$ instance will still be satisfied. While there is no need to do this for normal $\twosat$, this flip might allow us to satisfy the modular constraint. As a result, we allow the set $S$ of congruence classes in the recursive call to also include this possibility. This is the reason why we need to work with the more general form of modular constraint. To implement this idea to run in polynomial time is a bit subtle, as naively we could reduce an instance with $n$ variables to two instances with $n-O(1)$ variables leading to exponential runtime. To avoid this pitfall, we track the size of the allowed moduli $S$, and argue that if it doesn't increase in the second recursive call (one where we set $y=0$), we can truncate that call and return no solution for that branch. This is justified because any valid solution with $y=0$ remains valid when $y=1$, and the former is already ruled out in the first call. The increase in $|S|$ in one of the recursive calls implies a polynomially bounded solution to the recurrence for the runtime, with exponent at most $M-1$.

\subsubsection{$\hornsatmod_M$ and covering number of $\NAND \mod M$ over $\bits$ basis} Our algorithm and analysis for $\hornsatmod_M$, i.e., $\hornsat$ with a single global linear constraint modulo $M$, is very different from the $\twosat$ case. An important property of the $\hornsat$ instance is that the set of solutions in intersection-closed. Given an instance $\Psi$ of $\hornsat$ and a subset $S$ of its variables, one can efficiently find the minimal solution among all solutions of $\Psi$ which set the variables of $S$ to 1, this is called the $\FindMinimal(\Psi,S)$ routine. Now we run $\FindMinimal(\Psi,S)$ on all subsets $S$ upto a give size $R$ and check if any of the outputs satisfy the modular constraint. If none of them satisfy the modular constraint, we claim that $\Psi$ has no solution which satisfies the modular constraint. The running time of this $R$-round algorithm is $n^{R+O(1)}$. 
If the $R$-round algorithm for $\hornsatmod_M$ fails, then we show that it is because of a special kind of obstruction. To describe these obstructions, we will need a few definitions.

\begin{definition}
A multilinear polynomial $p(x_1,\dots,x_d)$ is said to represent $\NAND_d\mod M$ over $\bits^d$ if it has integer coefficients and
\[
p(x)
\begin{cases}
=& 0 \mod M \text{ if } x=\allones\\
\neq& 0 \mod M \text{ if }x\in \bits^d\setminus\{\allones\}
\end{cases}
\]
where $\allones$ is the all ones vector.
\end{definition}

\begin{definition}[Covering number]
The covering number of a multilinear polynomial $p(x)$, denoted by $\cov(p)$, is the minimum number of monomials of $p$ one can choose such that every variable that appears in $p$ appears in one of them.
\end{definition}
Note that the covering number is the minimum set cover of the family of subsets of variables given by the monomials.
An obstruction for the $R$-round algorithm to solve $\hornsatmod_M(n)$ correctly is a polynomial $p$ which represents $\NAND \mod M$ with at most $n+1$ monomials and $\cov(p)> R$. Therefore we have the following proposition.
\begin{proposition}
If every polynomial $p$ which represents $\NAND \mod M$ over $\bits$ basis with $n+1$ monomials has $\cov(p)\le R(n)$, then there exists an $n^{R+O(1)}$-time algorithm for $\hornsatmod_M(n)$.
\end{proposition}

\smallskip\noindent\textbf{When $M$ is a prime power.} In this case, we can show that any polynomial which represents $\NAND \mod M$ should have covering number at most $M-1$. Note that this bound is independent of the number of monomials in the polynomial. This implies that our algorithm with $R=M-1$ rounds solves $\hornsatmod_M(n)$ correctly.

\smallskip\noindent\textbf{When $M$ has $r\ge 2$ distinct prime factors.} In this case, it turns out that there are obstructions for any constant round algorithm. More precisely, there are polynomials which represent $\NAND \mod M$ with $n$ monomials, but their covering number is $\Omega\left((\log n/\log \log n)^{r-1}\right)$. Such polynomials can be obtained from polynomials of degree $O(d^{1/r})$ which represent $\NAND_d \mod M$ constructed by Barrington, Beigel and Rudich~\cite{BarringtonBR94}. This implies that the $R$-round algorithm will not work if we choose $R=o\left((\log n/\log \log n)^{r-1}\right)$. This by itself does not show hardness of $\hornsatmod_M$, it just shows that our algorithm doesn't work with constant rounds. But it turns out that we can use low-degree polynomial representations of $\NAND \mod M$ directly as a gadget to reduce $\threesat$ to $\hornsatmod_M$ without blowing up the size too much. More precisely, we show the following hardness result.
\begin{proposition}
Suppose there exists a polynomial which represents $\NAND_d \mod M$ over $\bits^d$ with degree $\Delta$. Assuming $\ETH$, $\hornsatmod_M(n)$ requires $2^
{\Omega(d)}$ time for some $n=(ed/\Delta)^{O(\Delta)}$.
\end{proposition}
Using the upper bound $\Delta=O(d^{1/r})$ from~\cite{BarringtonBR94}, we get $\exp\left(\Omega_M((\log n/\log\log n)^{r})\right)$ time $\ETH$-hardness for $\hornsatmod_M(n)$. Better upper bounds on $\Delta$ lead to better hardness results. The best known lower bound on $\Delta$ is $\Omega_M\left((\log d)^{1/(r-1)}\right)$ due to Barrington and Tardos~\cite{BarringtonT98}. If there is a polynomial whose degree matches this lower bound, then assuming $\ETH$, we can get $\bigexp{\bigexp{(\log n)^{1-1/r}}}$ hardness for $\hornsatmod_M(n)$.

What about sub-exponential time algorithms? We conjecture that the covering number of any polynomial which represents $\NAND \mod M$ with $n$ monomials should be $n^{o_M(1)}$ for any fixed $M$. If true, this would imply an $\exp(n^{o_M(1)})$-time algorithm for $\hornsatmod_M(n)$. We give evidence towards our conjecture by showing that the fractional covering number (which is an LP relaxation of covering number) of any such polynomial is indeed $n^{o_M(1)}$.

\subsubsection{$\linmod_M$ and sparsity of the $\OR \mod M$ over $\sbits$ basis} The algorithm and analysis for $\linmod_M$ has strong similarities to the analysis for $\hornsatmod_M$. The algorithm is quite simple to state: sample $T$ uniformly random solutions to the linear system over $\F_2$ and check if any of them satisfy the modular constraint.

To analyze how many samples $T$ to check so that the algorithm is correct with high probability, we need to prove that for any $d$-dimensional affine subspace of $\F_2^n$, there are either $0$ points in this space satisfying the modular constraint or there are at least $2^d/f(n, M)$ such points. Then $T=O(f(n,M))$ samples would suffice. By a simple reduction, it suffices to bound the maximal dimension $\cD(n, M)$ such that a $\cD$-dimensional affine subspace of $\F_2^n$ has \emph{exactly one} element whose Hamming weight is $a \mod M$ for some $a$. Quantitatively, we show that $f(n, M) \le O(2^{\cD(n, m)})$. In other words, the obstructions for our algorithm are large affine subspaces which have exactly one point which satisfies a linear constraint modulo $M$. When $M$ is a power of $2$, we prove that $\cD(n,M)\le M-1$. For $M=2^\ell M'$ for some odd $M'\ge 3$, we can get upper bounds on $\cD(n,M)$ from upper bounds on $\cD(n,M')$. So we can only focus on odd $M$. The obstructions for odd $M$ can be represented using polynomials like we did in the $\hornsatmod_M$ case.

\begin{definition}
A polynomial $p(x_1,\dots,x_d)$ is said to represent $\OR_d\mod M$ over $\sbits^d$ if it has integer coefficients and
\[
p(x)
\begin{cases}
=& 0 \mod M \text{ if } x=\allones\\
\neq& 0 \mod M \text{ if }x\in \sbits^d\setminus\{\allones\}
\end{cases}
\]
where $\allones$ is the all ones vector.
\end{definition}
The existence of a $d$-dimensional affine subspace of $\F_2^n$ with exactly one point satisfying the mod $M$ constraint is equivalent to an $(n+1)$-sparse polynomial representation of $\OR_d\mod M$ over $\sbits^d$. Thus, lower bounds on the sparsity of representations of $\OR \mod M$ in the $\sbits$ basis imply upper bounds on $\cD(n, M)$ and the runtime of our randomized algorithm. More precisely, we have the following proposition.
\begin{proposition}
Let $M$ be odd, if there is no polynomial which represents $\OR_d \mod M$ over $\sbits^d$ with at most $n+1$ monomials, then there exists a $O(2^dn^{O(1)})$ time randomized algorithm for $\linmod_M(n)$.
\end{proposition}

\smallskip\noindent\textbf{Connections to circuit complexity and locally decodable codes.} Note that an $s$-sparse polynomial which represents $\OR_d \mod M$ over $\sbits^d$ is equivalent to a two layer $\MOD_M \circ \MOD_2$ circuit which computes $\OR$ over $\bits^d$ with $s$ $\MOD_2$-gates at the bottom level connected to a $\MOD_M$ gate at the top.\footnote{Here a $\MOD_M$ gate outputs $0$ if the number of $1$-inputs is a multiple of $M$ and outputs $1$ otherwise.} Circuits made up of such $\MOD$ gates are extensively studied in circuit lower bounds~\cite{ChattopadhyayGPT06,ChattopadhyayW09}. An $s$-sparse polynomial representation of $\OR_d \mod M$ over $\sbits^d$ can be used to construct a \emph{matching vector family} (MVF) of size $2^d$ over $(\Z/M\Z)^s$. MVFs of large size and small dimension can be used to construct good \emph{locally decodable codes} (LDCs) and private information retrieval (PIR) schemes~\cite{Yekhanin08, Efremenko12, DvirGY11, DvirG16}. So lower bounds for LDCs or upper bounds on the size of MVFs imply lower bounds on the sparsity of polynomials which represent $\OR_d \mod M$ over $\sbits^d$. Concretely, superpolynomial lower bounds for the length of constant-query LDCs would imply $2^{o(n)}$ time algorithms for $\linmod_M(n)$ when $M$ is any constant; we find this connection quite surprising!

\smallskip\noindent\textbf{When $M$ is an odd prime power.}
In this case, we can show that the sparsity of a polynomial which represents $\OR_d \mod M$ over $\sbits^d$ is at least $2^{d/(M-1)}$, which is nearly tight. In other words, $\cD(n,M)\le (M-1)\log n$. This immediately implies a $n^{M+O(1)}$-time randomized algorithm for $\linmod_M$.

\smallskip\noindent\textbf{When $M$ has $r\ge 2$ odd prime factors.}
The results of~\cite{BarringtonBR94} imply that there exist degree $O(d^{1/r})$ polynomials which represent $\OR_d \mod M$ over $\sbits^d$, therefore they have at most $n=\exp\left(O(d^{1/r}\log d)\right)$ monomials. In other words, $\cD(n,M)\gtrsim (\log n/\log\log n)^r$. Thus there are obstructions for our algorithm to run in polynomial time. This by itself does not show hardness. But, similar to the $\hornsatmod_M$ case, we can directly use low-degree polynomials which represent $\OR_d \mod M$ over $\sbits^d$ to reduce $\threesat$ to $\linmod_M$ without increasing the size too much.
\begin{proposition}
Suppose there exists a polynomial of degree $\Delta$ which represents $\OR_d \mod M$ over $\sbits^d$. Assuming $\ETH$, $\linmod_M(n)$ takes $2^{\Omega(d)}$ time for some $n=(ed/\Delta)^{O(\Delta)}$.
\end{proposition}

Using the upper bound $\Delta=O(d^{1/r})$ from~\cite{BarringtonBR94}, we get $\exp\left(\Omega_M((\log n/\log\log n)^{r})\right)$ time $\ETH$-hardness for $\linmod_M(n)$. Better upper bounds on $\Delta$ lead to better hardness results. The best known lower bound on $\Delta$ is $\Omega_M\left((\log d)^{1/(r-1)}\right)$ due to Barrington and Tardos~\cite{BarringtonT98}. If there is a polynomial whose degree matches this lower bound, then assuming $\ETH$, we can get $\bigexp{\bigexp{(\log n)^{1-1/r}}}$-time hardness for $\linmod_M(n)$.

What about subexponential time algorithms? Unfortunately, we do not know any unconditional superlinear (i.e. $\omega(d)$) lower bounds on the sparsity of polynomials which represent $\OR_d \mod M$ over $\sbits^d$. Such a lower bound would imply $2^{o(n)}$-time algorithms for $\linmod_M(n)$ for any constant $M$. But using the connection to MVFs, and an upper bound on the size of MVFS due to~\cite{BhowmickDL14} assuming Polynomial Freiman-Ruzsa (PFR) conjecture from additive combinatorics, we can conclude that the sparsity of a polynomial representing $\OR_d\mod M$ over $\sbits^d$ is at least $\Omega_M(d\log d)$ under the same conjecture. This implies an $\exp(O_M(n/\log n))$-time algorithm for $\linmod_M(n)$ assuming the PFR conjecture.

\subsubsection{Connections to submodular minimization and integer programming.}\label{subsec:intro-IP}

A classic result in combinatorial optimization is that integer linear programs which have a \emph{totally unimodular} constraint matrix, i.e., every square submatrix has determinant in $\{-1, 0, 1\}$, can be solved in polynomial time. This is because the vertices of the feasible polytope of such a constraint matrix are integral. Recently this result has been extended to \emph{totally bimodular} constraint matrices, where every square submatrix has determinant in $\{-2, -1, 0, 1, 2\}$~\cite{ArtmannWZ17}. Such results have inspired a conjecture that any integer program for which every square submatrix has determinant bounded in absolute value by $M$, can be solved in $n^{O_M(1)}$ time.

A recent paper by N\"agele, Sudakov and Zenklusen~\cite{NageleSZ18} tries to lay the groundwork for proving this conjecture by considering a special case. This special case is finding the minimum cut in a directed graph such that the number of vertices on one side of the cut satisfies a modular condition modulo $M$. As cut functions of directed graphs are submodular, they generalized this question to the following algorithmic problem, denoted by $\submod_M(n)$: Given $a,M$, minimize a submodular function $f:\bits^n \to \Z$ (given oracle access) over all $x$ such that $\sum_{i=1}^n x_i =a \mod M$.\footnote{A function $f:\bits^n \to \Z$ is called submodular if for all $x, y \in \{0, 1\}^n$, $f(x) + f(y) \ge f(x \vee y) + f(x \wedge y)$, where $\vee$ and $\wedge$ are bitwise OR and AND.}

 Our algorithm for $\hornsatmod_M$ is actually inspired by the algorithm from~\cite{NageleSZ18} for $\submod_M(n)$. They show that the only obstructions to their $R$-round algorithm to work correctly are certain set families they called $(M,R,d)$-systems. They then showed that when $M$ is a prime power, $(M,M-1,d)$-systems do not exist for any $d$. This implies an $n^{O(M)}$-time algorithm for $\submod_M$ when $M$ is a prime power. They asked if a polynomial time algorithm exists for $\submod_M$ when $M$ has multiple prime factors.

We observe that $(M,R,d)$-systems are equivalent to polynomial representations of $\NAND \mod M$ over $\bits$ basis with $d$ monomials and covering number greater than $R$. Thus the obstructions for their algorithm for $\submod_M$ and our algorithm for $\hornsatmod_M$ are exactly the same!
Thus we can give a simpler proof of their result for prime power $M$, and also throw new light on $\submod_M$ when $M$ has multiple prime factors. In particular, if $M$ has $r$ distinct prime factors, this implies that there are obstructions for their $R$-round algorithm to work for $\submod_M(n)$ for $R=o_M\left((\log n/\log \log n)^{r-1}\right)$. This answers an open question from their paper and explains why they couldn't extend their algorithm for any constant $M$. If true, our conjecture that the covering number of a polynomial which represents $\NAND \mod M$ with $n$ monomials is $n^{o_M(1)}$, implies an $\exp(n^{o_M(1)})$ time algorithm for $\submod_M$ for any $M$. We also make a conjecture about the existence of certain submodular functions, which would allow us to prove superpolynomial $\ETH$-hardness results for $\submod_M$ when $M$ has multiple prime factors. It would be interesting if our methods have any implications for hardness of solving integer linear programs with bounded-minor constraint matrices, our results suggest that something different can happen at $M=6$.

\subsection{Future directions}

Given that our work is the first to study the effect of global modular constraints on the tractability of CSPs, and our results unearth a rich picture rife with interesting connections to many central topics such as algebraic complexity measures of Boolean functions, coding theory, and combinatorial optimization, there are naturally many questions and directions for future work. We list a few below.

\begin{itemize}
\itemsep=0ex
	\item Our work raises some intriguing new questions about polynomial representations modulo a  composite number $M$. Can one prove an unconditional non-trivial (i.e, $\omega(d)$) lower bound on the sparsity of polynomial representing $\mathsf{OR}_d \mod M$ over the $\{1,-1\}$-basis? Can one construct polynomial representations that are sparser than what is obtainable by simply appealing to the best known low-degree representations?

For the covering complexity measure for polynomials representing $\mathsf{NAND}\mod M$ over the $\{0,1\}$ basis with $n$ monomials, can we prove an upper bound of $n^{o_M(1)}$ (Conjecture~\ref{conj:covering_number})? This would imply a sub-exponential time algorithm for $\hornsatmod_M$ and $\submod_M$ via our connections.
\item The analysis of our algorithms for $\linmod_M$ and
  $\hornsatmod_M$ use the sparsity and covering complexity of
  appropriate polynomial representations. On the other hand, our
  hardness results for these problems use the degree measure of these
  polynomials. Can one prove that sparsity and covering complexity
  directly dictate hardness as well? This would complete a very
  pleasing picture by giving matching algorithmic and hardness results and tying the complexity to some complexity parameter associated with polynomial representations.
	\item Can one show hardness of submodular minimization with a global mod $6$ constraint? Constructing certain submodular functions as described in Conjecture~\ref{conj:submodular} would imply such a hardness result. 
What about hardness of solving integer linear programs with bounded minors?
	\item Of the major cases for our classification of Boolean CSPs, $\twosatmod_M$ is special in that its complexity is not tied to the complexity of certain polynomial representations.
Given that the polynomials representations which play a role for $\hornsatmod_M$ and $\linmod_M$ are very closely related to the polymorphisms of $\hornsat$ and $\lintwo$, $\AND$ and $\mathsf{XOR}$ respectively, how does the $\mathsf{MAJ}$ polymorphisms of $\twosat$ play a role in the tractability of $\twosatmod_M$? Understanding this would give us clues about how to generalize our results to CSPs over larger domains.
        \item  Furthermore, $\twosatmod_M$ is also the only case in which we lack a lower bound of the form $\Omega(n^{f(M)})$, for some nontrivial function of $f$.  Could it be that $\twosatmod_M$ is fixed-parameter tractable? For instance, does there exist a $2^M n^{O(1)}$ algorithm? We note that the dependence on $M$ cannot be polynomial (unless $\mathsf{P} = \mathsf{NP}$), as setting $M$ greater than $n$ would solve $\twosat$ with a global cardinality constraint, which is NP-complete.
	\item How do Mod-CSPs behave on a non-Boolean domain? Even for a domain of size three, the classification of ordinary CSPs is much more complex~\cite{DBLP:journals/jacm/Bulatov06}. As such, in order for such a program to be carried out, one needs to better understand the interplay between the global modular constraints and the polymorphisms of these CSPs. In particular, how do notions like cores, bounded width, identities, etc., interplay with the modular constraints? See \cite{DBLP:conf/dagstuhl/BartoKW17} for definitions of these terms.
	\item What other interesting global constraints can we impose on CSPs like $\twosat,\hornsat,\lintwo$, while still keeping them tractable? What happens if we add a global constraint over a non-abelian group i.e. a global constraint of the form $\prod_{i=1}^n g_i^{x_i}=g_0$ for some $g_0,g_1,\dots,g_n\in G$ where $G$ is a non-abelian group.
\end{itemize}

\subsection{Organization}

In Section~\ref{sec:prelim}, we formally define what a Mod-CSP is as well as state some basic facts about polynomials. In Section~\ref{sec:hornsat}, we study $\hornsatmod_M$. In Section~\ref{sec:lineq}, we study $\linmod_M$. In Section~\ref{sec:2sat}, we prove our algorithmic results for $\twosatmod_M$. In Section~\ref{app:dicot}, we show how the results for these individual problems imply some dichotomy-type results for general Boolean Mod-CSPs.

\section{Preliminaries}\label{sec:prelim}
In this section, we give standard definitions of (local) CSPs as well as a formal definition of Mod-CSPs.

\subsection{CSP basics}

The type of CSP we study depends on a few parameters. The number of variables is often denoted by $n$ and the variables are written as $x_1, \hdots, x_n$. The number of constraints is denote by $m$. Although the number of constraints in an instance tends to $\infty$ and $n \to \infty$, there are only finitely many types of constraints, which are specified by the \emph{template} $\Gamma$. To specify the structure of the template, we have a \emph{signature} $\sigma = (I, \ar)$, where $I$ is an index set of the constraints, and $\ar : I \to \mathbb N$ gives the \emph{arity} of each constraint, how many variables it takes as arguments. We formally define the template as follows.

\begin{definition}
  A \emph{template} with signature $\sigma$ over \emph{domain} $D$ is an ordered tuple $\Gamma = (C_i \subset D^{\ar_i} \mid i \in I).$
\end{definition}

For this template, we say that a CSP over $\Gamma$ is a formula $\Psi$ over variables $x_1, \hdots, x_n$ which is a CNF of constraints from $\Gamma$. That is,

\[
  \Psi(x_1, \hdots, x_n) = \bigwedge_{i = 1}^m C_{f(i)}(x_{j_{i,1}}, \hdots, x_{j_{i,\ar_{f(i)}}})
\]
for some $f:[m]\to I$. Note that variables need not be distinct.
 From this, we can define the corresponding computational problem.

\begin{definition}
  The decision problem $\CSP(\Gamma)$ asks, given an instance $\Psi$ does there exist a solution?
\end{definition}

The following are a few concrete examples of tractable $\Gamma$ in the Boolean domain ($D = \{0, 1\}$).

\begin{itemize}
\item $\twosat$ can be encoded by $\Gamma_{\twosat} = \{C \subset \bits^2 : |C| = 3\}$. There are several efficient algorithms for 2-SAT (e.g., \cite{krom1967decision, DBLP:conf/dagstuhl/BartoKW17}).

\item $\lintwo$ can be encoded by $\Gamma_{\lintwo} = \{\{(x, y, z) : x \oplus y \oplus z = a\} : a \in \bits\}$, where $\oplus$ is addition modulo $2$. $\lintwo$ can be solved in polynomial time using Gaussian elimination.

\item $\hornsat$ is encoded by $\Gamma_{\hornsat} = \{\{(x, y, z) : x \wedge y \rightarrow z\}, \{(0)\}, \{(1)\}\}$. $\hornsat$ can be solved in polynomial time using a deduce-and-propagate style of algorithm.

  Despite the apparent simplicity of the constraint, $\Gamma_{\hornsat}$ is $\mathsf{P}$-complete~\cite{cook2010logical}.

\item $\dualhornsat$ is encoded by $\Gamma_{\dualhornsat} = \{\{(x, y, z) : x \rightarrow y \vee z\}, \{(0)\}, \{(1)\}\}$. Note that $\Psi(x)$ is a $\dualhornsat$ formula if and only if $\Psi'(x) = \Psi(\neg x)$ is a $\hornsat$ formula.
\end{itemize}

As previously mentioned, the computational complexity of any such $\Gamma$ has been classified by the algebraic dichotomy theorem \cite{DBLP:conf/focs/Bulatov17,DBLP:conf/focs/Zhuk17}. In the Boolean case, this is known as Schaefer's theorem.

\begin{theorem}[\cite{Schaefer:1978}]
  Let $\Gamma$ be a Boolean CSP template. Then, either $\CSP(\Gamma)$ is $\mathsf{NP}$-complete or one of the following tractable cases.  In the last four cases below, $\Psi(x)$ is true if and only if $\Psi'(x, y)$ is true some $y$ and $\Psi'$ can be constructed in polynomial time.
  \begin{enumerate}
  \item For all $\Psi(x) \in \CSP(\Gamma)$, $x = (0, \hdots, 0)$ is a solution.
  \item For all $\Psi(x) \in \CSP(\Gamma)$, $x = (1, \hdots, 1)$ is a solution.
  \item For all $\Psi(x) \in \CSP(\Gamma)$, there is a formula $\Psi'(x, y) \in \CSP(\Gamma_{\twosat})$.
  \item For all $\Psi(x) \in \CSP(\Gamma)$, there is a formula $\Psi'(x, y) \in \CSP(\Gamma_{\lintwo})$.
  \item For all $\Psi(x) \in \CSP(\Gamma)$, there is a formula $\Psi'(x, y) \in \CSP(\Gamma_{\hornsat})$.
  \item For all $\Psi(x) \in \CSP(\Gamma)$, there is a formula $\Psi'(x, y) \in \CSP(\Gamma_{\dualhornsat})$
  \end{enumerate}
 
\end{theorem}

Note that the first two cases are rather superficial and case 6 ($\dualhornsat$) is equivalent to $\hornsat$ up to a global negation of the variables. As such, $\twosat$, $\lintwo$, and $\hornsat$ will be the focus of our study.

\subsection{Mod-CSPs}

With CSPs formally defined, it is now easy to state what a Mod-CSP is. To assist the reader, we give multiple definitions of a Mod-CSP, each with an increasing level of generality.

\begin{definition}
  Let $\Gamma$ be a template over the domain $\{0, 1\}$. Let $M \ge 2$ be an integer. An instance of $\modcsp(\Gamma, M)$ consists of an instance $\Psi(x_1, \hdots, x_n)$ of $\CSP(\Gamma)$ with an additional constraint,
  \[
    x_1 + x_2 + \cdots + x_n = a \mod M
  \]
  for some integer $a$.
\end{definition}

The only structure of arithmetic modulo $M$ we are using is that $\mathbb Z/M\mathbb Z$ (the integers modulo $M$) is an abelian group and that $\{0, 1\}$ has a ``natural'' inclusion in $\mathbb Z/M\mathbb Z$. We generalize both of these aspects in the next definition.

\begin{definition}
  Let $\Gamma$ be a template over a domain $D$. Let $(G, +, 0)$ be an abelian group. An instance of $\modcsp(\Gamma, G)$ consists of an instance $\Psi(x_1, \hdots, x_n)$ of $\CSP(\Gamma)$ and an additional constraint
  \[
    g_1(x_1) + \cdots + g_n(x_n) = a
  \]
  for some $a \in G$ and arbitrary maps $g_1, \hdots, g_n : D \to G$.
\end{definition}

\begin{remark}
  When $\Gamma$ is a Boolean template, $\modcsp(\Gamma, M)$ and $\modcsp(\Gamma, \mathbb Z/M\mathbb Z)$ are subtlely different because $\modcsp(\Gamma, \mathbb Z/M\mathbb Z)$ allows for weighted instances. Even so, these problems are polynomial-time equivalent as long as we make the mild, standard (e.g., \cite{DBLP:conf/dagstuhl/BartoKW17}) assumption that we allow equality of variables because we can simulate the term $i_j(x_j)$ by having $i_j(1) - i_j(0)$ copies of $x_j$ all set equal to each other and subtracting $i_j(0)$ from $a$ to account for the additive shift. Equality can be built as a gadget in $\twosat$, $\lintwo$, and $\hornsat$, so this assumption does not hurt our algorithmic or hardness results for those problems.
\end{remark}

In this paper, $G$ is always finite, and unless otherwise specified the cardinality of $G$ is independent of the number of variables in the instance. By the classification of finite abelian groups, $G$ is a subgroup of $(\mathbb Z/M\mathbb Z)^k$ for some positive integer $M$ and $k$. Therefore $\modcsp(\Gamma, G)$ can then be written as an instance of $\CSP(\Gamma)$ along with a system of $k$ weighted equations modulo $M$. This is the perspective used in Sections \ref{sec:hornsat} and \ref{sec:lineq}.

For completeness, we also define a ``list'' version of Mod-CSPs, whether the modular constraint need not be an equality but rather an inclusion in a set of elements. Note that the list version can be reduced to the equality version by a brute force search over every element of the list. This version will only be used in the 2-SAT analysis in Section~\ref{sec:2sat} and in the hardness analysis in Section~\ref{app:dicot}.

\begin{definition}
  Let $\Gamma$ be a template over domain $D$. Let $(G, +)$ be an abelian group and let $S \subset G$ be nonempty. An instance of $\modcsp(\Gamma, G, S)$ consists of an instance $\Psi$ of $\CSP(\Gamma)$ over variables $x_1, \hdots, x_n$ with a single additional constraint
  \[
    g_1(x_1) + \cdots + g_n(x_n) \in S,
  \]
  for some choice of maps $g_1, \hdots, g_n : D \to G$.
\end{definition}

\subsection{Lucas theorem and other polynomial preliminaries}
We will collect some useful lemmas about polynomial representations of modular constraints. The reader may skip these and come back when needed.
\begin{lemma}[Lucas's theorem]
\label{lem:LucasTheorem}
Let $p$ be a prime and let $a,b$ be non-negative integers with base $p$ expansions given by $a=a_0+a_1p+a_2p^2+\dots$ and $b=b_0+b_1p+b_2p^2+\dots$ for some $0\le a_i,b_i \le p-1$. Then
$$\binom{a}{b}=\prod_{i\ge 0} \binom{a_i}{b_i}\mod p$$ where $\binom{m}{n}=\frac{m(m-1)\dots (m-(n-1))}{n!}$ for $n>0$ and $\binom{m}{0}=1$ for all $m$.
\end{lemma}

\noindent Denote by $s_k(x)$ the $k^{th}$ elementary symmetric polynomial given by $s_k(x)=\sum_{i_1<i_2<\dots<i_k}x_{i_1}x_{i_2}\dots x_{i_k}$.

\begin{lemma}\label{lem:primepower_divisibility_Hamming}
Let $p$ be some prime, $\ell\ge 1$ and $0\le a\le p^\ell-1$ be integers. Let $a=a_0+a_1 p+ \dots +a_{\ell-1}p^\ell$ be the base $p$ expansion of $a$. Then for every $x\in \bits^n$,$$\Ham(x)=a\mod p^\ell \iff s_{p^t}(x)=a_i \mod p\ \forall\ 0\le t\le \ell-1.$$
\end{lemma}
\begin{proof}
Let $\Ham(x)=b_0+b_1 p +b_2 p^2+\dots$ be the base $p$ expansion of $\Ham(x)$. $\Ham(x)=a\mod p^\ell$ iff $a_t=b_t$ for all $0\le t \le \ell-1$.
By Lemma~\ref{lem:LucasTheorem},  $\binom{\Ham(x)}{p^t}=\binom{b_t}{1}=b_t \mod p$. Note that $\binom{\Ham(x)}{p^t} = s_{p^t}(x)$. Thus $\Ham(x)=a\mod p^\ell$ iff $s_{p^t}(x)=a_t \mod p$ for every $0\le t\le \ell-1$.
\end{proof}

\begin{lemma}\label{lem:primepower_divisibility_Hamming_2}
Let $p$ be some prime. For every $\ell\ge 1$ and $0\le a\le p^\ell-1$, there exists a degree $p^\ell-1$ polynomial $\phi_{\ell,a}\in \F_p[x_1,\dots,x_n]$ such that for every $x\in \bits^n$,
\[\phi_{\ell,a}(x) \mod p =
\begin{cases}
0& \text{ if } \Ham(x)=a\mod p^\ell\\
1& \text{ if } \Ham(x)\ne a\mod p^\ell.
\end{cases}
\]
\end{lemma}
\begin{proof}
Let $a=a_0+a_1 p+\dots+a_{\ell-1}p^{\ell-1}$ be the base $p$ expansion of $a$. By Lemma~\ref{lem:primepower_divisibility_Hamming} and Fermat's little theorem, $$\phi_
{\ell,a}(x)=1-\prod_{t=0}^{\ell-1}\left(1-(a_t-s_{p^t}(x))^{p-1}\right)$$ is the required polynomial of degree $(p-1)(1+p+\dots+p^{\ell-1})=p^\ell-1$.
\end{proof}

\begin{lemma}[DeMillo-Lipton-Schwartz-Zippel Lemma]\label{lem:SchwartzZippel}
Let $f\in \F_2[x_1,\dots,x_n]$ be a degree $r$ polynomial and let $x_0\in \F_2^n$ be some fixed point. Then $$\Pr_{x\in \F_2^n}[f(x)=f(x_0)]\ge \frac{1}{2^r}$$ where the probability is over a uniformly random point $x\in \F_2^n$.
\end{lemma}

\section{$\hornsatmod_M$}
\label{sec:hornsat}
Consider a Boolean CSP on variables $x_1, \hdots, x_n$ with three kinds of constraints: $x_i = 0$, $x_i = 1$, or $x_{i_1} \vee \neg x_{i_2} \vee \cdots \vee \neg x_{i_k} = 1$ (with only the first variable not negated)\footnote{This is equivalent to $x_{i_2}\wedge \dots \wedge x_{i_k}\rightarrow x_{i_1}$.}. Such a CSP is known as $\hornsat$. Fix an instance $\Psi$ of $\hornsat$ and let $\mathcal F \subset \{0, 1\}^n$ be the family of solutions. Think of each vector $v \in \{0, 1\}^n$ as the set $\{i : v_i = 1\}$, so $\mathcal F$ is viewed as a family of subsets of $[n]$. The important property of $\hornsat$ is that $\mathcal F$ is an intersecting family: $A, B \in \mathcal F$ implies that $A \cap B \in \mathcal F$. \footnote{This is equivalent to $\hornsat$ having $\AND_2$ as a polymorphism.} In essence, $\hornsat$ is a white-box model for studying intersecting families.

\begin{definition}
$\hornsatmod_M(n)$ is the following algorithmic problem: Given an instance of $\hornsat$ on $n$ variables $x_1,x_2,\dots,x_n$ along with a modular constraint $\sum_i a_i x_i = a_0 \mod M$, decide if there is a solution in $x\in \bits^n.$
\end{definition}

In this section, we will present an algorithm for $\hornsatmod_M$ and analyze its running time.
\subsection{Algorithm for $\hornsatmod_M$}

Fix an instance $\Psi$ of $\hornsat$ and let $\mathcal F \subset \{0, 1\}^n$ be the family of solutions.
The algorithm $\FindMinimal$ shows that given any $A \subset [n]$, we can efficiently find the unique minimal $B \in \mathcal F$ such that $B \supset A$. Note that the uniqueness of $B$ follows from the intersection-closed property of $\cF$.

\begin{algorithm}[h]
\caption{Algorithm $\FindMinimal$}
\label{alg:findminimal}
  \begin{itemize}
  \item Input: Instance $\Psi$ of $\hornsat$, $x \in \{0, 1\}^n$.
  \item Output: $y \in \{0, 1\}^n$ with $y\ge x$, $y$ satisfies $\Psi$, and $y$ minimal.
  \item $\FindMinimal(x)$
    \begin{enumerate}
    \item Set $y \ot x$.
    \item While $y$ does not satisfy $\Psi$,
      \begin{enumerate}
      \item Find a clause $C(y_{i_1}, \hdots y_{i_k})$ which fails.
      \item Find the minimal $z \ge y$ which satisfies $C$.
      \item If $z$ does not exist, return $\operatorname{NO-SOLUTION}$.
      \item Set $y \ot z$.
      \end{enumerate}
    \item Return $y$.
    \end{enumerate}
  \end{itemize}
\end{algorithm}
Note that $\FindMinimal$ will converge in $n$ rounds, because in each round the Hamming weight of $y$ will strictly increase.
We are now ready to present our algorithm for $\hornsatmod_M$. Our algorithm uses this $\FindMinimal$ routine, and is directly inspired by the algorithm in N\"agle, Sudakov, and Zenklusen~\cite{NageleSZ18} for submodular minimization with modular constraints. Wlog, we can assume that the coefficients $a_1,\dots,a_n=1$ in the modular constraint. This is because, we can create $a_i$ copies of variable $x_i$ with equality constraints among these copies. And equality can be implement using $\hornsat$ clauses as: $a=b$ iff $a\rightarrow b$ and $b\rightarrow a$. This will increase the variables by a factor of $M$. For a bit vector $x\in \bits^n$, $\Ham(x)$ denotes the Hamming weight of $x$ i.e. number of 1's in $x$.

\begin{algorithm}[H]
\caption{Algorithm for $\hornsatmod_M(n)$ with $R$ rounds}
\label{alg:horn_sat_modm}
  \begin{itemize}
  \item Input: instance $\Psi$ of $\hornsat$ with a global constraint $\sum x_i = r \mod M$.
  \item Output: either a solution $x \in \{0, 1\}^n$ or $\operatorname{NO-SOLUTION}$.
  \item Method:
    \begin{enumerate}
    \item For all $x \in \{0, 1\}^n$ with $\Ham(x) \le R$, test if $\FindMinimal(x)$ returns a satisfying assignment.
    \item Else, return $\operatorname{NO-SOLUTION}$.
    \end{enumerate}
  \end{itemize}
\end{algorithm}

We will relate the number of rounds needed in Algorithm~\ref{alg:horn_sat_modm} to certain combinatorial families called $(M,r,d)$-systems.

\begin{definition}[\cite{NageleSZ18}]
\label{def:Mrd_system}
A collection of subsets $\cF \subset 2^{[d]}$ is called an $(M,r,d)$-system if
\begin{enumerate}
\item $\cF$ is closed under intersections i.e. if $F,G\in \cF$ then $F\cap G\in \cF$.
\item For every $F\in \cF$, $|F|\neq d \mod M$.
\item For every subset $S\subset [d]$ of size at most $r$, there exists an $F\in \cF$ which contains it.
\end{enumerate}
\end{definition}

\begin{proposition}
\label{prop:horn_sat_modm_correctness}
  If Algorithm~\ref{alg:horn_sat_modm} with $R$ rounds fails on some instance of $\hornsatmod_M(n)$ then there exists an $(M,R,d)$-system  for some $d\le n$.
\end{proposition}
\begin{proof}
  Let $\Psi$ be some $\hornsatmod_M(n)$ instance on which the algorithm with $R$ rounds fails to find a solution even though one exists. Let $\cF$ be the set of solutions of the $\hornsat$ instance in $\Psi$ without the modulo $M$ constraint, $\cF$ is an intersection-closed family. Let $A$ be a minimal solution for $\Psi$ that obeys the modular constraint. Therefore $\FindMinimal(B) \subsetneq A$ for all $B \subset A$ with $|B| \le R$. Consider the family $\mathcal F^{A} = \{B : B\subset A,\ B \in \mathcal F\}$. We claim that $\mathcal H = \mathcal F^A \setminus \{A\}$ is a $(M, R, |A|)$-system.
\begin{enumerate}
\item $\mathcal H$ is a intersecting family since $\mathcal F^{A}$ is an intersecting family whose universe is $A$.
\item $|H| \neq |A| \mod M$ for all $H \in \mathcal H$ as $A$ is a minimal solution to $\Psi$.
\item Finally for all $S \subset A$ with $|S| \le R$, there is $H \in \mathcal H$ with $S \subset H$ because $\FindMinimal(S) \in \mathcal H$.
\end{enumerate}

\end{proof}

The following proposition from~\cite{NageleSZ18} gives bounds on $(M,R,d)$-systems when $M$ is a prime power.
\begin{proposition}[\cite{NageleSZ18}]
\label{prop:mrd_system_primepower}
Let $M$ be a prime power. Then there does not exist an $(M,R,d)$-system with $R\ge M-1$ for any $d$.
\end{proposition}
We will give a simpler proof of the above proposition in Section~\ref{sec:OR_modm_covering} by a connection to polynomials representing $\NAND\mod M$ over $\bits$ basis. Combining Propositions~\ref{prop:horn_sat_modm_correctness} and \ref{prop:mrd_system_primepower}, we have the following corollary.

\begin{corollary}
\label{cor:horn_sat_mod_primepower}
Let $M$ be a prime power. Then Algorithm~\ref{alg:horn_sat_modm} with $R=M-1$ rounds solves $\hornsatmod_M(n)$ correctly.
\end{corollary}

When $M$ is has multiple prime factors, we show that $\hornsatmod_M$ cannot be solved in polynomial time assuming $\ETH$. But we believe that for any fixed $M$, Algorithm~\ref{alg:horn_sat_modm} with $R=n^{o(1)}$ rounds should solve $\hornsatmod_M$. We make a conjecture about polynomial representations of $\NAND\mod M$ in Section~\ref{sec:OR_modm_covering} (Conjecture~\ref{conj:covering_number}) which would imply this.
\begin{proposition}
\label{cor:horn_sat_mod_nonprimepower}
Conjecture~\ref{conj:covering_number} implies that Algorithm~\ref{alg:horn_sat_modm} with $R=O_M(n^{o(1)})$ rounds solves $\hornsatmod_M(n)$ correctly for any $M$.
\end{proposition}

\subsection{$\hornsatmod_M$ with multiple modular constraints}
A natural extension of $\hornsatmod_M$ is to allow $k$ linear equations modulo $M$, which we will denote by $\hornsatmod_{M,k}$.
We can show that our algorithm and its analysis can be naturally extended to show that $\hornsatmod_{M,k}(n)$ can be solved in time $n^{k(M-1)+O(1)}$ when $M$ is a prime power. Note that once $k$ becomes linear in $n$, $\hornsatmod_{M,k}$ becomes NP-hard for any $M\ge 2$, thus exponential dependence in $k$ is necessary. A further generalization is to allow a bounded number of linear equations modulo different $M's$. A clean way to capture all these generalizations is to look at $\hornsat_G$ which is a $\hornsat$ instance along with a linear equation with coefficients from some finite abelian group $G$. This is because having $k$ equations modulo $M_1,M_2,\dots,M_t$ is equivalent to a single equation over $G=\Z/M_1\Z\times \Z/M_2\Z \times \dots \times \Z/M_k\Z$.
Our algorithm and its analysis for $\hornsatmod_M$ can be extended easily to work for $\hornsat_G$ for any finite abelian group. 
This is because any finite abelian group is a product of cyclic groups.
\begin{fact}[Structure theorem of finite abelian groups]
\label{fact:finite_abelian_structure}
Any finite abelian group $G$ is a finite product of cyclic groups of prime power order i.e. $G=\prod_i \Z/p_i^{k_i}\Z$ for some prime numbers $p_i$ (which may not be distinct) and $k_i\ge 1$.
\end{fact}
Instead of redoing everything for general groups, in the rest of this subsection we will reduce any instance of $\hornsat_G$ for any finite abelian group to a instance of $\hornsat_{\Z/M\Z}$ which is the same as $\hornsatmod_M$ without increasing the size too much.

 We will need the following lemma which shows how to convert a linear equation modulo a prime power to a polynomial equation modulo a prime.
\begin{lemma}\label{lem:covert_primepower_prime}
Let $p$ be any prime and $k\ge 1$. Let $a_0,a_1,\dots,a_n\in \set{0,1,\dots,p^k-1}$ be some integers. Then there exists a polynomial $f(x_1,\dots,x_n)$ of degree $p^k-1$ such that for every $x\in \bits^n$,
\[
f(x)=
\begin{cases}
&0 \mod p \text{ if } \sum_{i=1}^n a_ix_i=a_0 \mod p^k\\
&1 \mod p \text{ if } \sum_{i=1}^n a_ix_i\ne a_0 \mod p^k
\end{cases}.
\]
\end{lemma}
\begin{proof}
 Let $\hat{x}$ be the vector where each $x_i$ appears with multiplicity $a_i$. Then $\sum_{i=1}^n a_i x_i = \Ham(\hat{x})$. By Lemma~\ref{lem:primepower_divisibility_Hamming_2}, there exists a polynomial $\phi(y)$ of degree $p^k-1$ such that  for every bit vector $y$,
\[
\phi(y)=
\begin{cases}
&0 \mod p \text{ if } \Ham(y)=a_0 \mod p^k\\
&1 \mod p \text{ if } \Ham(y)\ne a_0 \mod p^k
\end{cases}.
\]
Therefore $f(x)=\phi(\hat{x})$ is the required polynomial.
\end{proof}

\begin{proposition}
Let $\Psi$ be an instance of $\hornsat_G(n)$ for some finite abelian group. Let $G=G_1\times G_2 \times \dots \times G_t$ where each $G_i=\prod_j \Z/p_i^{k_{ij}}\Z$ for some distinct primes $p_1,p_2,\dots,p_t$. Let $M=\prod_{i=1}^t p_i$ and $d=\max_{i\in [t]}\left(\sum_j (p_i^{k_{ij}}-1)\right)$. Then we can construct an instance $\Psi'$ of $\hornsatmod_M(N)$ where $N=\binom{n}{\le d}$ in $\poly(N)$ time such that $\Psi$ is satisfiable iff $\Psi'$ is satisfiable.
\end{proposition}
\begin{proof}
We have a constraint of the form $\sum_{\ell=1}^n g_\ell x_\ell =g_0$ for some $g_0,g_1,\dots,g_n\in G$. This is equivalent to a set of conditions of the form $\sum_{\ell=1}^n a_\ell x_\ell=a_0 \mod p_i^{k_{ij}}$ for each cyclic component in $G$. By Lemma~\ref{lem:covert_primepower_prime}, there exists a polynomial $f_{ij}(x)$ of degree $p_i^{k_{ij}}-1$ such that for every $x\in \bits^n$,
\[
f_{ij}(x)=
\begin{cases}
&0 \mod p \text{ if } \sum_{i=1}^n a_ix_i=a_0 \mod p_i^{k_{ij}}\\
&1 \mod p \text{ if } \sum_{i=1}^n a_ix_i\ne a_0 \mod p_i^{k_{ij}}
\end{cases}.
\]
Let $f_i(x)=1-\prod_j\left(1-f_{ij}(x)\right)$. The degree of $f_i$ is $d_i=\sum_j (p_i^{k_{ij}}-1)$. For every $x\in \bits^n$, $\sum_i g_ix_i=g_0$ in $G$ iff $\forall\ i\in [t]\ f_i(x)=0 \mod p_i$. We can combine the $t$ polynomial conditions $f_i(x)=0 \mod p_i$ for $i\in [t]$ into a single polynomial condition $f(x)=0 \mod M$ by Chinese remainder theorem. The degree of $f$ is $d=\max_{i\in [t]} d_i$.

$\Psi'$ will have $N=\binom{n}{\le d}$ variables corresponding to monomials of degree at most $d$ in the variables $x_1,\dots,x_n$ i.e. for every monomial $\prod_{i\in S} x_i$ with $|S|\le d$, we create a variable $y_S$ in $\Psi'$. Intuitively we would want $y_S=\prod_{i\in S} x_i$. To impose this, we will add $\hornsat$ constraints of the form $y_S\implies y_i$ for each $i\in S$ and $\wedge_{i\in S}y_i \implies y_S$. We will set $y_\phi=1$. We will also add all the original $\hornsat$ constraints of $\Psi$ into $\Psi'$ by replacing $x_i's$ with $y_i's$. Let $f(x)=\sum_{|S|\le d} b_S \prod_{i\in S} x_i$ for some integer coefficients $b_S$. We will impose the modular constraint $\sum_{|S|\le d} b_Sy_S =0 \mod M$ to $\Psi'$. Now it is clear that $\Psi'$ is satisfiable iff $\Psi$ is satisfiable. Moreover the reduction only takes $\poly(N)$ time.
\end{proof}

We have the following immediate corollary.
\begin{corollary}
Let $M=p^\ell$ be a prime power. Then $\hornsatmod_{M,k}$ can be solved in time $n^{k(M-1)(p-1)+O(1)}$.
\end{corollary}
By analyzing Algorithm~\ref{alg:horn_sat_modm} directly, one can actually reduce the running time in the above corollary to $n^{k(M-1)+O(1)}$.

\subsection{Covering number of $\NAND_d \mod M$ over $\bits^d$}
\label{sec:OR_modm_covering}
We will relate the existence of $(M,r,d)$-systems to polynomials which represent $\NAND \mod M$ over $\bits$ basis.

\begin{definition}
A polynomial $p(x_1,\dots,x_n)$ is said to represent $\NAND_n\mod M$ over $\bits^n$ if it has integer coefficients and
\[
p(x)
\begin{cases}
=& 0 \mod M \text{ if } x=\allones\\
\neq& 0 \mod M \text{ if }x\in \bits^n\setminus\{\allones\}
\end{cases}
\]
where $\allones$ is the all ones vector.
\end{definition}

We will now define the notion of covering number of polynomials.
\begin{definition}[Covering number]
The covering number of a multilinear polynomial $p(x)$, denoted by $\cov(p)$, is the minimum number of monomials of $p$ one can choose such that every variable that appears in $p$ appears in one of them.
\end{definition}
Monomials of a multilinear polynomial can be thought of subsets of variables. Thus $\cov(p)$ is the usual covering number of the set system corresponding to the monomials of $p$. We are now ready to prove the relation between $(M,r,d)$-systems and polynomials which represent $\NAND \mod M$ over $\bits$ basis. For a polynomial $p$, define $\coeffnorm{p}$ as the sum of the absolute value of its coefficients.

\begin{proposition}
\label{prop:Mrd_system_to_polynomial}
Suppose there exists an $(M,r,d)$-system with $n$ maximal sets, then there exists a polynomial $p(x_1,\dots,x_n)$ which represents $\NAND_n \mod M$ with $\coeffnorm{p}\le d+M-1$ and $\cov(p) > r$. Conversely given a polynomial $p(x_1,\dots,x_n)$ with non-negative coefficients which represents $\NAND_n \mod M$ with $|p|\le d$ and $\cov(p)> r$, there exists an $(M,r,d)$-system with at most $n$ maximal sets.
\end{proposition}
\begin{proof}
Let $\cF$ be an $(M,r,d)$-system with maximal sets $F_1,\dots,F_n$. 
Define the polynomial $\phi(x_1,x_2,\dots,x_n)$ as follows: $$\phi(x)=\sum_{a\in [d]}\prod_{i: a\notin F_i}x_i.$$
From the definition, it is clear that $\coeffnorm{\phi}\le d$.
For $x\in \bset^n$, $$\left|\cap_{i: x_i=0} F_i\right|=\sum_{a\in [d]}\prod_{i:x_i=0}\ind(a\in F_i)=\sum_{a\in [d]}\prod_{i:a\notin F_i}x_i=\phi(x).$$
Therefore for every $x\ne \ones$, we have $\phi(x)\ne d\mod M$ and $\phi(\ones)= d \mod M.$ Since any subset $S\subset [d]$ of size $r$ is contained in some $F_i$, any $r$ monomials in $\phi$ cannot cover all the $n$ variables $x_1,\dots,x_n$ i.e. $\cov(\phi)> r$. Now $p(x)=\phi(x)-(d\mod M)$ is the required polynomial which represents $\NAND_n \mod M$.

We will now show the converse. Let $p(x_1,\dots,x_n)$ be a multilinear polynomial with non-negative coefficients which represents $\NAND_n \mod M$ with $\coeffnorm{p}\le d$ and $\cov(p)> r$. Define the map $\phi:\bset^n \to \bset^d$ as $\phi(x)=(\prod_{i\in T}x_i)_T$ where the coordinates of $\phi$ are monomials of $p$ occurring with multiplicity equal to their coefficient in $p$. So, $$\Ham(\phi(x))=\sum_{i=1}^d \phi_i(x)= p(x).$$ If $a\odot b$ is the coordinate-wise product, then $\phi(x)\odot \phi(y)=\phi(x\odot y)$. Define
$$\cF=\{\phi(x):x\in \bset^n\setminus \{\ones\}\}.$$ We claim that $\cF$ is an $(M,r,d)$-system.
\begin{enumerate}
\item Let $A,B\in \cF$ and let $A=\phi(x),B=\phi(y)$ for some $x,y\in \bits^n\setminus\{\ones\}$. Then $A\cap B=\phi(x)\odot\phi(y)=\phi(x\odot y)$. Since $x\odot y\ne \ones$, $A\cap B \in \cF$.

\item  $d=\Ham(\phi(\ones))=p(\ones)=0 \mod M.$ For $A\in \cF$, $A=\phi(x)$ for some $x\in \bits^n\setminus \{\ones\}$. So $|A|=\Ham(\phi(x))=p(x)\ne 0 \mod M.$ Thus $|A|\ne d \mod M$.

\item Let $S\subset [d]$ be of size at most $r$. Let $x\in \bits^n$ be such that the variables appearing in the monomials $\{\phi_i(x):i\in S\}$ are set to $1$ and the rest of the variables are set to $0$. Since $\cov(p)> r$, $x\ne \allones$. Then the set $\phi(x)\in \cF$ contains $S$. \qedhere
\end{enumerate}
\end{proof}

Thus proving upper bounds on the covering number of polynomials which represent $\NAND \mod M$ over $\bits$ basis implies upper bounds on the number of rounds needed in Algorithm~\ref{alg:horn_sat_modm}. In the next few subsections, we will focus on proving covering number upper bounds.

\subsubsection{When $M$ is a prime power}
In this section, we will prove bounds on the covering number of polynomials which represent $\NAND \mod M$ over $\bits$ basis when $M$ is a prime power. We will collect some facts that we will need.

\begin{fact}
If $p$ is any prime, any function $f:\bits^n \to \F_p$ has a unique representation as a multilinear polynomial.
\end{fact}

The following lemma explicitly gives the polynomial which calculates $\NAND_n \mod p$ exactly over $\bits^n$.
\begin{fact}\label{fact:exact_NAND_modp}
Suppose a multilinear polynomial $f(z_1,\dots,z_n)$ exactly represents $\NAND_n \mod p$ over $\bits^n$ for some prime $p$ i.e.
\[
f(z)
\begin{cases}
=& 0 \mod p \text{ if } z=\allones\\
=& 1 \mod p \text{ if }x\in \bits^n\setminus\{\allones\}
\end{cases}.
\] Then $f(z)=1-\prod_{i\in [n]} z_i.$
\end{fact}

\begin{proposition}
\label{prop:covering_nand_primepower}
Let $M$ be a prime power, then any polynomial $f(x_1,\dots,x_n)$ which represents $\NAND_n \mod M$ over $\bits^n$ has $\cov(f)\le M-1.$
\end{proposition}
\begin{proof}
Let $M=p^k$ for some prime $p$. Wlog, we can assume that $f$ has coefficients in $\set{0,1,\dots,M-1}$. Let $\Psi_f(x)$ be the vector of monomials of $f$ where each monomial occurs with multiplicity equal to its coefficient in $f$. Therefore for $x\in \bits^n$, $\Ham(\Psi_f(x))=f(x)$. By Lemma~\ref{lem:primepower_divisibility_Hamming_2}, there exists a polynomial $\phi$ of degree $p^k-1=M-1$ such that
\[
\phi(\Psi_f(x))
\begin{cases}
=& 0 \mod p \text{ if } f(x)=0\mod p^k\\
=& 1 \mod p \text{ if } f(x)\ne 0 \mod p^k
\end{cases}.
\]
Therefore $\phi(\Psi_f(x))$ exactly represents $\NAND_n \mod p$ over $\bits^n$. By Fact~\ref{fact:exact_NAND_modp}, $\phi(\Psi_f(x))=1-\prod_{i\in [n]}x_i$. In particular, $\phi(\Psi_f(x))$ contains the monomial $\prod_{i\in [n]}x_i$. Since $\phi$ has degree $M-1$, every monomial in $\phi(\Psi_f(x))$ is the product of at most $M-1$ monomials in $f$. Thus there should be at most $M-1$ monomials in $f$ whose union contains all the variables and so $\cov(f)\le M-1$.
\end{proof}

\begin{proof}[Proof of Proposition~\ref{prop:mrd_system_primepower}]
This follows immediately by combining Propositions~\ref{prop:Mrd_system_to_polynomial} and \ref{prop:covering_nand_primepower}.
\end{proof}

\subsubsection{When $M$ has multiple prime factors}

We will now focus on the case, when $M$ has multiple prime factors. In this case, we will first show that there cannot be a constant bound on the covering number. For this, we need the following proposition by Barrington, Beigel and Rudich~\cite{BarringtonBR94}, from the paper where they first introduced polynomial representations modulo composites. It shows that there are non-trivial low degree polynomials representing $\NAND \mod M$ if $M$ has multiple prime factors.
\begin{proposition}[\cite{BarringtonBR94}]\label{prop:ORdegree_upperbound_bits}
Suppose $M$ has $r$ distinct prime factors. There exists an explicit degree $O_M(t^{1/r})$ polynomial which represents $\NAND_t \mod M$ over $\bits^t$. Moreover it can computed in time polynomial in its size.
\end{proposition}

\begin{proposition}
\label{prop:covering_nand_nonprimepower_construction}
Let $M$ be some fixed positive integer with $r\ge 2$ distinct prime factors. Then there exists a polynomial $p$ with $d$ monomials which represents $\NAND \mod M$ over $\bits$ basis such that $$\cov(f)\gtrsim_M \left(\frac{\log d}{\log\log d}\right)^{r-1}$$ for infinitely many $d$.
\end{proposition}
\begin{proof}
Let $p(x_1,\dots,x_t)$ be a polynomial of degree $O_M(t^{1/r})$ which represents $\NAND_t \mod M$ over $\bits^t$ as given by Proposition~\ref{prop:ORdegree_upperbound_bits}. The number of monomials in $p$ is $d\le \binom{t}{\le O_M(t^{1/r})}$. Now note that $$\cov(p)\ge \frac{t}{\deg(p)} \gtrsim_M t^{1-1/r} \gtrsim_M (\log d/\log\log d)^{r-1}.\qedhere$$
\end{proof}

\begin{corollary}
Let $M$ be some fixed positive integer with $t$ distinct prime factors. Then there exists an $(M,r,d)$-system with $r\gtrsim_M \left(\frac{\log d}{\log\log d}\right)^{t-1}$ for infinitely many $d$.
\end{corollary}
This addresses the open problem raised in~\cite{NageleSZ18}, where they asked if $(M,r,d)$-systems exist with $r=\omega(1)$ when $M$ has multiple prime factors. Though we show that the covering number can grow with $d$, we also conjecture that it shouldn't grow too quickly.

\begin{conjecture}
\label{conj:covering_number}
Let $M$ be some fixed constant. Any polynomial $f$ with at most $d$ monomials that represents $\NAND \mod M$ over $\bits$ basis should have $\cov(f)=d^{o_M(1)}$.
\end{conjecture}

We can prove a weaker form of Conjecture~\ref{conj:covering_number}. We will show that the natural LP relaxation of covering number is indeed small. For this we need the following lower bound on the degree of polynomials which represent $\NAND\mod M$.
\begin{proposition}[\cite{BarringtonT98}]
\label{prop:degree_OR_modm_lowerbound}
Let $M$ be a fixed constant with $r$ distinct prime factors. Let $p(x_1,\dots,x_t)$ be a polynomial representing $\NAND_t \mod M$ over $\bits^t$, then $\deg(p)\gtrsim_M (\log t)^{1/(r-1)}$.
\end{proposition}

\begin{fact}[Chernoff bound]
\label{fact:chernoff}
Let $Z_1,Z_2,\dots,Z_n$ be independent random variables taking values in $[0,1]$ and let $Z=Z_1+Z_2+\dots+Z_n$. Then for any $t>1$, $$\Pr[Z\ge t]\le \left(\frac{e\E[Z]}{t}\right)^t.$$
\end{fact}

\begin{proposition}
\label{prop:covering_LP}
Let $M$ be some fixed constant with $r$ distinct prime factors and let $p(x_1,\dots,x_n)$ be a multilinear polynomial with $d$ monomials representing $\NAND_n \mod M$ over $\bits^n$. Then $$\cov(p)\le \log(n)\cdot \exp\left(C_M(\log d)^{1-1/r}\right)$$ where $C_M>0$ is a constant depending only on $M$.
\end{proposition}
\begin{proof}
We will think of the monomials of $p$ as sets $S\subset [n]$ and let $\cF$ be the collection of monomials in $p$. We are interested in the minimum set cover from $\cF$ which covers all of $[n]$. We can write the following linear programming relaxation for this problem.

\begin{equation}
\label{eqn:CoverLP}
\begin{aligned}
\min \sum_{S\in \cF} w_S\\
w_S\ge 0\\
\forall i\in [n]\ \sum_{S\ni i} w_S \ge 1
\end{aligned}
\end{equation}

Let $L$ be the optimum value of the LP~(\ref{eqn:CoverLP}) attained for some $(w_S^*)_{S\in\cF}$. Clearly $L$ is a lower bound on the minimum set cover. By picking each subset $S\in \cF$ in the cover with probability $w_S^*$, it is not hard to see that any fixed element in $[n]$ is covered with a constant probability and the number of sets picked is $O(L)$. By repeating this $O(\log n)$ times, with high probability, all the elements of $[n]$ will be covered. Therefore $\cov(p)\le O(L\log n)$. We will now prove an upper bound on $L$.
We can write the dual of the LP~(\ref{eqn:CoverLP}) as follows:
\begin{equation}
\label{eqn:PackingLP}
\begin{aligned}
\max \sum_{i\in [n]} p_i\\
p_i\ge 0\\
\forall S\in\cF\ \sum_{i\in S}p_i \le 1
\end{aligned}
\end{equation}
By LP duality the optimum value of the LP~(\ref{eqn:PackingLP}) is also $L$ and is achieved for some $p_1^*,p_2^*,\dots,p_n^*$. Now let $\rho$ be a random restriction the variables $x_1,\dots,x_n$ where each $x_i$ is set to $1$ with probability $1-p_i^*/L^\eps$ where $\eps>0$ is a small constant that we will choose later. The restricted polynomial $p|_\rho$ represents $\NAND \mod M$ on the remaining variables. The expected number of remaining variables in $p|_\rho$ is $(\sum_{i\in [n]} p_i^*)/L^{\eps} = L^{1-\eps}$. So $p|_\rho$ has $\Omega(L^{1-\eps})$ variables left with probability $1-o(1)$.
 \begin{claim}
 $p|_\rho$ has degree $O(\frac{\log d}{\eps\log L})$ with probability $1-o(1)$.
 \end{claim}
 \begin{proof}
 Fix some $S\in \cF$ and let $t=10\frac{\log d}{\eps\log L}$. Let $Z$ be the number of variables left in $S$ after the random restriction. By Chernoff bound (Fact~\ref{fact:chernoff}),
\begin{align*}
 \Pr[Z \ge t] \le \left(\frac{e\E[Z]}{t}\right)^t = \left(\frac{e(\sum_{i\in S}p_i^*)}{L^\eps t}\right)^t \le \left(\frac{e}{L^\eps t}\right)^t =\exp (-t\log(L^\eps t/e)) \le \frac{1}{d^2}
\end{align*}
By union bounding over all the $d$ sets in $\cF$, we can conclude that every monomial in $p|_\rho$ has degree at most $t$ with probability at least $1-1/d$.
 \end{proof}
 So there exists a restriction $\rho$ such that $p|_\rho$ has degree $O(\frac{\log d}{\eps\log L})$ and $\Omega(L^{1-\eps})$ variables. So by Proposition~\ref{prop:degree_OR_modm_lowerbound}, we get $$(\log (L^{1-\epsilon}))^{1/(r-1)}\lesssim_M \frac{\log d}{\eps \log L}.$$
Choosing $\eps=1/2$, gives $L\le \exp\left(C_m(\log d)^{1-1/r}\right)$ for some constant $C_m>0$ depending only on $m$.
\end{proof}

So Proposition~\ref{prop:covering_LP} proves Conjecture~\ref{conj:covering_number} for polynomials which represent $\NAND_n \mod M$ over $\bits^n$ with $d$ monomials if $n=2^{d^{o(1)}}$. But there can be such polynomials where $n=2^{\Omega(d)}$. Showing that the covering number is small for such polynomials is open.

\subsection{Hardness of $\hornsatmod_M$}
\label{sec:hornsat_hardness}
In this section, we will show hardness for $\hornsatmod_M$ when $M$ has multiple prime factors assuming exponential time hypothesis ($\ETH$). We will show that if there are low degree polynomials representing $\OR_d \mod M$, then solving $\hornsatmod_M$ is hard.

\begin{conjecture}[Exponential time hypothesis ($\ETH$)~\cite{ImpagliazzoP01,ImpagliazzoPZ01}]
There is no $2^{o(m)}$ time algorithm for $\threesat$ with $m$ clauses.
\end{conjecture}

\begin{proposition}\label{prop:hornsatmod_hardness}
Suppose $f(\cdot)$ is some function such that for every $d$, there exists a degree $f(d)$ polynomial which represents $\NAND_d \mod M$ over $\bits^d$ which can be computed efficiently\footnote{It should be computable in time which is polynomial in its size.}. Then assuming $\ETH$, solving $\hornsatmod_M(n)$ requires at least $2^{\Omega(m)}-\poly(n)$ time for some $m$ such that $f(m)\log(m/f(m)) \gtrsim \log n$.
\end{proposition}
\begin{proof}
Choose the largest $m$ such that $n\ge \binom{3m}{\le 3f(m)}$, such an $m$ will satisfy $f(m)\log(m/f(m))=\Omega(\log n)$. Suppose $\phi(x_1,\dots,x_t)=C_1(x)\wedge C_2(x)\wedge \dots \wedge C_m(x)$ is some $\threesat$ instance with $m$ clauses and $t\le 3m$ variables where each $C_i(x)$ depends on at most 3 variables. The variables $x_1,\dots,x_t$ take $\bits$ values and each $C_i(x)$ is a polynomial of degree at most $3$ which takes these $\bits$ values and outputs $1$ if the $i^{th}$ clause is satisfied and $0$ if it is not.
So $\phi$ is satisfiable iff there exists some $x\in \bits^t$ such that $C_1(x)=\dots=C_m(x)=1$. Now let $p(z_1,\dots,z_m)$ be a polynomial of degree $f(m)$ which represents $\NAND_m \mod M$. Then $\phi$ is satisfiable iff there exists some $x\in \bits^t$ such that the polynomial $\Gamma(x)=p(C_1(x),\dots,C_m(x)) = 0 \mod M$. The polynomial $\Gamma$ has degree at most $3f(m)$ and $t$ variables, so it has at most $N=\binom{t}{\le 3f(m)}$ monomials. Let $\Gamma(x)=\sum_{S\subset [t]:|S|\le 3f(m)} a_S \prod_{i\in S}x_i$. Wlog we can assume that $a_S\in {0,1,2,\dots,M-1}$ because we only care about its values modulo $M$. We will now create an instance $\Psi$ of $\hornsatmod_M(N)$ on $N\le n$ variables such that $\Psi$ has a solution iff $\phi$ is satisfiable.
The variables in $\Psi$ will be indexed by subsets $S\subset [t]$ with $|S|\le 3f(m)$, let us denote them by $z_S$. Intuitively, we would want $z_S=\prod_{i\in S} x_i$. To enforce this, we add the following $\hornsat$ clauses to $\Psi$.
\begin{itemize}
\item For each variable $z_S$, add the clause $z_S \rightarrow z_i$ for every $i\in S$.
\item For each variable $z_S$, add the clause $\wedge_{i\in S} z_i \rightarrow z_S$.
\end{itemize}
Finally, to $\Psi$ we add the modular constraint $\sum_S a_Sz_S= 0 \mod M$. $\Psi$ will have at most $O(tN)$ clauses and $\Psi$ has a solution iff there exists an $x\in \bits^n$ such that $\Gamma(x) = 0 \mod M$. Therefore the $\hornsatmod_M(N)$ instance $\Psi$ has a solution iff the $\threesat$ formula $\phi$ is satisfiable. 
The running time of the reduction is $\poly(N)$. Therefore we can solve $\threesat$ with $m$ clauses in time $\poly(N)+T(N)$ where $T(N)$ is the time it takes to solve $\hornsatmod_M(N)$. This proves the required claim.
\end{proof}

\begin{remark}
  Note that the instance of $\hornsat$ produces by the $\ETH$ reduction does not have any constraints which force a variable to be a particular constant. This observation is needed for the dichotomy result in Appendix~\ref{app:dicot}.
\end{remark}

\begin{remark}
While the obstructions to our algorithm are polynomials which represent $\NAND \mod M$ over $\bits$ basis with high covering number, the gadgets used in the hardness proof are low degree polynomials which represent $\NAND \mod M$. Though these are clearly related, it is tempting to believe that the obstructions for the optimal algorithm should be the right gadgets to prove tight hardness results. Can we use polynomials which represent $\NAND \mod M$ over $\bits$ basis with high covering number in the hardness reduction? Can we start with something else other than $\threesat$ in the reduction?
\end{remark}

\begin{proposition}\label{prop:ORdegree_upperbound_primepower_bits}
For any integer $M\ge 2$, there exists a degree $\ceil{d/(M-1)}$ polynomial which represents $\NAND_d \mod M$ over $\bits^d$.
\end{proposition}
\begin{proof}
Partition the variables $x_1,x_2,\dots,x_d$ into $M-1$ parts of size at most $d'=\ceil{d/(M-1)}$. We can compute the $\NAND$ of each part exactly with a degree $d'$ polynomial of the form $1-\prod_{i=1}^{d'}x_i$. Adding these polynomials which compute $\NAND$ on each part exactly, we get a polynomial which represents $\NAND_d \mod M$ over $\bits^d$.
\end{proof}

We have shown that Algorithm~\ref{alg:horn_sat_modm} runs in time $n^{M+O(1)}$ when $M$ is a prime power. Combining Propositions \ref{prop:ORdegree_upperbound_primepower_bits} and \ref{prop:hornsatmod_hardness} we have the following corollary, which shows that our algorithm is nearly tight assuming $\ETH$ when $M$ is a prime power.
\begin{corollary}
Suppose $M\le n$. Assuming $\ETH$, solving $\hornsatmod_M(n)$ requires at least $n^{\Omega(M/\log M)}$ time.
\end{corollary}
We will now show hardness for $M$ which is not a prime power.

\begin{corollary}
Suppose $M$ has $r$ distinct prime factors. Assuming $\ETH$, solving $\hornsatmod_M(n)$ requires at least $\exp(\Omega_M((\log n /\log\log n)^r))$ time.
\end{corollary}
\begin{proof}
By Proposition \ref{prop:ORdegree_upperbound_bits}, we can take $f(m)= O_M(m^{1/r})$ in Proposition~\ref{prop:hornsatmod_hardness}. $m^{1/r}\log m = \Omega_M(n)$ implies that $$m\gtrsim_M \left(\frac{\log n}{\log\log n}\right)^r$$ which implies the required bound.
\end{proof}

When $M$ has $r>1$ prime factors, the lowest degree needed to represent $\OR_d \mod M$ is not well understood. The best upper known upper bound is $O_M(d^{1/r})$ as in Proposition~\ref{prop:ORdegree_upperbound_bits}~\cite{BarringtonBR94}. The best lower bound on the degree is $\Omega_M\left((\log d)^{1/(r-1)}\right)$ due to Barrington and Tardos~\cite{BarringtonT98}. If there is a polynomial whose degree matches this lower bound, then assuming $\ETH$, we can get $\bigexp{\bigexp{(\log n)^{1-1/r}}}$ hardness for $\hornsatmod_M(n)$.

\subsection{Submodular minimization with modular constraints}
\label{sec:submodular}

A function $f:2^{[n]}\to \R$ is called submodular if for every $S,T\subset [n]$, $f(S\cup T)+f(S\cap T)\le f(S)+f(T)$. We will identify $2^{[n]}$ with $\bits^n$ below, by identifying subsets with their indicator vectors.

\begin{definition}
$\submod_M(n)$ denotes the following problem\footnote{\cite{NageleSZ18} don't have coefficients in their original definition, they only look at $\Ham(x)$. But by making copies of variables, one can reduce the more general problem with coefficients to their version.}. Given an evaluation oracle to a submodular function $f:\bits^n\to \R$ and integers $m\in \Z_{>0}$ and $0\le a_0,a_1,\dots,a_n \le M-1$, find $\min\{f(x):x\in \bits^n, \sum_{i=1}^n a_ix_i =a\mod M\}.$
\end{definition}

\cite{NageleSZ18} showed that $\submod_M(n)$ can be solved in $n^{O(M)}$ time when $M$ is a prime power. They asked if their methods can be extended to prove that $\submod_M(n)$ can be solved in $n^{O_M(1)}$ for $M$ which are not prime powers. Their algorithm and its analysis is closely related to the existence of $(M,r,d)$-systems and so to $\hornsatmod_M$. We showed that, assuming $\ETH$ , $\hornsatmod_M$ cannot solved in polynomial time if $M$ is not a prime power. Can we show a similar hardness result for $\submod_M$? The following conjecture will imply such a hardness result.

Let $0\le r\le d$ be integers and let $n=\binom{d}{\le r}$. Let $\phi:\bits^d \to \bits^n$ given by $\phi(x)=(\prod_{i\in S} x_i)_{S\subset [d]:|S|\le r}$. Let $\cF_{r,d}=\phi(\bits^d)$. Note that if we think of $\cF_{r,d}$ as a collection of subsets of $[n]$, $\cF_{r,d}$ is an intersection-closed family.
\begin{conjecture}
\label{conj:submodular}For every $0\le r \le d$ and $n=\binom{d}{\le r}$, there exists a submodular function $f:\bits^n \to \Z$ such that:
\begin{enumerate}
\item $f(x)$ can be evaluated in $\poly(n)$ time for every $x\in \bits^n$.
\item For all $x\in \cF_{r,d}$ (defined as above), $f(x)\le -1$ and for all $x\notin \cF_{r,d}$, $f(x)\ge 0$.
\end{enumerate}
\end{conjecture}

\begin{proposition}\label{prop:submodular_hardness}
Suppose $\phi(\cdot)$ is some function such that for every $d$, there exists a degree $\phi(d)$ polynomial which represents $\NAND_d \mod M$ over $\bits^d$ which can be efficiently computed. Then assuming $\ETH$ and Conjecture~\ref{conj:submodular}, solving $\submod_M(n)$ requires at least $2^{\Omega(m)}/\poly(n)$ time for some $m$ such that $\phi(m)\log (m/\phi(m)) =\Omega(\log n)$.
\end{proposition}
\begin{proof}
Note that the hard instances of $\hornsatmod_M(n)$ constructed in~\ref{prop:hornsatmod_hardness} all have the same $\hornsat$ constraints, and only differ in the modular constraint. These instances need $2^{\Omega(m)}-\poly(n)$ time to solve assuming $\ETH$ for some $m$ such that $\phi(m)\log (m/\phi(m)) =\Omega(\log n)$. And the set of solutions to the $\hornsat$ constraints in these hard instances is $\cF_{3\phi(m),t}$ for some $t=O(m)$.
Given an instance $\Psi$ of $\hornsatmod_M(n)$ from these set of hard instances we can reduce it to a $\submod_M(n)$ instance $\Psi'$ where the modular constraint remains the same and the submodular function takes negative values on $\cF_{3\phi(m),t}$ and non-negative values else where, as given by Conjecture~\ref{conj:submodular}. This is a valid reduction because the value of $\Psi'$ is negative iff $\Psi$ is satisfiable. If $\submod_M(n)$ can be solved in $T(n)$ time (assuming unit time evaluation oracle), then $\Psi$ can be solved in $\poly(n)T(n)$ time (because each evaluation oracle access now costs $\poly(n)$ time). This implies that $T(n)\gtrsim 2^{\Omega(m)}/\poly(n)$.
\end{proof}

\begin{corollary}
Suppose $M$ has $r$ distinct prime factors. Assuming $\ETH$ and Conjecture~\ref{conj:submodular}, solving $\submod_M(n)$ requires at least $\exp(\Omega_M((\log n /\log\log n)^r))$ time.
\end{corollary}
\begin{proof}
By Proposition \ref{prop:ORdegree_upperbound_bits}, we can take $f(m)= O_M(m^{1/r})$ in Proposition~\ref{prop:submodular_hardness}. $m^{1/r}\log m = \Omega_M(n)$ implies that $$m\gtrsim_M \left(\frac{\log n}{\log\log n}\right)^r$$ which implies the required bound.
\end{proof}

Since the running time of the algorithm for $\submod_M(n)$ from~\cite{NageleSZ18} depends on the existence of $(M,R,d)$-systems, Conjecture~\ref{conj:covering_number} will imply non-trivial algorithms for $\submod_M(n)$ for any fixed $M$.
\begin{proposition}
\label{prop:submod_alg}
Conjecture~\ref{conj:covering_number} implies that for any fixed $M$, $\submod_M(n)$ can be solved in $\exp(n^{o_M(1)})$ time.
\end{proposition}

\section{$\linmod_M$}
\label{sec:lineq}
A $\lintwo$ instance is a system of linear equations modulo $2$ in $n$ variables $x_1,x_2,\dots,x_n \in \bits$ i.e. each equation is of the form $\sum_i a_ix_i =a_0$ for some $a_0,a_1,\dots,a_n\in \bits$. Satisfiability of a $\lintwo$ instance can be solved in polynomial time by Gaussian elimination.

\begin{definition}
$\linmod_M(n)$ is the following algorithmic problem: Given an instance of $\lintwo$ on $n$ variables $x_1,x_2,\dots,x_n$ along with a modular constraint $\sum_i a_i x_i = a_0 \mod M$, decide if there is a solution in $x\in \bits^n.$
\end{definition}
\noindent In this section, we will present an algorithm for $\linmod_M$ and analyze its running time.
\subsection{Algorithm for $\linmod_M$}
Wlog, we can assume that coefficients $a_1,\dots,a_n=1$ in the modular constraint. This is because we can make $a_i$ copies of $x_i$ and add equality constraints among the copies. And equality is a $\lintwo$ constraint as $a=b$ iff $a\oplus b=0$. Since we can assume that original coefficients $a_1,\dots,a_n\in {0,1,\dots,M-1}$, this increases the number of variables by a factor of $M$. Consider the following algorithm for this problem which depends on the parameter $R$, the number of rounds. We can calculate a basis for the set of solutions of a $\lintwo$ instance in polynomial time, so we will start with such a basis.

\begin{algorithm}[H]

\caption{Algorithm for $\linmod_M(n)$ with $R$ rounds}
\label{alg:LIN2MODm}
  \begin{itemize}
  \item Input: An affine subspace $V$ of $\F_2^n$ given by $V= \linearspan{v_1,v_2,\dots,v_d}+b$ for some linearly independent vectors $v_1,\dots,v_d\in \F_2^n$ and some vector $b\in V$, and positive integers $a,M$.
  \item Output: either a solution $x\in V$ with $\Ham(x) = a \mod M$ or $\operatorname{NO-SOLUTION}$ if no such $x$ exists.
  \item Method:
    \begin{enumerate}
    \item If there exists a subset $S\subset [d]$ of size at most $R$ such that $\Ham(b+\sum_{i\in S} v_i) = a \mod M$, output this solution.
    \item Else, return $\operatorname{NO-SOLUTION}$.
    \end{enumerate}
  \end{itemize}
\end{algorithm}

We will now prove that if we choose the number of rounds $R$ appropriately depending on $n,M$, then Algorithm~\ref{alg:LIN2MODm} solves $\linmod_M(n)$ correctly. Since the running time of the algorithm is $O(n^R)$, the smaller the $R$ the better. Surprisingly, the value of $R$ required depends crucially on the prime factor decomposition of $M$! Let us start with a simple proposition which shows that if Algorithm~\ref{alg:LIN2MODm} fails, then there should be a special kind of obstruction.
\begin{proposition}\label{prop:lin2modm_correctness}
If Algorithm~\ref{alg:LIN2MODm} with $R$ rounds fails on an instance of $\linmod_M(n)$, then there exists an affine subspace $U$ of $\F_2^n$ with dimension greater than $R$ with exactly one point $x^*\in U$ such that $\Ham(x^*)= a \mod M$.
\end{proposition}
\begin{proof}
Suppose Algorithm~\ref{alg:LIN2MODm} failed to find a solution after $R$ rounds. Therefore there exists a solution $x^*=b+\sum_{i\in S} v_i$ given a subset $S\subset [d]$ of size $|S|>R$ such that $\Ham(x^*)= a \mod M$. Wlog, we can assume that $S$ has the minimum size among such sets. Now let $U$ be the affine subspace given by $U=b+\linearspan{v_i:i\in S}$. The dimension of $U$ is $|S|$ which is greater than $R$. By minimality of $S$, every point in $y\in U$ other than $x^*$ has $\Ham(y) \neq a \mod M$.
\end{proof}

\begin{definition}
Let $n,M$ be some positive integers. $\ORdim(n,M)$ denotes the largest dimension of an affine subspace $C$ in $\F_2^n$ such that for some $0\le a \le M$, there exists exactly one point $x_0\in C$ such that $\Ham(x_0)=a \mod M$. 
\end{definition}
Therefore Proposition~\ref{prop:lin2modm_correctness} implies the following corollary.

\begin{corollary}\label{cor:lin2mod_rounds}
Algorithm~\ref{alg:LIN2MODm} with $R$ rounds solves $\linmod_M(n)$ correctly if $R\ge \ORdim(n,M)$.
\end{corollary}

The following proposition gives upper bounds on $\ORdim(n,M)$ which in turn imply upper bounds on the number of rounds sufficient for our algorithm. The bounds depend crucially on the prime factor decomposition of $M$. Our bounds when $M$ is a prime power are nearly tight. The bound for general $M$ is conditional on a conjecture in additive combinatorics called the Polynomial Freiman-Ruzsa (PFR) conjecture, which we will define in Section~\ref{sec:sparsityOR}.
\begin{proposition}
\label{prop:weight_modM_subspace}
Let $n,M$ be positive integers. Then:
\begin{enumerate}
\item $\ORdim(n,M)\le M-1$ if $M$ is a power of 2.
\item $\ORdim(n,M)\le (M-1)\log_2(n+1)$ if $M$ is an odd prime power.
\item $\ORdim(n,M)\le 1+(M'-1)\log_2(n+1)$ if $M=2M'$ for some odd prime power $M'$.
\item $\ORdim(n,M)\lesssim_\ell (M'\log n)^{2^{\ell-1}-1}$ if $M=2^\ell M'$ for some odd prime power $M'$ and $\ell \ge 2$ where $c>0$ is some absolute constant.
\item $\ORdim(n,M)\le c_M n/\log n$ for some sufficiently large constant $c_M>0$ depending only on $M$ assuming the PFR conjecture.
\end{enumerate}
\end{proposition}

We will prove Proposition~\ref{prop:weight_modM_subspace} in Section~\ref{sec:sparsityOR} by a connection to sparsity of polynomials which represent $\OR\mod M$ over $\sbits$ basis.

So when $M$ is a power of 2, we have a polynomial time algorithm for $\linmod_M(n)$. And when $M$ is a product of a power of 2 and an odd prime power, we have a quasipolynomial time algorithm. For general $M$, we have a slightly non-trivial running time of $\exp(O(n\log\log n/\log n))$, whereas the trivial algorithm which checks every solution takes $\exp(\Omega(n))$ time. By using randomization, we can considerably speed up the above algorithms. For this we make use the following proposition, which uses an amplification trick. It allows us to conclude that if there is one solution, then there should be many solutions.

\begin{proposition}\label{prop:amplification}
Let $V$ be an affine subspace of $\F_2^n$ and let $N(V,a,M)=\left|\{x\in V: \Ham(x)= a \mod M\}\right|.$ Then, $$N(V,a,M)\ne 0 \Rightarrow N(V,a,M)\ge \frac{|V|}{2^{\ORdim(n,M)+1}}.$$
\end{proposition}
\begin{proof}
Let $D=\ORdim(n,M)+1$. Wlog we can assume that the dimension of $V$ is greater than $D$, otherwise the bound is trivially true. Since $N(V,a,M)\ne 0$, we can find some $x^*\in V$ such that $\Ham(x^*)= a \mod M$. Let $T\subset V\setminus\{x^*\}$ be the set of all points in $y\in V\setminus\{x^*\}$ such that $\Ham(y) = a \mod M$.
Pick a random affine subspace $A$ inside $V$ of dimension $D$ passing through $x^*$. By the definition of $\ORdim(n,M)$, there exists an other point $z\in A\setminus\{x^*\}$ such that $\Ham(z)= a \mod M$. Therefore $|T\cap A|\ge 1$. Therefore,
\begin{align*}
1 &\le \E_A[|T\cap A|]\\
&= \E_A[\sum_{x\in V\setminus\{x^*\}} \indicator_T(x)\indicator_A(x)]\\
&= \sum_{x\in V\setminus\{x^*\}} \indicator_T(x)\E_A[\indicator_A(x)]\\
&= \sum_{x\in V\setminus\{x^*\}} \indicator_T(x) \frac{2^D-1}{|V|-1}=|T|\frac{2^D-1}{|V|-1}.\\
\end{align*}
Therefore $N(V,a,M)=1+|T|\ge 1 + \frac{|V|-1}{2^D-1}\ge \frac{|V|}{2^D}.$
\end{proof}

\begin{algorithm}
\caption{Randomized algorithm for $\linmod_M(n)$ with $T$ trials}
\label{alg:LIN2MOD_randomized}
  \begin{itemize}
  \item Input: An affine subspace $V$ of $\F_2^n$ and positive integers $a,M$.
  \item Output: a solution $x\in V$ with $\Ham(x) = a \mod M$ or $\operatorname{NO-SOLUTION}$.
  \item Method:
    \begin{enumerate}
    \item Pick a uniformly random subset of $T$ points from $V$ and output a solution if any of them satisfies the modular condition.
    \item Else, return $\operatorname{NO-SOLUTION}$.
    \end{enumerate}
  \end{itemize}
\end{algorithm}

\noindent Combining Proposition~\ref{prop:weight_modM_subspace} with Proposition~\ref{prop:amplification} we get the following corollary.
\begin{corollary}\label{prop:linmod_oddprimepower_randomized_alg}
Algorithm~\ref{alg:LIN2MOD_randomized} with $T$ trials outputs correctly given an instance of $\linmod_M(n)$ with probability at least $2/3$ if $T\ge 4 \cdot 2^{\ORdim(n,M)}$.
\end{corollary}
\begin{proof}
By Proposition~\ref{prop:amplification}, if there exists a solution, then Algorithm~\ref{alg:LIN2MOD_randomized} will find it with probability $$1-\left(1-1/2^{\cD(n,M)+1}\right)^T\ge 1-e^{-T/2^{\ORdim(n,M)+1}}\ge 1-e^{-2}\ge 2/3.\qedhere$$
\end{proof}

\subsection{$\linmod_M$ with multiple modular constraints}
A natural extension of $\linmod_M$ is to allow $k$ linear equations modulo $M$, which we will denote by $\linmod_{M,k}$.
As $k$ becomes polynomial in $n$ (i.e., $k \ge n^c$ for some constant $c > 0$), $\linmod_{M,k}$ becomes\footnote{This result follows as a consequence of Schaefer's theorem. If $M \ge 3$, a linear equation $\mod M$ can simulate a 1-in-3-SAT constraint via $x_1 + x_2 + x_3 = 1 \mod M$. Since polynomial-sized instances of 1-in-3-SAT are NP-hard, $\linmod_{M,k}$ must be NP-hard for polynomial-sized $k$.} NP-hard for any $M\ge 3$.
We can show that our algorithm and its analysis can be naturally extended to show that $\linmod_{M,k}(n)$ can be solved in time $n^{k(M-1)+O(1)}$ when $M$ is an odd prime power. Similar to $\hornsat_G$, $\lintwo_G$ which has a linear equation with coefficients from a finite abelian group $G$ is the most general form of this problem. For example, $\linmod_{M,k}$ is the same as $\lintwo_{(\Z/M\Z)^k}$. Our algorithm and its analysis can be adapted for general groups. Instead of redoing everything for general groups, we will now present a reduction from $\lintwo_G$ to $\linmod_M$ similar to our reduction from $\hornsat_G$ to $\hornsatmod_M$.

\begin{proposition}
Let $\Psi$ be an instance of $\lintwo_G(n)$ for some finite abelian group. Let $G=G_1\times G_2 \times \dots \times G_t$ where each $G_i=\prod_j \Z/p_i^{k_{ij}}\Z$ for some distinct primes $p_1<p_2<\dots<p_t$. Let
\[
M=
\begin{cases}
\prod_{i=1}^t p_i &\text{  if $p_1\ne 2$}\\
2^{d_1+1} \prod_{i=2}^t p_i &\text{  if $p_1=2$}
\end{cases}
\]
 where $d_1=\sum_j (2^{k_{1j}}-1)$ and $d=\max_{i\in [t]}\left(\sum_j (p_i^{k_{ij}}-1)\right)$. Then we can construct an instance $\Psi'$ of $\linmod_M(N)$ where $N=\binom{n}{\le d}$ in $\poly(N)$ time such that $\Psi$ is satisfiable iff $\Psi'$ is satisfiable.
\end{proposition}
\begin{proof}
We have a constraint of the form $\sum_{\ell=1}^n g_\ell x_\ell =g_0$ for some $g_0,g_1,\dots,g_n\in G$. This is equivalent to a set of conditions of the form $\sum_{\ell=1}^n a_\ell x_\ell=a_0 \mod p_i^{k_{ij}}$ for each cyclic component in $G$. By Lemma~\ref{lem:covert_primepower_prime}, there exists a polynomial $f_{ij}(x)$ of degree $p_i^{k_{ij}}-1$ such that for every $x\in \bits^n$,
\[
f_{ij}(x)=
\begin{cases}
&0 \mod p \text{ if } \sum_{i=1}^n a_ix_i=a_0 \mod p_i^{k_{ij}}\\
&1 \mod p \text{ if } \sum_{i=1}^n a_ix_i\ne a_0 \mod p_i^{k_{ij}}
\end{cases}.
\]
Let $f_i(x)=1-\prod_j\left(1-f_{ij}(x)\right)$. The degree of $f_i$ is $d_i=\sum_j (p_i^{k_{ij}}-1)$. For every $x\in \bits^n$, $\sum_i g_ix_i=g_0$ in $G$ iff $\forall\ i\in [t]\ f_i(x)=0 \mod p_i$. We now have two cases depending on $p_1=2$ or not.\\
\textbf{Case 1:} $p_1\ne 2$\\
Now for $y\in \sbits^n$, let $$h_i(y_1,\dots,y_n)=f_i\left(\frac{1-y_1}{2},\dots,\frac{1-y_n}{2}\right) \mod p_i.$$ We can assume that the coefficients of $h_i$ are in $\set{0,1,\dots,p_i-1}$ by inverting $2$ mod $p_i$.
We can combine the $t$ polynomial conditions $h_i(y)=0 \mod p_i$ for $i\in [t]$ into a single polynomial condition $h(y)=0 \mod M$ for $M=p_1p_2\dots p_t$ by Chinese remainder theorem. The degree of $h$ is $d=\max_{i\in [t]} d_i$.\\
\textbf{Case 2:} $p_1=2$\\
We will define $h_i$ as before for $i\ge 2$. And define $$h_1(y)=2^{d_1}f_1\left(\frac{1-y_1}{2},\dots,\frac{1-y_n}{2}\right).$$ Since the degree of $f_1$ is $d_1$, $h_1$ has integral coefficients. And $h_1\left(\frac{1-y_1}{2},\dots,\frac{1-y_n}{2}\right)=0 \mod 2$ iff $h_1(y)=0 \mod 2^{d_1+1}$.
We can combine $h_1(y)=0 \mod 2^{d_1+1}$ and remaining the $t-1$ polynomial conditions $h_i(y)=0 \mod p_i$ for $2\le i \le t$ into a single polynomial condition $h(x)=0 \mod M$ for $M=2^{d_1+1}p_2p_3\dots p_t$ by Chinese remainder theorem. The degree of $h$ is $d=\max_{i\in [t]} d_i$.

We now have a polynomial $h$ of degree $d$ such that for $x\in \bits^n$, $\sum_{\ell=1}^n g_\ell x_\ell = g_0$ iff $h((-1)^x)=\mod M$ where $(-1)^x=((-1)^{x_1},\dots,(-1)^{x_n})$. Let $$h((-1)^x)=\sum_{S\subset [n]:|S|\le d} a_S (-1)^{\sum_{i\in S}x_i}.$$
We are now ready to create an instance $\Psi'$ of $\linmod_M$. $\Psi'$ will have $N=\binom{n}{\le d}$ variables. For every subset $S\in \binom{[n]}{\le d}$, we create a variable $z_S$ in $\Psi'$. Intuitively we would want $z_S=\sum_{i\in S} x_i$. To impose this, we will add linear constraints of the form $z_S=\sum_{i\in S} z_i$ for every $S$. We will set $z_\phi=0$. We will also add all the original linear constraints of $\Psi$ into $\Psi'$ by replacing $x_i's$ with $z_i's$. Finally we will impose the modular constraint $\sum_{|S|\le d} a_S(1-2z_S) =0 \mod M$ to $\Psi'$, note that $(-1)^{z_S}=(1-2z_S)$. Now it is clear that $\Psi'$ is satisfiable iff $\Psi$ is satisfiable. Moreover the reduction only takes $\poly(N)$ time.
\end{proof}

We have the following immediate corollary.
\begin{corollary}
Let $M=p^\ell$ for some odd prime $p$. Then $\linmod_{M,k}$ can be solved in time $n^{O(k(M-1)(p-1))}$ with high probability.
\end{corollary}
\noindent By directly analyzing our Algorithm~\ref{alg:LIN2MOD_randomized}, one can reduce the running time to $n^{k(M-1)+O(1)}$ in the above corollary.

\subsection{Sparsity of polynomials representing $\OR_d \mod M$ over $\sbits^d$}
\label{sec:sparsityOR}
In this section, we will prove the upper bounds on $\ORdim(n,M)$ stated in Proposition~\ref{prop:weight_modM_subspace} by a reduction to understanding the sparsity of polynomials which represent $\OR \mod M$ over $\sbits$ basis (this is presented in Proposition~\ref{prop:subspace_to_ormodm_poly}). This reduction is only useful when $M$ is odd. So we will start with a separate simple argument for the case when $M$ is a power of 2, and then show how to reduce general $M$ with an odd factor to the case when $M$ itself is odd.

\subsubsection{$M$ is a power of 2}

\begin{proposition}
$\ORdim(n,2^\ell)\le 2^\ell-1.$
\end{proposition}
\begin{proof}
Let $C$ be an arbitrary affine subspace of $\F_2^n$ of dimension $d \ge 2^\ell$ which contains a point $x_0\in C$ such that $\Ham(x_0)=a \mod 2^\ell$.
By Lemma~\ref{lem:primepower_divisibility_Hamming_2}, there exists a polynomial $\phi_{\ell,a}\in \F_2[x_1,\dots,x_n]$ of degree $2^\ell-1$ such that for every $x\in \F_2^n$, $\Ham(x)=a \mod 2^\ell$ iff $\phi_{\ell,a}(x)=0$. The restriction of $\phi_{\ell,a}$ to $C$ is also a degree $2^\ell-1$ polynomial in $d$ variables. By Lemma~\ref{lem:SchwartzZippel}, $$\Pr_{x\in C}[\phi_{\ell,a}(x)=\phi_{\ell,a}(x_0)]\ge \frac{1}{2^{2^\ell-1}}.$$ Since $\Ham(x_0)=a\mod 2^\ell$, $\phi_{\ell,a}(x_0)=0$. Therefore $|\{x\in C: \phi_{\ell,a}(x)=0\}|\ge 2^{d-(2^\ell-1)}\ge 2$.
\end{proof}
This proves part (1) of Proposition~\ref{prop:weight_modM_subspace}.

\subsubsection{When $M=2^\ell M'$ for some odd $M'$}
In this subsection, we will reduce the case when $M=2^\ell M'$ to the case when $M$ is odd. If $M=2M'$ for some odd $M'$, then the reduction is easy.
\begin{proposition}\label{prop:reduction_M=2M'}
Let $M=2M'$ for some odd $M'$ then $\ORdim(n,M)\le 1+\ORdim(n,M').$
\end{proposition}
\begin{proof}
Supppse $C$ is an affine subspace in $\F_2^n$ of dimension $d =\ORdim(n,M)$ which contains exactly one point $x_0$ such that $\Ham(x_0)=a \mod M$. Let $C'=C\cap\set{x:\Ham(x)=a\mod 2}$. Now $C'$ is an affine subspace in $\F_2^n$ of dimension $\ge d-1$ such that there $x_0$ is the only point in $C'$ with $\Ham(x_0)=a \mod M'$. Therefore $d-1\le \ORdim(n,M')$, which proves the claim.
\end{proof}

To analyze the case when $M=2^\ell M'$ for $\ell\ge 2$ and some odd $M'$, we will need the following lemma which states a low degree polynomial over $\F_2$ has a large subspace in which it is constant.
\begin{lemma}[\cite{CohenT15}]
\label{lem:lowdegree_subspace_constant}
Let $f_1,f_2,\dots,f_t \in \F_2[x_1,\dots,x_n]$ be polynomials of degree at most $r$ and let $x_0\in \F_2^n$ be some fixed point. Then there exists an affine subspace containing $x_0$ of dimension $\Omega((n/t)^{1/(r-1)})$ on which each of the $f_i$ is constant.
\end{lemma}

\begin{proposition}\label{prop:reduction_M=2ellM'}
Let $M=2^\ell M'$ for some $\ell\ge 2$ and odd $M'$. Then $$\ORdim(n,M)\le O(\ell)\cdot \ORdim(n,M')^{2^{\ell-1}-1}.$$
\end{proposition}
\begin{proof}
Suppose $C$ is an affine subspace in $\F_2^n$ of dimension $d=\ORdim(n,M)$ which contains exactly one point $x_0$ such that $\Ham(x_0)=a \mod M$. Let $a=a_0+a_1 2+\dots+a_{\ell-1}2^{\ell-1}$ be the binary expansion of $a$.
By Lemma~\ref{lem:primepower_divisibility_Hamming}, for every $x\in \F_2^n$, $\Ham(x)=a \mod 2^\ell$ iff $s_{2^t}(x)=a_t\ \forall\ 0\le t\le \ell-1$.
The restriction of $s_{2^t}$ to $C$ is still a degree $2^t$ polynomial. Therefore by Lemma~\ref{lem:lowdegree_subspace_constant}, there exists an affine subspace $A$ of $C$ containing $x_0$ of dimension $\Omega((d/\ell)^{1/(2^{\ell-1}-1)})$ on which $s_1,s_2,\dots,s_{2^{\ell-1}}$ are constant. Therefore at every point $x\in A$, $s_{2^t}(x)=s_{2^t}(x_0)=a_t$ for every $0\le t \le \ell-1.$
Therefore for every $x\in A$, $\Ham(x)=a \mod 2^\ell$. Thus $A$ contains exactly one point $x$ such that $\Ham(x)= a \mod M'$, namely $x_0$. Therefore $\dim(A)\le \ORdim(n,M')$ which implies that $d\le O(\ell) \cdot \ORdim(n,M')^{2^{\ell-1}-1}$.
\end{proof}

\subsubsection{Reduction of $\ORdim(n,M)$ to sparsity of $\OR \mod M$ over $\sbits$ basis for odd $M$}

Let's start with the definition of a polynomial representation of $\OR \mod M$ over $\sbits$ basis.
\begin{definition}
A polynomial $p(x_1,\dots,x_d)$ is said to represent $\OR_d\mod M$ over $\sbits^d$ if it has integer coefficients and
\[
p(x)
\begin{cases}
=& 0 \mod M \text{ if } x=\allones\\
\neq& 0 \mod M \text{ if }x\in \sbits^d\setminus\{\allones\}
\end{cases}
\]
where $\allones$ is the all ones vector.
\end{definition}

To make the required reduction, we need the following proposition which relates the Hamming weights of points in a $d$-dimensional affine subspace to evaluations of a polynomial over $\sbits^d$. For $\ba=(a_1,\dots,a_\ell)\in \F_2^\ell$, let $(-1)^\ba$ denote the vector $((-1)^{a_1},\dots,(-1)^{a_\ell})$. Given a polynomial $p$ with integer coefficients, define $\coeffnorm{p}$ as sum of the absolute value of its coefficients. Note that the number of monomials in $p$ is always at most $|p|$.

\begin{proposition}\label{prop:subspace_to_polynomial}
Let $\bb,\bu_1,\dots,\bu_d \in \F_2^n$. Then there exists a multilinear polynomial $p(z_1,\dots,z_d)$ with integer coefficients and $\coeffnorm{p}=n$ such that for every $\by\in \F_2^d$, $$n-2\Ham\left(\bb+\sum_i y_i \bu_i\right)=p((-1)^\by).$$
Conversely, given any multilinear polynomial $p(z_1,\dots,z_d)$ with integer coefficients and $\coeffnorm{p}=n$, there exists $\bb,\bu_1,\dots,\bu_d\in \F_2^n$ satisfying the above identity.
\end{proposition}
\begin{proof}
We will first start with $\bb,\bu_1,\dots,\bu_d\in \F_2^n$ and construct such a polynomial $p$. Let $\bb=(b_1,\dots,b_n)$ and $\bu_i=(u_{i1},\dots,u_{in})$. Define the map $\phi:\sbits^d \to \sbits^n$ as: $$\phi_t(z_1,\dots,z_d)= (-1)^{b_t}\prod_{i\in [d]: u_{it}=1}z_i.$$ Then,
\begin{align*}
\phi_t((-1)^\by)&= (-1)^{b_t}\prod_{i\in [d]:u_{it}=1} (-1)^{y_i}
= (-1)^{b_t}\prod_{i\in [d]} (-1)^{y_iu_{it}}
= (-1)^{b_t+\sum_{i\in [d]}y_i u_{it}}.
\end{align*}
Therefore for any $\by\in \F_2^d$, $\phi((-1)^\by)=(-1)^{\bb+\sum_{i=1}^d y_i \bu_i}$. Define $p$ as: $p(z)=\sum_{t=1}^n \phi_t(z).$
Note that $p$ has integer coefficients and $\coeffnorm{p}=n$. And finally, $$\Ham(\bb+\sum_{i\in [d]}y_i \bu_i)=\sum_{t=1}^n \frac{1}{2}\left(1-\phi_t((-1)^\by)\right)=\frac{1}{2}\left(n-p((-1)^\by)\right).$$
To prove the converse, we just execute the steps of the above construction in reverse. Given a polynomial $p(z_1,\dots,z_d)$ with integer coefficients and $\coeffnorm{p}=n$, let $\phi_1(z),\dots,\phi_n(z)$ be (signed, possibly repeated) monomials in $z$ be such that $p(z)=\sum_{t=1}^n \phi(z)$. Now define $\bb,\bu_1,\dots,\bu_d\in \F_2^n$ such that $\phi_t(z)=(-1)^{b_t}\prod_{i\in [d]:u_{it}=1}z_i$ is true, explicitly
$$u_{it}=\begin{cases}
1 & \text{ if } \phi_t(z) \text{ contains } z_i\\
0 & \text{ else }
\end{cases}.$$
Then by the same argument as above, the required identity is satisfied.
\end{proof}

The following proposition shows the connection between affine subspaces which are obstructions for Algorithm~\ref{alg:LIN2MODm} and polynomials representing $\OR \mod M$ over $\sbits$ basis.
\begin{proposition}\label{prop:subspace_to_ormodm_poly}
Suppose $M$ is odd. Let $d=\ORdim(n,M)$, then there exists a polynomial $f(x_1,\dots,x_d)$ with at most $n+1$ monomials that represents $\OR_d\mod M$ over $\sbits^d$. Conversely, given a polynomial with $n$ monomials which represents $\OR_d \mod M$ over $\sbits^d$, $\ORdim(n',M)\ge d$ for some $n'\le M\cdot n$.
\end{proposition}
\begin{proof}
Suppose $U$ is a $d$-dimensional affine subspace of $\F_2^n$ which contains exactly one point $x^*\in U$ such that $\Ham(x)= a \mod M$.
Let $\bb,\bu_1,\dots,\bu_d\in \F_2^n$ be such that $U=\bb+\linearspan{\bu_1,\dots,\bu_d}$. By Proposition~\ref{prop:subspace_to_polynomial}, there exists a polynomial $p(z_1,\dots,z_d)$ with integer coefficients and at most $n$ monomials such that for every $\by\in \F_2^d$, $n-2\Ham(\bb+\sum_{i=1}^d y_i \bu_i)=p((-1)^\by)$. Suppose $x^*=\bb+\sum_i y_i^*\bu_i$, then $p(z)= (n-2a)\mod M$ for exactly one $z$ in $\sbits^d$ given by $z=z^*=(-1)^{y^*}$ (here we are using the fact that $M$ is odd).

Define the polynomial $f(z)=p(z\odot z^*)-(n-2a)$ where $z\odot z^*$ is the coordinate wise product. Note that $f(\allones)= 0 \mod M$ and $f(z)\neq 0 \mod M$ for all $z\in \sbits^d \setminus \{\allones\}$. Thus $f(z)$ is a polynomial which represents $\OR_d \mod M$ over $\sbits^d$ and $f$ has at most $n+1$ monomials.

To prove the converse, suppose $f(z)$ is a polynomial with $n$ monomials which represents $\OR_d \mod M$ over $\sbits^d$. Wlog, we can assume that the coefficients of $f$ are in $\{0,1,\dots,M-1\}$. Let $n'=|f|$, by our assumption about coefficients of $f$, $n'\le Mn$. By the converse part in Proposition~\ref{prop:subspace_to_polynomial}, there exists $\bb,\bu_1,\dots,\bu_{d}\in \F_2^{n'}$, such that $n'-2\Ham(\bb+\sum_{i=1}^d y_i \bu_i)=f((-1)^\by)$ for every $\by\in \F_2^d$. Note that $\bu_1,\dots,\bu_d$ must be linearly independent. If not, there exists a $\by\ne 0$ such that $f((-1)^\by)=f(\allones)=0 \mod M$ which is a contradiction.

Therefore in the $d$-dimensional affine subspace given by $V=\{\bb+\sum_{i=1}^{n'} y_i \bu_i: \by\in \F_2^d\}$, there exists exactly one point $\bx^*\in V$ (given by $\bx^*=\bb)$) such that $\Ham(\bx^*)=(n/2) \mod M$. Thus  $\ORdim(n',M)\ge d$.

\end{proof}

Because of the above proposition, if we prove sparsity lower bounds on polynomials which represent $OR \mod M$ then we get good upper bounds on the number of rounds that will be enough in Algorithm~\ref{alg:LIN2MODm}.

\subsubsection{When $M$ is an odd prime power}
Now we will show that when $M$ is an odd prime power, a polynomial which represents $\OR_d \mod M$ over $\sbits^d$ should have exponential number of monomials. We will collect some facts that we need to prove this.

\begin{fact}
If $p$ is an odd prime, any function $f:\sbits^d \to \F_p$ has a unique representation as a multilinear polynomial.
\end{fact}

The following lemma explicitly gives the polynomial which calculates $\OR_d \mod p$ exactly.
\begin{fact}\label{fact:exact_OR_modp}
Suppose a multilinear polynomial $f(z_1,\dots,z_d)$ exactly represents $\OR_d \mod p$ over $\sbits^d$ for some odd prime $p$ i.e.
\[
f(z)
\begin{cases}
=& 0 \mod p \text{ if } z=\allones\\
=& 1 \mod p \text{ if }x\in \sbits^d\setminus\{\allones\}
\end{cases}.
\] Then $f(z)$ has $2^d$ monomials and explicitly given by, $$f(z)=1-\prod_{i\in [d]}\frac{(1+z_i)}{2}.$$
\end{fact}

The following proposition provides a sparsity lower bound when $M$ is an odd prime power.
\begin{proposition}\label{prop:OR_mod_oddprimepower_sparity}
Suppose $M$ is an odd prime power. If a polynomial $f(z)$ represents $\OR_d \mod M$ over $\sbits^d$, then $f$ has at least $2^{d/(M-1)}$ monomials.
\end{proposition}
\begin{proof}
Let $M=p^k$ for some odd prime $p$. Wlog, we can assume that the coefficients of $f$ are in $\set{0,1,\dots,M-1}$. Let $\Psi'_f(z)$ be the vector of monomials of $f$ evaluated at  $z$, where each monomial appears with multiplicity equal to its coefficient in $f$, and let $N$ be its length. For $z\in \sbits^d$, $\Psi'_f(z)\in \sbits^N$ and $f(z)=\sum_{i\in [N]} (\Psi'_f(z))_i$. Let $\Psi_f(z)$ be the vector of the same length as $\Psi'_f(z)$ whose coordinates are given by $$(\Psi_f(z))_i=\frac{1+(\Psi'_f(z))_i}{2}.$$ For $z\in \sbits^d$, $\Psi_f(z)\in \bits^N$ and $$\Ham(\Psi_f(z))=\sum_{i=1}^N \frac{1+(\Psi'_f(z))_i}{2}=\frac{N}{2}+\frac{f(z)}{2}.$$ Therefore for $z\in \sbits^d$, $\Ham(\Psi_f(z)) = N/2 \mod p^k$ iff $z=\allones$.

By Lemma~\ref{lem:primepower_divisibility_Hamming_2}, there exists a polynomial $\phi$ of degree $p^k-1=M-1$ such that
\[
\phi(\Psi_f(z))=
\begin{cases}
&0 \mod p \text{ if } \Ham(\Psi_f(z))=N/2 \mod p^k\\
&1 \mod p \text{ if } \Ham(\Psi_f(z))\ne N/2 \mod p^k
\end{cases}.
\]

Therefore $\phi(\Psi_f(z))$ exactly represents $\OR_d \mod p$ and therefore by Fact~\ref{fact:exact_OR_modp}, it has $2^d$ monomials. Since $\phi$ has degree $M-1$, the number of monomials in $\phi(\Psi_f(z))$ is at most $\binom{N'}{\le M-1}$ where $N'$ is the number of (distinct) monomials in $f$. Therefore, $N'\ge 2^{d/(M-1)}$.
\end{proof}

Thus we have the following corollary which proves part (2) of Proposition~\ref{prop:weight_modM_subspace}.
\begin{corollary}\label{cor:subspace_modm_oddprimepower}
Let $M$ be an odd prime power. Then $\ORdim(n,M)\le (M-1)\log_2(n+1)$.
\end{corollary}
\begin{proof}
Suppose $C$ is an affine subspace of $\F_2^n$ of dimension $d>(M-1)\log_2(n+1)$ which contains exactly one point $x_0$
such that $\Ham(x_0)=a \mod M$. Then by Proposition~\ref{prop:subspace_to_polynomial}, there exists a polynomial $p(x_1,\dots,x_d)$ with at most $n+1$ monomials such that $p(x)$ represents $OR_d \mod M$ over $\sbits^d$. Therefore by Proposition~\ref{prop:OR_mod_oddprimepower_sparity}, $(n+1)\ge 2^{d/(M-1)}$ which is a contradiction.
\end{proof}

We remark that the above bound is nearly tight. Let $n=(M-1)2^d$. The Hadamard code is a subspace of $\F_2^{2^d}$ of dimension $d$ such that every non-zero point in the subspace has weight $2^{d-1}$. Decompose $\F_2^n=\bigoplus \F_2^{2^d}$ where the copies of $\F_2^{2^d}$ are supported on mutually disjoint sets of variables. Let $V$ be the subspace of $\F_2^n$ which is the direct sum of Hadamard codes in each copy of $\F_2^{2^d}$. Then $V$ has dimension $(M-1)d=(M-1)\log_2(n/(M-1))$ and every non-zero point in $V$ has weight in $\set{1\cdot 2^{d-1}, 2\cdot 2^{d-1},\dots,(M-1)\cdot 2^{d-1}}$ which is non-zero modulo $M$.

Combining Corollary~\ref{cor:subspace_modm_oddprimepower} with Propositions~\ref{prop:reduction_M=2M'} and \ref{prop:reduction_M=2ellM'} implies parts (3) and (4) of Proposition~\ref{prop:weight_modM_subspace}.

\subsubsection{When $M$ has multiple odd prime factors}
We will now focus on the case when $M$ is odd and has multiple prime factors. Unfortunately, in this case we do not know any unconditional super linear lower bounds on the sparsity of polynomials representing $\OR_d \mod M$ over $\sbits^d$. But we can get a conditional super linear lower bound, assuming Polynomial Freiman-Ruzsa (PFR) conjecture which is a well-known conjecture in additive combinatorics. We achieve this by constructing matching vector families in $(\Z/M\Z)$ starting from sparse representations of $\OR_d \mod M$ over $\sbits^d$. We will first define matching vector families.

\begin{definition}
A matching vector family (MVF) over $\Z/M\Z$ of rank $r$ and size $N$ is a collection of vectors $\bu_1,\dots,\bu_N\in (\Z/M\Z)^r$ and $\bv_1,\dots,\bv_N \in (\Z/M\Z)^r$ such that for every $i,j\in [N]$:
\[\\
\inpro{\bu_i}{\bv_j}
\begin{cases}
&=0 \mod M \text{ if } i=j\\
&\ne 0 \mod M \text { if } i\ne j.
\end{cases}
\]
\end{definition}

MVFs over $\Z/M\Z$ of low rank and large size have found applications in many areas. They are used in the construction of constant query locally decodable codes~\cite{Yekhanin08,Efremenko12,DvirGY11}, Ramsey graphs~\cite{Grolmusz00,Gopalan14}, private information retrieval schemes~\cite{DvirG16} and secret sharing schemes~\cite{LiuV18}. In particular, this implies that lower bounds for constant query locally decodable codes give lower bounds on the rank of MVFs of a given size. For example, super polynomial lower bounds on the length of constant query locally decodable codes imply that the sparsity of a polynomial representing $\OR_d \mod M$ over $\sbits^d$ should be $\omega(d)$. But only polynomial lower bounds on constant query locally decodable codes are known~\cite{KatzT00,KerenidisdW04}\footnote{For 2-query locally decodable codes, it is known that the length of the encoding should be exponential in the message length~\cite{KerenidisdW04}. But for $q\ge 3$, the best lower bounds are only polynomial.}. In fact, we do not even know any strong unconditional lower bounds on the rank of MVFs over $\Z/M\Z$ i.e. results of the form $N\le \exp(o_M(r))$. But assuming the PFR conjecture, the following bound is known. We will not state the PFR conjecture here, for the precise statement see ~\cite{BhowmickDL14}.

\begin{proposition}[\cite{BhowmickDL14}]
\label{prop:MVF_lowerbound}
Assuming the Polynomial Freiman-Ruzsa conjecture over $(\Z/M\Z)^r$, any MVF over $\Z/M\Z$ of rank $r$ should have size $N \le \exp\left(O_M(r/\log r)\right).$
\end{proposition}

We are now ready to prove the super linear lower bound on the sparsity of polynomials representing $\OR_d \mod M$ over $\sbits^d$ assuming PFR.
\begin{proposition}\label{prop:OR_mod_nonprimepower_sparsity}Assuming the Polynomial Freiman-Ruzsa (PFR) conjecture in $(\Z/M\Z)^d$, any polynomial which represents $\OR_d \mod M$ over $\sbits^d$ needs to have $\Omega_M(d\log d)$ monomials.
\end{proposition}
\begin{proof}
Given a sparse polynomial which represents $\OR_d \mod M$ over $\sbits^d$, we will construct a MVF over $\Z/M\Z$ of small rank and large size. The construction is a based on a similar construction due to Sudan which first appeared in~\cite{Gopalan09}.
Suppose $p(z)=\sum_{t=1}^r a_t \prod_{i\in S_t} z_i$ is a polynomial with $r$ monomials which represents $\OR_d \mod M$ over $\sbits^d$. Define the $2^d \times 2^d$ matrix $A$ with rows and columns indexed by $\sbits^d$ as $$A(z,z')=p(z\odot z') \mod M$$ where $z\odot z'$ is the component-wise product. Note that the diagonal entries $$A(z,z)=p(z\odot z)=p(\allones) =0 \mod M$$ and the off-diagonal entries are $$A(z,z')=p(z\odot z') \ne 0 \mod M.$$
Moreover the rank of the matrix $A$ is at most $r$ since $$A(z,z')=p(z\odot z')=\sum_t a_t\prod_{i\in S_t} z_iz_i'=\inpro{\left(a_t\prod_{i\in S_t}z_i\right)_t}{\left(\prod_{i\in S_t}z_i'\right)_t}.$$
Therefore the set of vectors $\bu_z=(a_t\prod_{i\in S_t}z_i)_{t\in [r]}\in (\Z/M\Z)^r$ and $\bv_{z'}=(\prod_{i\in S_t}z'_i)_{t\in [r]}\in (\Z/M\Z)^r$ for $z,z'\in \sbits^d$ form a MVF over $(\Z/M\Z)$ of size $N=2^d$ of rank $r$. By Proposition~\ref{prop:MVF_lowerbound}, $$N\le \exp\left(O_M(r/\log r)\right)\Rightarrow d\le O_M(r/\log r) \Rightarrow r \ge \Omega_M(d\log d).\qedhere $$
\end{proof}

So we have the following corollary which proves part (5) of Proposition~\ref{prop:weight_modM_subspace}.
\begin{corollary}\label{cor:subspace_modm_pfr}
Assuming the PFR conjecture, for every positive integer $M$, there exists a constant $c_M$ depending only on $M$ such that, $\ORdim(n,M)\le c_M n/\log n$.
\end{corollary}
\begin{proof}
Suppose $C$ is an affine subspace of $\F_2^n$ of dimension $d>c_M n/\log n$ which contains exactly one point $x_0$ such that $\Ham(x_0)=a \mod M$. Then by Proposition~\ref{prop:subspace_to_polynomial}, there exists a polynomial $p(x_1,\dots,x_d)$ with at most $n+1$ monomials such that $p(x)$ represents $OR_d \mod M$ over $\sbits^d$. Therefore by Proposition~\ref{prop:OR_mod_nonprimepower_sparsity}, $(n+1)\ge \Omega_M(d\log d)$ which is a contradiction if we choose $c_M$ sufficiently large.
\end{proof}

\subsection{Hardness of $\linmod_M$}
\label{sec:linmod_hardness}

We will now show hardness for $\linmod_M$ when $M$ has multiple odd prime factors assuming exponential time hypothesis ($\ETH$). We will show that if there are low degree polynomials representing $\OR_d \mod M$, then solving $\linmod_M$ is hard.

\begin{proposition}\label{prop:linmod_hardness}
Let $M$ be odd and suppose $f(\cdot)$ is some function such that for every $d$, there exists a degree $f(d)$ polynomial which represents $\OR_d \mod M$ over $\sbits^d$ which is efficiently computable. Then assuming $\ETH$, solving $\linmod_M(n)$ requires at least $2^{\Omega(m)}-\poly(n)$ time for some $m$ such that $f(m)\log (m/f(m)) \gtrsim \log(n)$.
\end{proposition}
\begin{proof}
Choose the largest $m$ such that $\binom{3m}{\le 3f(m)}\le n$, such an $m$ will satisfy $f(m)\log (m/f(m)) \gtrsim \log(n)$. Suppose $\phi(x)=C_1(x)\wedge C_2(x)\wedge \dots \wedge C_m(x)$ is some $\threesat$ instance with $m$ clauses and $t\le 3m$ variables where each $C_i(x)$ depends on at most 3 variables. We can assume that the variables $x_1,\dots,x_t$ take $\sbits$ values and each $C_i(x)$ is a polynomial which takes these $\sbits$ values and outputs $1$ if the $i^{th}$ clause is satisfied and $-1$ if it is not.
So $\phi$ is satisfiable iff there exists some $x\in \sbits^t$ such that $C_1(x)=\dots=C_m(x)=1$. Now let $p(z_1,\dots,z_m)$ be a polynomial of degree $f(m)$ which represents $\OR_m \mod M$. Then $\phi$ is satisfiable iff there exists some $x\in \sbits^t$ such that the polynomial $\Psi(x)=p(C_1(x),\dots,C_m(x)) = 0 \mod M$. The polynomial $\Psi$ has degree at most $3f(m)$ and so it has at most $\binom{t}{\le 3f(m)}\le n$ monomials. Wlog we can assume that $\Psi$ has coefficients in $\set{1,2,\dots,M-1}$ because we only care about its values modulo $M$, let us denote these coefficients by $a_1,a_2,\dots,a_n$. Emulating the proof of Proposition~\ref{prop:subspace_to_polynomial}, there exists $\bu_1,\dots,\bu_t\in \F_2^n$ (which can be computed efficiently from $\Psi$)  such that for every $\by \in \F_2^t$, $$\Psi((-1)^\by)=\sum_{j=1}^n a_j\left(1-2(\sum_{i=1}^t y_i \bu_i)_j\right)=\sum_{j=1}^n a_j -2\sum_{j=1}^n a_j(\sum_{i=1}^t y_i \bu_i)_j.$$ Therefore $\Psi(x)= 0 \mod M$ for some $x\in \sbits^t$ iff there exists some $\bx' \in \linearspan{\bu_1,\dots,\bu_t}$ such that $\sum_{j=1}^n a_j(\bx')_j= (\sum_{j=1}^n a_j)/2 \mod M$. We can write the condition $\bx' \in \linearspan{\bu_1,\dots,\bu_t}$ as a system of linear equations over $\F_2$ that $\bx'$ should satisfy, explicitly, $U^\perp \bx'=\bar{0}$ where $U^\perp$ is the matrix whose rows form a basis for the orthogonal complement of $\linearspan{\bu_1,\bu_2,\dots,\bu_t}$.

Thus we reduced an instance of $\threesat$ with $m$ clauses to an instance of $\linmod_M(n)$. The reduction itself takes $\poly(n)$ time. By $\ETH$, $\threesat$ requires $2^{\Omega(m)}$ time. This proves that we need $2^{\Omega(m)}-\poly(n)$ time to solve $\linmod_M(n)$.
\end{proof}

\begin{remark}
Note that the gadgets we used in the hardness proof are low-degree polynomials which represent $\OR \mod M$ over $\sbits$ basis. Whereas the obstructions to our algorithm are sparse polynomials which represent $\OR \mod M$ over $\sbits$ basis. It is tempting to believe that the obstructions to the optimal algorithm should be the right gadgets that should be used in the hardness proof. Here is a different reduction. Start with a GAP-3LIN instance $\phi(x)=(E_1(x),\dots,E_m(x))$ over $t$ variables where it is promised that either $1-\epsilon$ fraction of equations are satisfiable or less than $1/2+\eps$ fraction are satisfiable. GAP-3LIN is NP-hard and there are near-linear time reductions from $\threesat$ to GAP-3LIN~\cite{MoshkovitzR10}. Suppose $p(z_1,\dots,z_m)$ is a polynomial over $\sbits^m$ such that it weakly represents this (partial) threshold function modulo $M$. That is the values of $p(z)\mod M$ when $\sum_i z_i \ge (1-2\eps)m$ and when $\sum_i z_i\le 2\eps m$ are disjoint, say $S_1$ and $S_0$ respectively. Then $\phi$ is $(1-\eps)$-satisfiable iff there exists some $x\in \sbits^t$ such that $\Psi(x)=p(E_1(x),\dots,E_m(x))\in S_1$. But note that in the $\sbits$ basis, the sparsity of $\Psi$ is the same as sparsity of $p$. Thus we get a good hardness reduction if there are sparse polynomials which weakly represent the $(1-\eps,1/2+\eps)$-threshold partial function. Another interesting question is, can we use the non-existence of such sparse polynomials in creating a good algorithm for $\linmod_M$?
\end{remark}

\begin{proposition}\label{prop:ORdegree_upperbound_primepower_sbits}
For any odd integer $M\ge 2$, there exists a degree $\ceil{d/(M-1)}$ polynomial which represents $\OR_d \mod M$ over $\sbits^d$.
\end{proposition}
\begin{proof}
Partition the variables $x_1,x_2,\dots,x_d$ into $M-1$ parts of size at most $d'=\ceil{d/(M-1)}$. We can compute the $\OR$ of each part exactly with a degree $d'$ polynomial of the form $1-\prod_{i=1}^{d'}\left(\frac{1+x_i}{2}\right)$. Note that powers of 2 in the denominator can be inverted $\mod M$ to get a polynomial with integer coefficients. Adding these polynomials which compute $\OR$ on each part exactly, we get a polynomial which represents $\OR_d \mod M$ over $\sbits^d$.
\end{proof}

We have shown that $\linmod_M(n)$ can be solved in randomized time $n^{M+O(1)}$ when $M$ is an odd prime power. Combining Propositions \ref{prop:ORdegree_upperbound_primepower_sbits} and \ref{prop:linmod_hardness} we have the following corollary, which shows that our this running time is nearly tight assuming $\ETH$ when $M$ is an odd prime power.
\begin{corollary}
Suppose $3\le M\le cn$ be an odd integer for some small enough constant $0<c<1$. Assuming $\ETH$, solving $\linmod_M(n)$ requires at least $n^{\Omega(M/\log M)}$ time.
\end{corollary}

\begin{corollary}
Suppose $M$ has $r$ distinct odd prime factors. Assuming $\ETH$, solving $\linmod_M(n)$ requires at least $\exp(\Omega_M((\log n /\log\log n)^r))$ time.
\end{corollary}
\begin{proof}
Let $M=2^\ell M'$ for some odd $M'$. Since we can reduce $\linmod_{M'}(n)$ to $\linmod_M(n)$ easily, it is enough to show hardness for $\linmod_{M'}(n)$. So wlog, we can assume $M$ is odd.

When $M$ is odd, a polynomial which represents $\NAND$ over $\bits$ basis can be converted into a polynomial with represents $\OR$ over $\sbits$ basis by a linear basis change which preserves the degree.
Therefore by Proposition \ref{prop:ORdegree_upperbound_bits}, we can take $f(m)= O_M(m^{1/r})$ in Proposition~\ref{prop:linmod_hardness}. $m^{1/r}\log m = \Omega_M(\log n)$ implies that $$m\gtrsim_M \left(\frac{\log n}{\log\log n}\right)^r$$ which implies the required bound.
\end{proof}

\section{$\twosatmod_M$}\label{sec:2sat}

In this section, we present an algorithm for $\twosatmod_M$. The algorithm is recursive and for the recursion to work, we need to consider the more general list version of $\twosatmod_M$. Our algorithm works for any abelian group $G$ in place of $\Z/M\Z$. So we will consider $\twosat$ with a global modular constraint over a finite abelian group $G$, which we call $\modcsp(\twosat, G, S)$. In this definition, $S$ is the set of permitted values in the modular constraint.

The $\twosat$ instance itself will be over variables $\{x_1, \hdots, x_n\}$. We say that the set of literals are $V := \{x_1, \hdots, x_n, \bar{x}_1, \hdots, \bar{x}_n\}$, where $\bar{x}_i$ represents the negation of $x_i$. We let $E \subset V \times V$ be the constraints, where $(y, z) \in E$ implies that $(y \vee z)$ is a constraint.

As is standard, we can interpret a constraint $y \vee z$ as a pair of implications $\bar{y} \rightarrow z$ and $\bar{z} \rightarrow y$. As such, we consider a complementary implication digraph $(V, F)$ where $F = \{(\bar{y}, z) : (y, z) \in E\} \cup \{(\bar{z}, y) : (y, z) \in E\}.$

For the global modular constraint, we have a constraint of the form
\[
  \sum_{j = 1}^n g_j(x_j) \in S,
\]
where $g_j : \{0, 1\} \to G$. To make it more symmetric in the literals, we can write each $g_j$ as the sum of indicator functions $g_j(x) = g_{x_j}(x) + g_{\bar{x_j}}(\bar{x})$ such that
\begin{align}
  g_{x_j}(x) &= \begin{cases}
    g_j(1) & x = 1\\
    0 & x = 0
  \end{cases} & g_{\bar{x}_j}(x) &= \begin{cases}
    g_j(0) & x = 1\\
    0 & x = 0
  \end{cases}.
        \label{eq:5}
\end{align}

\subsection{Preprocessing}

The first step of our algorithm is to do standard preprocessing on the constraint graph. Note that if our implication graph has a cycle $y_1 \rightarrow y_2 \rightarrow \cdots \rightarrow y_k \rightarrow y_1$, then we can deduce that $y_1 = y_2 = \cdots = y_k$, and we can replace these variables by a single variable $y$ and replace the corresponding $g_y$'s by their sum. By the duality of the implications, there must also be a component $\bar{y}_k \rightarrow \bar{y}_{k-1} \rightarrow \cdots \rightarrow \bar{y}_1 \rightarrow \bar{y}_k$. Thus, whenever two vertices are merged in the constraint graph, their negations are also merged. The only exception is if a cycle contains both a $y$ and its negation $\bar{y}$, in which case we can safely output \textsf{NO SOLUTION}. This preprocessing is described in Algorithm~\ref{alg:2-SAT-pre}.

\begin{algorithm}
  \caption{Preprocessing for $\modcsp(\twosat, G, S)$}
  \label{alg:2-SAT-pre}
  \begin{itemize}
  \item Input: Instance $\Psi$ of $\twosat$ over $x_1, \hdots, x_n$ along with maps $g_j : \{0, 1\} \to G$ and a global constraint $\sum_{j=1}^n g_j(x_j)\in S$.
  \item Output: Either \textsf{NO SOLUTION} or a new instance $\Psi'$ on $\{x'_1, \hdots, x'_k\}$ with an acyclic implication digraph $(V' = \{x'_1, \hdots, x'_k, \bar{x}'_1, \hdots, \bar{x}'_k\}, F')$ as well as new maps $g'_y : \{0, 1\} \to G$, $y \in V'$ such that $\Psi$ is satisfiable iff $\Psi'$ is satisfiable.
  \item PREPROCESS \begin{enumerate}
    \item Let $V := \{x_1, \hdots, x_n, \bar{x}_1, \hdots, \bar{x}_n\}$, and construct $g_y, y \in V$ as described in (\ref{eq:5}).
    \item Construct the implication set $F$, and compute the strongly connected components in $F$.
    \item If any $y, \neg{y}$ appear in the same component, then output \textsf{NO SOLUTION}. Otherwise, label the components $\{C_1, \hdots, C_k, \bar{C}_1, \hdots, \bar{C}_k\}$, where $\bar{C}_i$ has the complements of the literals in $C_i$.
    \item For $i \in \{1, \hdots, k\}$
      \begin{enumerate}
      \item Let $x'_i$ and $\bar{x}'_i$ be new variables representing $C_i$ and $\bar{C}_i$.
      \item Let $g'_{x'_i} = \sum_{y \in C_i} g_y$ and $g'_{\bar{x}'_i} = \sum_{\bar{y} \in C_i} g_{\bar{y}}$.
      \end{enumerate}
    \item For every pair of components for which there is at least one edge from one to the other, add an edge between the corresponding variables to $F'$.
    \end{enumerate}
  \end{itemize}
\end{algorithm}

It is clear that the run-time of the algorithm is $O(n+m)$. The correctness of this algorithm follows from the following claim.

\begin{claim}
  There is a bijection between solutions to $\Psi$ and $\Psi'$ which preserves the weights according to $i$ and $i'$, respectively.
\end{claim}

\begin{proof}
  Fix a solution $\ba := (a_1, \hdots, a_n)$ to $\Psi$. For each strongly connected component $C_i$ (or $\bar{C}_i$), because there is a directed walk of implications from any pair of literals $y, z$ in the same $C_i$, all those literals must have the same common value in that strongly connected component. Thus, $\ba'=(a'_1, \hdots, a'_k)$, where $a'_i$ is defined to be the common value of $C_i$, is well-defined and the map $\ba \mapsto \ba'$ is injective. This $\ba'$ is a valid solution to $\Psi'$ as any implication $y' \to z'$ in $\Psi'$ is constructed from an implication $y \to z$ in $\Psi$. Also, as $C_i$ and $\bar{C}_i$ have complementary variables, $x'_i$ and $\bar{x}'_i$ are complementary.

  By definition,
  \begin{align*}
    &\sum_{i = 1}^k g'_{x'_i}(a'_i) + \sum_{i=1}^k g'_{\bar{x}'_i}(\bar{a}'_i)\\
    &= \sum_{i=1}^k\sum_{y \in C_i} g_{y}(a'_i) + \sum_{i=1}^k\sum_{\bar{y} \in \bar{C}_i} g_{\bar{y}}(\bar{a}'_i)\\
    &= \sum_{i=1}^k\sum_{y \in C_i} g_{y}(a_y) + \sum_{i=1}^k\sum_{\bar{y} \in \bar{C}_i} g_{\bar{y}}(a_{\bar{y}})\ \ \ \ \ \ \ \ \left(a_y = \begin{cases}a_i & \text{if }y=x_i\\\bar{a}_i &\text{if }y=\bar{x}_i\end{cases}\right)\\
    &= \sum_{i=1}^n g_{x_i}(a_i) + \sum_{i=1}^n g_{\bar{x}_i}(\bar{a}_i).
  \end{align*}

  For the other direction, consider a solution $\ba' := (a'_1, \hdots, a'_k)$ to $\Psi'$. We can lift this solution in the opposite manner by setting literal $x_i$ equal to the literal of $\ba'$ corresponding to the strongly connected component to which $x_i$ belongs. This lifting is also injective, and it preserves the property that all the constraints are satisfied, as whenever an implication $x_i \to x_j$ is needed, it exists between the corresponding components. By running the above equations in reverse, we have that this lifted solution also preserves the modular constraint.

  Thus, we have a bijection between the solution sets of these problems.
\end{proof}

It is also not hard to see that if there is an edge from $C_i$ to $C_j$, there must also be an edge from $\bar{C}_j$ to $\bar{C}_i$. This, along with logic from similar cases, shows that $\Psi'$ does indeed encode a $\twosat$ instance.

\subsection{Acyclic case}

With the preprocessing algorithm complete, we can now assume that the implication graph of the $\twosat$ instance is a directed acyclic graph, DAG. The algorithm is presented in Algorithm~\ref{alg:2-SAT-DAG}.

\begin{algorithm}[H]
  \caption{DAG algorithm}\label{alg:2-SAT-DAG}
  \begin{itemize}
  \item Input: Instance $\Psi$ of $\twosat$ as exhibited by an acyclic implication digraph on $V = \{x_1, \hdots, x_n, \bar{x}_1, \hdots, \bar{x}_n\}$ and edge set $F$. Maps $g_j : \{0, 1\} \to G$ for all $j\in [n]$ and a global constraint $\sum_{j=1}^n g_j(x_j)\in S$.
  \item Output: Either \textsf{NO SOLUTION} or an assignment $x_1, \hdots, x_n \in \{0, 1\}^n$ satisfying the global constraint.
  \item 2-SAT-DAG($V$, $F$, $\{g_v\}$, $S$) \begin{enumerate}
    \item Define $g_v:\bits \to G$ for all $v\in V$ as in (\ref{eq:5}).
    \item If $n = 0$, check if $0 \in S$. If so, output the empty assignment. Otherwise output \textsf{NO SOLUTION}.
    \item Otherwise, select a literal $y$ with outdegree $0$.
    \item Set $y = 1$; that is, let $V' = V \setminus \{y, \bar{y}\}$ and $F' = F \cap (V' \times V')$ and $S' = \{s - g_{y}(1) : s \in S\}.$
    \item Check if 2-SAT-DAG($V'$, $F'$, $\{g_v\}$, $S'$) has a solution, and if so, add ``$y = 1$'' and ``$\bar{y} = 0$'' to the assignment.
    \item Otherwise, set $y = 0$. This choice forces all literals $z$ in the DAG for which there is a path from $z$ to $y$ to be assigned the value $0$. Let $W$ be the set of all such literals (including $y$ itself), and let $\bar{W}$ be the set of complements of literals in $W$. If $W \cap \bar{W} \neq \emptyset$ output \textsf{NO SOLUTION}.
    \item Otherwise, set $V'' = V \setminus (W \cup \bar{W})$, $F'' = F \cap (V'' \times V'')$, and $S'' = \{s - \sum_{z \in \bar{W}} g_z(1) : s \in S\} \cup \{s - g_{y}(1) - \sum_{z \in \bar{W} \setminus \{\bar{y}\}} g_z(1)\}$.
    \item If $|S''| = |S|$, then output \textsf{NO SOLUTION}.
    \item Check if 2-SAT-DAG($V''$, $F''$, $\{g_v\}$, $S''$) has a solution. If so, take the assignment and add ``$z = 0$'' and ``$\bar{z} = 1$'' for all $z \in W$, and return this solution.
    \item Otherwise, output \textsf{NO SOLUTION}.
    \end{enumerate}
  \end{itemize}
\end{algorithm}

Both the run-time and analysis of this algorithm take some work to analyze.

\subsubsection{Correctness of Algorithm~\ref{alg:2-SAT-DAG}}

We prove correctness by inducting on $n$. Clearly the case $n = 0$ is correct, as there are no variables and so the modular constraint is equivalent to $0 \in S$.

Assume the induction hypothesis is true for all $n < N$ for some positive integer $N$. Because the DAG is acyclic, a vertex of outdegree $0$ must exist; call this literal $y$ as in the algorithm. Note that $\bar{y}$ must have indegree $0$.  Any valid assignment to the $\twosat$ constraints (ignoring the modular constraint) must still be valid when $y$ is set equal to $1$ (and $\bar{y}$ is set equal to $0$). In particular, this means that setting $y = 1$, $\bar{y} = 0$ and removing both literals from the digraph will leave us with a valid DAG $\twosat$ instance. By the induction hypothesis, this similar instance can be solved (with the adjusted set $S'$) and then can be lifted back up to get a solution to $\Psi$ and the global modular constraint. Thus, if the smaller instance has a valid solution, we are done.

Now what if the smaller instance fails to have a valid solution? Clearly then any valid solution must have $y = 0$ and $\bar{y} = 1$. Thus, any element of $W$ (those literals which through a chain of implications lead to $y$) must also have value $0$, and those which belong to $\bar{W}$ must have value $1$ (since they are lead to by a chain of implications from $\bar{y}$). Thus, if $\bar{W}$ and $W$ intersect, the $\twosat$ instance is inconsistent and we can safely reject.

Let $S''_0 = \{s - \sum_{z \in \bar{W}} g_z(1) : s \in S\}$. Clearly running 2-SAT-DAG on the digraph with $W \cup \bar{W}$ deleted will work using $S''_0$. But recall we can flip the values of $y$ and $\bar{y}$ to get another valid solution. Thus, if a solution on the smaller DAG has weight in $S''_1 = \{s - g_{y}(1) - \sum_{z \in \bar{W}\setminus \{\bar{y}\}} g_z(1) : s \in S\}$, we can lift to a valid solution on the full DAG. If we run 2-SAT-DAG on the smaller digraph with $S'' = S''_0 \cup S''_1$, we can always lift back to a valid solution.

But, if $|S''| = |S|$, then $S''_1 = S''_0 = S''$. Thus, any valid solution with $y = 0$ is also a valid solution with $y = 1$. But, we have already ruled out in the first recursive case that no solutions with $y = 0$ exist. Thus, we can safely reject in this scenario.

When we do run on the smaller digraph, we know from the above logic that getting a modular value in $S''$ is equivalent to getting a modular value in $S''\setminus S''_1 \subset S''_0$ (as any other element would imply a contradiction). Thus, as long as \textsf{NO SOLUTION} is not output from the recursive call, we can extend to a full solution.

\subsubsection{Run-time of Algorithm~\ref{alg:2-SAT-DAG}}

We claim that this algorithm runs in time $O((m+n) \cdot n^{|G| - |S|})$. Let $f(n, m, k)$ be an upper bound on the running time on instances with $n$ variables, $m$ implications, and $k = |G| - |S|$. If $k = 0$, then the modular constraint is trivial so $f(n, m, 0) = O(n + m)$ (the run-time of 2-SAT). For $k \ge 1$, note that
\[
  f(n, m, k) \le f(n - 1, m, k) + f(n - 1, m, k - 1) + O(m+n).
\]
This is because, $f(n - 1, m, k)$ is an upper bound on the work of the first recursive call, and $f(n - 1, m, k-1)$ is an upper bound on the work of the second recursive call (if it is run). This recursion is consistent with a run-time of $O((m + n)n^{k})$, as desired. This yields the following proposition.

\begin{proposition}
  Let $G$ be an Abelian group. Then an instance of $\twosatmod_G$ on $n$ variables and $m$ constraints can be solved in $(n + m)^{|G| + O(1)}$ time.
\end{proposition}

For special groups, this analysis can be improved. For example, when $G = \mathbb F_2^k$, then the algorithm in fact runs in $(n + m)^{O(k)}$ time. 

\appendix
\newcommand{\XOR}{\mathsf{XOR}}

\section{Establishing the Boolean Mod-CSP dichotomy}\label{app:dicot}

\subsection{PP-reductions, Polymorphisms, and Galois correspondence}

One family of simple gadget reductions from one CSP to another are known as pp-reductions. These are the gadget reductions used by Schaefer to prove the dichotomy for Boolean CSPs. They are formally defined as follows.

\begin{definition}
  Let $\Gamma$ and $\Gamma'$ be templates over a domain $D$. We say that there is a \emph{primitive positive reduction} from $\Gamma'$ to $\Gamma$ if for all $C' \in \Gamma'$ there exist $C_1, \hdots, C_k \in \Gamma$ (perhaps with repetition) such that
  \[
    C'(x_1, \hdots, x_{\ell}) = \exists y_1, \hdots, y_{\ell'} \bigwedge_{i=1}^k C_i(z_{i,1}, \hdots, z_{i,\ar_i}),
  \]
  where each $z_{i,j}$ is an $x_{i'}$ or a $y_{i'}$, allowing for repetition.
\end{definition}
Informally a pp-reduction means that every constraint in $\Gamma'$ can be expressed as a conjunction of constraints in $\Gamma$ possibly with the addition of some auxiliary variables.

Note that if there is a pp-reduction from $\Gamma'$ to $\Gamma$ then there exists a polynomial-time reduction from $\CSP(\Gamma')$ to $\CSP(\Gamma)$ (in fact a logspace reduction). (e.g., \cite{Chen2009, DBLP:conf/dagstuhl/BartoKW17}).

As stated, it is rather difficult to determine whether there exists a pp-reduction between two templates $\Gamma$ and $\Gamma'$. This issue can be resolved by looking at the \emph{polymorphisms} of these constraint templates.

\begin{definition}
  Let $C \subset D^k$ be a constraint. A \emph{polymorphism} of $C$ is a function $f : D^L \to D$ such that for all $x^1, \hdots, x^L \in C$, we also have that $f(x^1, \hdots, x^L) \in C$.\footnote{Here $f$ acts coordinate-wise i.e. the $j^{th}$ coordinate of $f(x^1,\dots,x^L)$ is obtained by applying $f$ to the $j^{th}$ coordinates of $x^1,\dots,x^L$.} More pictorially (c.f., \cite{DBLP:conf/dagstuhl/BartoKW17}),
  \begin{center}
    \begin{tabular}{ccccc}
      & $x^1_1$ & $\cdots$ & $x^1_k$ &$\in C$\\
      & $x^2_1$ & $\cdots$ & $x^2_k$ &$\in C$\\
      & $\vdots$ & $\vdots$ & $\vdots$ & $\in C$\\
      & $x^L_1$ & $\cdots$ & $x^L_k$ & $\in C$\\\hline
     $f\Downarrow$ & $y_1$ & $\cdots$ & $y_k$ & $\in C$
    \end{tabular}
  \end{center}
  The set of such polymorphisms is denoted by $\Pol(C)$. For a general template $\Gamma$, the set of polymorphisms is \[\Pol(\Gamma) := \bigcap_{C \in \Gamma} \Pol(C).\]
\end{definition}

Here are a few examples (see \cite{DBLP:conf/dagstuhl/BartoKW17} for many more).

\begin{enumerate}
\item Let $\MAJ_k : \{0, 1\}^k \to \{0, 1\}$ be the bitwise majority operator on $k$ bits, then $\MAJ_k \in \Pol(\twosat)$ for all odd $k$.
\item Let $\XOR_k : \{0, 1\}^k \to \{0, 1\}$ be the bitwise XOR on $k$ bits, then $\XOR_k \in \Pol(\threexor)$ for all odd $k$.
\item Let $\AND_k : \{0, 1\}^k \to \{0, 1\}$ be the bitwise AND operator on $k$ bits, then $\AND_k \in \Pol(\hornsat)$ for all $k$.
\end{enumerate}

Intuitively, polymorphisms capture high-dimensional symmetries in the constraints. If the constraints have many symmetries (such as linear constraints are closed under affine operations), then the corresponding CSPs should be more likely to be tractable. This can be stated rigorously as a \emph{Galois correspondence}.

\begin{theorem}[Galois correspondence for pp-reductions, \cite{jeavons88}]\label{thm:galois}
  Let $\Gamma$ and $\Gamma'$ be templates over a domain $D$. There exists a pp-reduction from $\Gamma'$ to $\Gamma$ if and only if $\Pol(\Gamma) \subseteq \Pol(\Gamma')$.
\end{theorem}

Thus, to classify the computational complexity of CSPs, it suffices to classify sets of polymorphisms.  Such an investigation was done by Post~\cite{post} (in a slightly more general context) in the case of Boolean polymorphisms. This classification along with the Galois correspondence yields an elegant restatement of Schaefer's theorem.

\begin{theorem}[Schaefer's theorem, polymorphism version~\cite{Schaefer:1978,Chen2009,DBLP:conf/dagstuhl/BartoKW17}]
  Let $\Gamma$ be a Boolean template. Either $\CSP(\Gamma)$ is NP-complete or it falls into one of the six following cases.
  \begin{enumerate}
  \item $0 \in \Pol(\Gamma)$, in which case ``all zeros'' is a solution to every instance.
  \item $1 \in \Pol(\Gamma)$, in which case ``all ones'' is a solution to every instance.
  \item $\AND_2 \in \Pol(\Gamma)$, in which case $\Gamma$ is pp-reducible to $\hornsat$.
  \item $\OR_2 \in \Pol(\Gamma)$, in which case $\Gamma$ is pp-reducible to $\dualhornsat$.
  \item $\MAJ_3 \in \Pol(\Gamma)$, in which case $\Gamma$ is pp-reducible to $\twosat$.
  \item $\XOR_3 \in \Pol(\Gamma)$, in which case $\Gamma$ is pp-reducible to $\lintwo$.
  \end{enumerate}
\end{theorem}

\subsection{Extension to Mod-CSPs}

We now would like to take this theory of CSPs and port it to Mod-CSPs. To start, we show that the notion of pp-reduction is still meaningful for Mod-CSPs.

\begin{claim}
  Fix an Abelian group $G$ and $S \subset G$. Consider two $\Gamma_1$ and $\Gamma_2$ such that there is a pp-reduction from $\Gamma_1$ to $\Gamma_2$, then there is a polynomial time reduction from $\modcsp(\Gamma_1, G, S)$ to $\modcsp(\Gamma_2, G, S)$.
\end{claim}
\begin{proof}
  Consider an instance of $\modcsp(\Gamma_1, G, S)$ with local constraints $\Psi(x_1, \hdots, x_n)$ and the global constraint
  \begin{align}
    g_1(x_1) + \cdots + g_n(x_n) \in S. \label{eq:global-app}
  \end{align}
  The pp-reduction says that $\Psi(x_1, \hdots, x_n)$ is equivalent to $\Psi'(x_1, \hdots, x_n; y_1, \hdots, y_m)$, where $\Psi'$ is a formula with constraints from $\Gamma_2$.

  Also observe that (\ref{eq:global-app}) is equivalent to
  \begin{align}
    g_1(x_1) + \cdots g_n(x_n) + h_1(y_1) + \hdots + h_m(y_m) \in S \label{eq:global-app2}
  \end{align}
  where $h_1 = \cdots = h_m = 0$.

  Thus, $\Psi'$ and (\ref{eq:global-app2}) in $\modcsp(\Gamma_2, G, S)$ is equivalent to $\Psi$ and (\ref{eq:global-app}) in $\modcsp(\Gamma_1, G, S)$.
\end{proof}

Since Mod-CSPs are preserved under pp-reductions, Theorem~\ref{thm:galois} tells us that Mod-CSPs with more polymorphisms are at least as tractable.

To understand the complexity of Mod-CSPs, we use a result of Post~\cite{post}, which is explicitly stated in \cite{Chen2009}. First, we need a definition

\begin{definition}[c.f., \cite{Chen2009}]
  An operator $f : D^L \to D$ is \emph{essentially unary} if there exists $i \in [L]$ such that $f(x_1, \hdots, x_L) = f(y_1, \hdots, y_L)$ whenever $x_i = y_i$.
\end{definition}

Note this definition says that constant functions are essentially unary. The other common example are \emph{dictator} (or \emph{projection}) functions: $f(x) = x_i$.

\begin{theorem}[Theorem 5.1 of \cite{Chen2009}]\label{thm:post}
  Let $\Gamma$ be a Boolean template such that there exists $f(x_1, \hdots, x_L) \in \Pol(\Gamma)$ which is not essentially unary. Then, at least one of $\OR_2$, $\AND_2$, $\MAJ_3$, $\XOR_3$ is in $\Pol(\Gamma)$.
\end{theorem}

Now, we show that Boolean Mod-CSPs whose polymorphisms only have essentially unary operators are NP-complete.

\begin{lemma}\label{lem:mod-csp-hardness}
  Let $\Gamma$ be a CSP template over a domain $D = \{0, 1\}$. Let $G$ be an Abelian group, and let  $S$ be a nontrivial subset of $G$ ($S \neq \emptyset$ and $S \neq G$). If $\Pol(\Gamma)$ consists entirely of essentially unary operators, then $\modcsp(\Gamma, G, S)$ is NP-complete.
\end{lemma}

By virtue of Schaefer's dichotomy theorem, this result does not hold for CSPs as constant polymorphisms can lead to tractability. Thus, we need to the use the global modular constraint to ``break'' these constant solutions.

\begin{proof}
  This will be shown via a reduction from graph 3-coloring like in \cite{DBLP:conf/coco/BrakensiekG16}. Let $C = \{1, 2, 3\}$ be the colors. Let $F = \{(1, 2), (1, 3), (2, 1), (2, 3), (3, 1), (3, 2)\}$ be all valid ways of coloring an edge.

  Let $(H, E)$ be a connected graph. For each vertex $v \in H$, construct a collection of variables $x_{v}(d_1, d_2, d_3)$ for all $d_1, d_2, d_3 \in D$. Likewise, for each edge $(u, v) \in E$ (think of the edge as directed so we can distinguish the two vertices), construct $y_{(u, v)}(d_{(i,j)})_{(i, j)\in F}$ for all $d_{(i, j)} \in D$ where $(i, j) \in F$.

  To talk about assignments to the variables, we let $f_v : D^C \to D$ be such that $f_v(d_1, d_2, d_3)$ is the value assigned to $x_{v}(d_1, d_2, d_3)$. For each edge $(u, v)$ we define $g_{u, v} : D^F \to D$ similarly.

  Now, we constrain that the $f$'s and $g$'s are polymorphisms. Fix a $v \in H$. For any constraint $R \in \Gamma$ on $k$ variables and for all $r^c \in R$ for all $c \in C$, we then specify that
  \[
    (f_v(r^1_1, r^2_1, r^3_1), f_v(r^1_2, r^2_2, r^3_2), \hdots, f_v(r^1_k, r^2_k, r^3_k)) \in R.
  \]
  By definition of a polymorphism, the valid assignments to $f_v$ are precisely the polymorphisms of $\Gamma$.

  Likewise, for all $(u, v) \in E$ and $R \in \Gamma$ and for all $r^e \in R$ for $e \in F$ we specify
  \[
    (g_{u,v}(r^e_1)_{e \in F}, \hdots g_{u, v}(r^e_k)_{e \in F}) \in R.
  \]

  So far we haven't linked these different polymorphisms to each other. To do that, we specify\footnote{Note that $=$ can always be simulated by using a common variable for all the equal instances, so we do not need to add $=$ to $\Gamma$.} that
  \[
    f_u(d_1, d_2, d_3) = g_{u, v}(d_{(1,2)}, d_{(1,3)}, d_{(2, 1)}, d_{(2, 3)}, d_{(3, 1)}, d_{(3, 2)})
  \]
  if for all $(i, j) \in F$, $d_i = d_{(i, j)}$. We likewise say that
  \[
    f_v(d_1, d_2, d_3) = g_{u, v}(d_{(1,2)}, d_{(1,3)}, d_{(2, 1)}, d_{(2, 3)}, d_{(3, 1)}, d_{(3, 2)})
  \]
  if for all $(i, j) \in F$, $d_j = d_{(i, j)}$. Formally, we are saying that $f_u$ and $f_v$ are \emph{minors} (or \emph{projections}) of $g_{u, v}$.

  Let $0 \in G$ be the identity and pick $s_0 \in S$ (possibly $0$). Let $s_1 \in G \setminus S$.

  Fix $v_0 \in V$. Specify that
  \begin{align*}
    g_{x_{v_0}(0, 0, 0)}(d) &= \begin{cases}
      0 & d = 0\\
      s_1 - s_0 & d = 1
    \end{cases}\\
    g_{x_{v_0}(1, 0, 0)}(d) &= \begin{cases}
      s_1 & d = 0\\
      s_0 & d = 1
    \end{cases}
  \end{align*}

  Let all other $i$'s be $0$. This completes the global constraint.

  Now, we need to show that are reduction is complete and sound. For completeness, if $(H, V)$ has a valid 3-coloring $c : H \to C$, there must be a permutation of the colors such that $c(v_0) = 1$. Consider the assignment

  \begin{align*}
    f_v(d_1, d_2, d_3) &= d_{c(v)},\ \ \ \ \ \ v \in H\\
    g_{u, v}(d_{(i,j)}) = d_{(c(u), c(v))}\ \ \ \ \ \ (u, v) \in E.
  \end{align*}
  It is clear that $f_v$ and $g_{u, v}$ are polymorphisms and that the minor constraints are satisfied. For the global constraint, observe that $g_{x_{v_0}(0, 0, 0)}(0) + g_{x_{v_0}(1, 0, 0)}(1) = 0 + s_0 \in S$.

  For the soundness, imagine that the instance of $\modcsp(\Gamma, G, S)$ has a solution. Thus, each $f_v$ and $g_{u, v}$ either is constant or nontrivially depends on a single coordinate. It is apparent from the minor relations that if $g_{u, v}$ depends nontrivially on coordinate $(i, j)$ then $f_u$ must depend nontrivially on coordinate $i$ and $f_v$ must depend nontrivially on coordinate $j$. Conversely, if $f_u$ depends nontrivially on coordinate $i$ then $g_{u, v}$ depends nontrivially on coordinate $(i, j)$ for some $j$. Similarly, if $f_v$ depends nontrivially on coordinate $j$ then $g_{u, v}$ depends nontrivial on $(i, j)$ for some $i$.

  Since $(H, E)$ is connected, we must either have that all of the $f$'s and $g$'s are constant or they all nontrivially depend on some coordinate. If the latter case occurs, we can assign a color to each vertex $v \in H$ based on which coordinate $f_v$ nontrivially depends on. The relations between these coordinates in the previous paragraph shows that this assignment is a valid $3$-coloring.

  Thus, $(H, E)$ is 3-colorable as long as the assignment is not constant in each polymorphism. But, if the assignment is constant on each $f_v$, and in particular $f_{v_0}$, this would imply that the global constraint either satisfies
  \[
    g_{x_{v_0}(0, 0, 0)}(0) + g_{x_{v_0}(1, 0, 0)}(0) = s_1,
  \]
  or
  \[
    g_{x_{v_0}(0, 0, 0)}(1) + g_{x_{v_0}(1, 0, 0)}(1) = s_1,
  \]
  but $s_1 \not\in S$, so we have a contradiction.
\end{proof}

\begin{remark}
  Note that the size of the reduction is linear in the size of the original instance $|H| + |V|$. Since graph 3-coloring cannot be done in $2^{o(|H|+|V|)}$ time assuming ETH, we have that such $\modcsp(\Gamma, G, S)$ cannot be done solve in $2^{o(n)}\poly(m)$ time (where $n$ is the number of variables and $m$ is the number of constraints) assuming ETH.
\end{remark}

\begin{remark}
  This result also holds for non-Boolean domains $D = \{1, \hdots, k\}$. The reduction is essentially identical, except the global constraint is modified so that there are $k$ nontrivial functions with
  \begin{align*}
    g_1(d) + \cdots g_k(d) &= s_1 \text{ for all } d \in D\\
    g_1(1) + \cdots g_k(k) &= s_0,
  \end{align*}
  which is certainly possible as there are $k^2$ variables ($k$ per function), but only $k+1$ constraints.
\end{remark}

\begin{remark}
  For non-Boolean CSPs, weaker conditions on $\Pol(\Gamma)$ are known to imply that $\CSP(\Gamma)$ Is NP-complete (e.g., \cite{DBLP:conf/dagstuhl/BartoKW17}). We leave as a challenge to the reader to find a suitable extension of such results to (non-Boolean) Mod-CSPs.
\end{remark}

With these structural results for Boolean Mod-CSPs, we can now state a few dichotomy-like results.

\subsection{Classification for prime powers}

\begin{theorem}\label{thm:boolean-mod-csp-dicot}
  Let $\Gamma$ be a Boolean CSP template. Let $G$ be a nontrivial Abelian group whose order is a prime power. Then, we have the following classification.

  \begin{enumerate}
  \item If one of  $\MAJ_3, \OR_2, \AND_2 \in \Pol(\Gamma)$, then $\modcsp(\Gamma, G)\in \mathsf{P}$.
  \item Otherwise, if $\XOR_3 \in \Pol(\Gamma)$ and $|G|$ is a power of two, then $\modcsp(\Gamma, G) \in \mathsf{P}$.
  \item Otherwise, if $\XOR_3 \in \Pol(\Gamma)$, then $\modcsp(\Gamma, G) \in \mathsf{RP} \cap \mathsf{QP}$.
  \item Otherwise, $\modcsp(\Gamma, G)$ is NP-complete.
  \end{enumerate}
\end{theorem}

\begin{proof}
  We prove the cases in order.

  \begin{enumerate}
  \item If $\MAJ_3 \in \Pol(\Gamma)$, then by Schaefer's theorem there is a pp-reduction from $\CSP(\Gamma)$ to $\twosat$. Thus, there is a pp-reduction from $\modcsp(\Gamma, G)$ to $\twosatmod_G$, which can be solved in polynomial time by Section~\ref{sec:2sat}.

    If $|G|$ is a prime power and $\AND_2 \in \Pol(\Gamma)$, then by Schaefer's theorem there is a pp-reduction from $\modcsp(\Gamma, G)$ to $\hornsatmod_G$, which can be solved in polynomial time by Corollary~\ref{cor:horn_sat_mod_primepower}. Likewise, if $\OR_2 \in \Pol(\Gamma)$, then there is a pp-reduction from $\modcsp(\Gamma, G)$ to $\dualhornsatmod_G$ which can also be solved in polynomial time by the same theorem, since solving $\hornsat$ and $\dualhornsat$ instances are equivalent.
  \item If $|G|$ is a power of two and $\XOR_3 \in \Pol(\Gamma)$, by Schaefer's theorem this is a pp-reduction from $\modcsp(\Gamma, G)$ to $\linmod_G$, which can be solved in deterministic polynomial time by Corollary~\ref{cor:lin2mod_rounds}.
  \item If $|G|$ is an odd prime power and $\XOR_3 \in \Pol(\Gamma)$, by Schaefer's theorem this is a pp-reduction from $\modcsp(\Gamma, G)$ to $\linmod_G$, which can be solved in randomized polynomial time and deterministic quasi-polynomial time by Proposition~\ref{prop:linmod_oddprimepower_randomized_alg} and Corollary~\ref{cor:lin2mod_rounds}.
  \item By Theorem~\ref{thm:post}, the only polymorphisms of $\Gamma$ are essentially unary. Thus, by Lemma~\ref{lem:mod-csp-hardness}, $\modcsp(\Gamma, G)$ is NP-complete. \qedhere
  \end{enumerate}
\end{proof}

\subsection{Partial classification for non-prime powers}\label{subsec:noprime}

\newcommand{\ORAND}{\operatorname{ORAND}}
\newcommand{\ANDOR}{\operatorname{ANDOR}}

Recall, that when we motivated Mod-CSPs, we said in the Boolean case, there are essentially only three nontrivial cases: $\hornsatmod_G$, $\linmod_G$, and $\twosatmod_G$. As shown in Theorem~\ref{thm:boolean-mod-csp-dicot}, this view is correct when $|G|$ is a prime power, as all three problems admit polynomial time algorithms. When $|G|$ is a non-prime power, the general classification is a bit more complicated. In particular, although $\hornsatmod_G$ fails to have a polynomial-time algorithm, a special case of the problem does.

To define this special case, consider the operators
\begin{align*}
  \ANDOR(x, y, z) &= x \wedge (y \vee z)\\
  \ORAND(x, y, z) &= x \vee (y \wedge z).
\end{align*}

Note that $\ANDOR(x, y, y) = \AND_2(x, y)$ and $\ORAND(x, y, y) = \OR_2(x, y)$, so any CSP with one of these as a polymorphism is pp-reducible to either $\hornsat$ or $\dualhornsat$. But, in the case where the group $G$ does not have prime power order, $\hornsat$ and $\dualhornsat$ have lower bounds away from $\mathsf{P}$ (assuming ETH). Thus, we need to study such problems separately. By the classification of Theorem~\ref{thm:post2}, understanding the case of $\ANDOR$ and $\ORAND$ will be the key to completing the complexity classification of Boolean Mod-CSPs in the case of non-prime power moduli.

The corresponding CSP for $\ANDOR$ is rather simple.

\begin{proposition}[e.g., \cite{post,creignou2008structure,allender2009complexity}]
  Let $\Gamma$ be a Boolean CSP template. If $\ANDOR \in \Pol(\Gamma)$, then $\Gamma$ is pp-reducible to a template $\Lambda_{\ANDOR}$ with constraints of the form
  \[
    \Lambda_{\ANDOR} = \left\{ \{(1)\}, \{(x, y) : x \rightarrow y\} \bigcup_{k = 1}^{\infty}\{(x_1, \hdots, x_k) : \neg x_1 \vee \neg x_2 \vee \cdots \vee \neg x_k\}\right\}
  \]
\end{proposition}
\begin{remark}
  This corresponds to the clone $S_{00}$ and co-clone $IS_{00}$ in Post's lattice.
\end{remark}

For $\ORAND$, the corresponding template is the negation of the above
\[
  \Lambda_{\ORAND} = \left\{\{(0)\}, \{(x, y) : x \rightarrow y\} \bigcup_{k = 1}^{\infty}\{(x_1, \hdots, x_k) : x_1 \vee x_2 \vee \cdots \vee x_k\}\right\}
\]

We now show that both of these problems are tractable for any Abelian group $G$.

\begin{lemma}\label{lem:ANDOR}
  For all finite Abelian groups $G$, $\modcsp(\Lambda_{\ANDOR}, G)$ and $\modcsp(\Lambda_{\ORAND}, G)$ are tractable in $O((n+m)^{|G|})$ time.
\end{lemma}

\begin{proof}
  First, note that these two problems are equivalent up to flipping $0$ and $1$. As a result, we restrict attention to $\modcsp(\Lambda_{\ORAND}, G).$

  The overall algorithm is rather similar to the one for $\twosatmod_G$ with the auxillary set $S \subset G$. As a result, we only state the major differences.

  Consider the directed graph spanned by the implications. Like in the algorithm for $\twosatmod_G$ we can contract the strongly connected components to single variables. Note that this contract operation preserves $x_1 \vee \cdots \vee x_k$ (although it may reduce to a smaller number of variables). Also, if any variables are forced in value, we can propagate that information through the digraph and update the modular constraint.

  Now, take a vertex $x$ of the digraph which has zero outdegree. Set $x = 1$, and solve the remaining instance. If a solution is found, then quit. Otherwise, set $x = 0$, but observe like in the 2-SAT algorithm, any solution with $x = 0$ yields another solution when $x$ is set back to $1$. Thus, we can either expand $S$ in the $x = 0$ branch, or if $S$ does not expand, we can skip the branch entirely. Thus, we get a $O((n+m)^{|G|})$ algorithm like for 2-SAT.
\end{proof}

With this algorithmic result, we need another fact about Post's lattice. Stating this result requires us to define a couple of variants of $\hornsat$ and $\lintwo$.

\begin{itemize}
\item $\hornsatnoconst$ are instances of $\hornsat$ without any constraints of the form $x = 0$ or $x = 1$.
\item $\dualhornsatnoconst$ are instances of $\dualhornsat$ without any constraints of the form $x = 0$ or $x = 1$.
\item $\lintwoevenzero$ are instances of $\lintwo$ where every linear constraint is of the form $x_{i_1} \oplus \cdots \oplus x_{i_k} = 0$, where $k$ is even.
\end{itemize}

Note that the hardness (assuming ETH) of $\modcsp(\hornsatnoconst, G)$ and $\modcsp(\dualhornsatnoconst, G)$ follow from the proof of Proposition~\ref{prop:hornsatmod_hardness}, as no constants are specified in the constructed instance.

The hardness of $\modcsp(\lintwoevenzero, G)$ is a bit more technical and requires a slight modification of the proof of hardness of $\linmod_M$ in Proposition~\ref{prop:linmod_hardness}.

\begin{claim}
  If $|G|$ is divisible by $r \ge 2$ distinct odd primes, then $\modcsp(\lintwoevenzero, G)$ requires $\exp(\Omega_{|G|}((\log n/ \log \log n)^r))$ time assuming $\ETH$.
\end{claim}

\begin{proof}
  Let $M$ be a product of the $r$ distinct odd primes dividing $|G|$. There is a subgroup of $G$ isomorphic to $\mathbb Z/M\mathbb Z$, so it suffices to prove the hardness of $\modcsp(\lintwoevenzero, \mathbb Z/M\mathbb Z)$. The proof is very similar to the proof of Proposition~\ref{prop:linmod_hardness}, so will only sketch the main differences.  Let $m$ be some positive integer to be chosen later. By Proposition~\ref{prop:ORdegree_upperbound_bits}, there exists a degree $O(m^{1/r})$ polynomial $q(z_1,z_2,\dots,z_m)$ which represents $\OR\mod M$ over $\sbits^m$ i.e. $q(z)=0\mod M$ iff $z=\allones$. Let $a_0$ be such that $2a_0$ is not a quadratic residue modulo $M$ i.e. there doesn't exist any $b$ such that $2a_0= b^2 \mod M$. Now define $$p(z_0,z_1,\dots,z_m)=a_0+z_0(q(z_1,\dots,z_m)^2-a_0).$$ Now we claim that $p$ represents $\OR \mod M$ over $\sbits^{m+1}$. If $z_0=1$, then $p(1,z_1,\dots,z_m)=q(z_1,\dots,z_m)^2 = 0 \mod M$ iff $z_1,\dots,z_m=1$. If $z_0=-1$, then $p(-1,z_1,\dots,z_m)=2a_0-q(z_1,\dots,z_m)^2\ne 0\mod M$ for any $z_1,\dots,z_m\in \sbits$.

  Now let $\Psi$ be a $\threesat$ instance with $m$ clauses $C_1(x),\cdots,C_m(x)$ and $t\le 3m$ variables $x_1,\dots,x_t$. We can assume that the variables take $\sbits$ values and $C_i(x)$ is a degree 3 polynomial which takes these $\sbits$ values and outputs $1$ if the clause is satisfied and $-1$ if not. Now consider the polynomial $\Gamma(x_0,x_1,\dots,x_t)=p(x_0,C_1(x),\dots,C_m(x))$. It is easy to see that $\Psi$ is satisfiable iff there exists $x_0,x_1,\dots,x_t\in \sbits$ such that $\Gamma(x_0,x_1,\dots,x_t)=0\mod M$. Let the degree of $\Gamma$ be $\Delta\le O(m^{1/r})$. Let $$\Gamma(x_0,x_1,\dots,x_t)=\gamma_0+\gamma_S \sum_{S\subset \set{0,1,\dots,t}, 1\le |S|\le \Delta} \prod_{i\in S}x_i.$$ Note that every non-constant monomial of $\Gamma$ has the variable $x_0$ in it. When we convert it to an instance of $\linmod_M$ as in the proof of Proposition~\ref{prop:linmod_hardness}, then the subspace is spanned by $\bu_0,\bu_1,\dots,\bu_t$ where $\bu_0=\allones$ as $x_0$ appears in every non-constant monomial. When we write $\bx' \in \linearspan{\bu_0,\bu_1,\dots,\bu_t}$ as $U^\perp\bx'=0$, all the rows of $U^\perp$ are orthogonal to $\bu_0=\allones$ and therefore have even number of 1s. Thus we get an instance of $\lintwo$ where every equation has an even number of variables and the constant term is zero, which is precisely $\lintwoevenzero$. The final instance has $n\le \binom{t+1}{\le \Delta}\le \binom{3m}{\le O(m^{1/r})}$ variables. Therefore we can choose $m=\Omega_M((\log n/\log \log n)^r)$.
\end{proof}

From these, we can classify a slice of Post's lattice.

\begin{theorem}[e.g., \cite{post,creignou2008structure,allender2009complexity}]\label{thm:post2}
  Let $\Gamma$ be a Boolean CSP. Then, if none of $\ANDOR, \ORAND, \MAJ_3 \in \Pol(\Gamma)$, then one of the problems $\hornsatnoconst$,\\ $\dualhornsatnoconst$ or $\lintwoevenzero$ is pp-reducible to $\Gamma$.
\end{theorem}

From this, we can now prove the classification for the non-prime-power case.

\begin{theorem}\label{thm:boolean-mod-csp-dicot2}
  Let $\Gamma$ be a Boolean CSP template. Let $G$ be a nontrivial Abelian group, such that $|G| = p_1^{e_1} \hdots p_r^{e_r}$ where $p_1 < p_2 < \cdots < p_r$ are primes, $e_1, \hdots, e_r \ge 1$ and $r \ge 2$. Then, we have the following classification.

  \begin{enumerate}
  \item If one of $\MAJ_3, \ANDOR, \ORAND \in \Pol(\Gamma)$, then $\modcsp(\Gamma, G) \in \mathsf{P}$.
  \item Otherwise, if $p_1 = 2$ and $r = 2$ and $\XOR_3 \in \Pol(\Gamma)$. Then, $\modcsp(\Gamma, G) \in \mathsf{QP}$.
  \item Otherwise, if one of $\XOR_3$, $\AND_2$, $\OR_2$ is in $\Pol(\Gamma)$, then $\modcsp(\Gamma, G)$ cannot be solved in less than quasi-polynomial time, assuming the exponential time hypothesis.
  \item Otherwise, $\modcsp(\Gamma, G)$ is $\mathsf{NP}$-complete.
  \end{enumerate}
\end{theorem}

\begin{proof}
  Again, we prove the results in order.
  \begin{enumerate}
  \item This follows from Section~\ref{sec:2sat} and Lemma~\ref{lem:ANDOR}.

  \item This follows from Corollary~\ref{cor:lin2mod_rounds}.

  \item By Theorem~\ref{thm:post2}, we have that one of $\hornsatnoconst$,\\$\dualhornsatnoconst$, or $\lintwoevenzero$ pp-reduces to $\Gamma$, so the quasi-polynomial lower bounds (assuming ETH) of Corollary~\ref{cor:subspace_modm_pfr} and Corollary~\ref{prop:ORdegree_upperbound_bits} apply.

  \item By Theorem~\ref{thm:post} and Lemma~\ref{lem:mod-csp-hardness}, we have this $\mathsf{NP}$-completeness result. \qedhere

  \end{enumerate}
\end{proof}

\bibliographystyle{alpha}
\bibliography{references,csp_master}

\newcommand{\etalchar}[1]{$^{#1}$}
\begin{thebibliography}{CGPT06}

\bibitem[ABI{\etalchar{+}}09]{allender2009complexity}
Eric Allender, Michael Bauland, Neil Immerman, Henning Schnoor, and Heribert
  Vollmer.
\newblock The complexity of satisfiability problems: Refining schaefer's
  theorem.
\newblock {\em Journal of Computer and System Sciences}, 75(4):245--254, 2009.

\bibitem[AWZ17]{ArtmannWZ17}
Stephan Artmann, Robert Weismantel, and Rico Zenklusen.
\newblock A strongly polynomial algorithm for bimodular integer linear
  programming.
\newblock In {\em Proceedings of the 49th Annual ACM SIGACT Symposium on Theory
  of Computing}, pages 1206--1219. ACM, 2017.

\bibitem[BBR94]{BarringtonBR94}
David A~Mix Barrington, Richard Beigel, and Steven Rudich.
\newblock Representing boolean functions as polynomials modulo composite
  numbers.
\newblock {\em Computational Complexity}, 4(4):367--382, 1994.

\bibitem[BDL14]{BhowmickDL14}
Abhishek Bhowmick, Zeev Dvir, and Shachar Lovett.
\newblock New bounds for matching vector families.
\newblock {\em SIAM Journal on Computing}, 43(5):1654--1683, 2014.

\bibitem[BG16]{DBLP:conf/coco/BrakensiekG16}
Joshua Brakensiek and Venkatesan Guruswami.
\newblock New hardness results for graph and hypergraph colorings.
\newblock In Ran Raz, editor, {\em 31st Conference on Computational Complexity,
  {CCC} 2016, May 29 to June 1, 2016, Tokyo, Japan}, volume~50 of {\em LIPIcs},
  pages 14:1--14:27. Schloss Dagstuhl - Leibniz-Zentrum fuer Informatik, 2016.

\bibitem[BJK05]{Bulatov2005}
A.~Bulatov, P.~Jeavons, and A.~Krokhin.
\newblock Classifying the {Complexity} of {Constraints} {Using} {Finite}
  {Algebras}.
\newblock {\em SIAM Journal on Computing}, 34(3):720--742, January 2005.

\bibitem[BKW17]{DBLP:conf/dagstuhl/BartoKW17}
Libor Barto, Andrei~A. Krokhin, and Ross Willard.
\newblock Polymorphisms, and how to use them.
\newblock In Andrei~A. Krokhin and Stanislav Zivny, editors, {\em The
  Constraint Satisfaction Problem: Complexity and Approximability}, volume~7 of
  {\em Dagstuhl Follow-Ups}, pages 1--44. Schloss Dagstuhl - Leibniz-Zentrum
  fuer Informatik, 2017.

\bibitem[BM10]{lmcs:1025}
Andrei~A. Bulatov and Daniel Marx.
\newblock The complexity of global cardinality constraints.
\newblock {\em Logical Methods in Computer Scienced}, Volume 6, Issue 4,
  October 2010.

\bibitem[Bul06]{DBLP:journals/jacm/Bulatov06}
Andrei~A. Bulatov.
\newblock A dichotomy theorem for constraint satisfaction problems on a
  3-element set.
\newblock {\em J. {ACM}}, 53(1):66--120, 2006.

\bibitem[Bul17]{DBLP:conf/focs/Bulatov17}
Andrei~A. Bulatov.
\newblock A dichotomy theorem for nonuniform csps.
\newblock In Umans \cite{DBLP:conf/focs/2017}, pages 319--330.

\bibitem[CGPT06]{ChattopadhyayGPT06}
Arkadev Chattopadhyay, Navin Goyal, Pavel Pudl{\'{a}}k, and Denis
  Th{\'{e}}rien.
\newblock Lower bounds for circuits with mod{\_}m gates.
\newblock In {\em 47th Annual {IEEE} Symposium on Foundations of Computer
  Science ({FOCS} 2006), 21-24 October 2006, Berkeley, California, USA,
  Proceedings}, pages 709--718, 2006.

\bibitem[Che09]{Chen2009}
Hubie Chen.
\newblock A {Rendezvous} of {Logic}, {Complexity}, and {Algebra}.
\newblock {\em ACM Comput. Surv.}, 42(1):2:1--2:32, December 2009.

\bibitem[CKZ08]{creignou2008structure}
Nadia Creignou, Phokion~G Kolaitis, and Bruno Zanuttini.
\newblock Structure identification of boolean relations and plain bases for
  co-clones.
\newblock {\em Journal of Computer and System Sciences}, 74(7):1103--1115,
  2008.

\bibitem[CN10]{cook2010logical}
Stephen Cook and Phuong Nguyen.
\newblock {\em Logical foundations of proof complexity}, volume~11.
\newblock Cambridge University Press Cambridge, 2010.

\bibitem[CSS10]{Creignou:2010:NBC:1805950.1805954}
Nadia Creignou, Henning Schnoor, and Ilka Schnoor.
\newblock Nonuniform boolean constraint satisfaction problems with cardinality
  constraint.
\newblock {\em ACM Trans. Comput. Logic}, 11(4):24:1--24:32, July 2010.

\bibitem[CT15]{CohenT15}
Gil Cohen and Avishay Tal.
\newblock Two structural results for low degree polynomials and applications.
\newblock {\em Approximation, Randomization, and Combinatorial Optimization.
  Algorithms and Techniques}, page 680, 2015.

\bibitem[CW09]{ChattopadhyayW09}
Arkadev Chattopadhyay and Avi Wigderson.
\newblock Linear systems over composite moduli.
\newblock In {\em 50th Annual {IEEE} Symposium on Foundations of Computer
  Science, {FOCS} 2009, October 25-27, 2009, Atlanta, Georgia {USA}}, pages
  43--52, 2009.

\bibitem[DG16]{DvirG16}
Zeev Dvir and Sivakanth Gopi.
\newblock 2-server pir with subpolynomial communication.
\newblock {\em Journal of the ACM (JACM)}, 63(4):39, 2016.

\bibitem[DGY11]{DvirGY11}
Zeev Dvir, Parikshit Gopalan, and Sergey Yekhanin.
\newblock Matching vector codes.
\newblock {\em SIAM Journal on Computing}, 40(4):1154--1178, 2011.

\bibitem[Efr12]{Efremenko12}
Klim Efremenko.
\newblock 3-query locally decodable codes of subexponential length.
\newblock {\em SIAM Journal on Computing}, 41(6):1694--1703, 2012.

\bibitem[FV98]{DBLP:journals/siamcomp/FederV98}
Tom{\'{a}}s Feder and Moshe~Y. Vardi.
\newblock The computational structure of monotone monadic {SNP} and constraint
  satisfaction: {A} study through datalog and group theory.
\newblock {\em {SIAM} J. Comput.}, 28(1):57--104, 1998.

\bibitem[GL16]{Guruswami2SAT2016}
Venkatesan Guruswami and Euiwoong Lee.
\newblock Complexity of approximating csp with balance / hard constraints.
\newblock {\em Theory of Computing Systems}, 59(1):76--98, Jul 2016.

\bibitem[Gop09]{Gopalan09}
Parikshit Gopalan.
\newblock A note on {E}fremenko's locally decodable codes.
\newblock In {\em Electronic Colloquium on Computational Complexity (ECCC)},
  volume~16, 2009.

\bibitem[Gop14]{Gopalan14}
Parikshit Gopalan.
\newblock Constructing ramsey graphs from boolean function representations.
\newblock {\em Combinatorica}, 34(2):173--206, 2014.

\bibitem[Gro00]{Grolmusz00}
Vince Grolmusz.
\newblock Superpolynomial size set-systems with restricted intersections mod 6
  and explicit ramsey graphs.
\newblock {\em Combinatorica}, 20(1):71--86, 2000.

\bibitem[IP01]{ImpagliazzoP01}
Russell Impagliazzo and Ramamohan Paturi.
\newblock On the complexity of k-sat.
\newblock {\em Journal of Computer and System Sciences}, 62(2):367--375, 2001.

\bibitem[IPZ01]{ImpagliazzoPZ01}
Russell Impagliazzo, Ramamohan Paturi, and Francis Zane.
\newblock Which problems have strongly exponential complexity?
\newblock {\em Journal of Computer and System Sciences}, 63(4):512--530, 2001.

\bibitem[Jea98]{jeavons88}
Peter Jeavons.
\newblock On the algebraic structure of combinatorial problems.
\newblock {\em Theoretical Computer Science}, 200(1--2):185 -- 204, 1998.

\bibitem[KDW04]{KerenidisdW04}
Iordanis Kerenidis and Ronald De~Wolf.
\newblock Exponential lower bound for 2-query locally decodable codes via a
  quantum argument.
\newblock {\em Journal of Computer and System Sciences}, 69(3):395--420, 2004.

\bibitem[Kro67]{krom1967decision}
Melven~R Krom.
\newblock The decision problem for a class of first-order formulas in which all
  disjunctions are binary.
\newblock {\em Mathematical Logic Quarterly}, 13(1-2):15--20, 1967.

\bibitem[KT00]{KatzT00}
Jonathan Katz and Luca Trevisan.
\newblock On the efficiency of local decoding procedures for error-correcting
  codes.
\newblock In {\em Proceedings of the thirty-second annual ACM symposium on
  Theory of computing}, pages 80--86. ACM, 2000.

\bibitem[LV18]{LiuV18}
Tianren Liu and Vinod Vaikuntanathan.
\newblock Breaking the circuit-size barrier in secret sharing.
\newblock In {\em Proceedings of the 50th Annual {ACM} {SIGACT} Symposium on
  Theory of Computing, {STOC} 2018, Los Angeles, CA, USA, June 25-29, 2018},
  pages 699--708, 2018.

\bibitem[MR10]{MoshkovitzR10}
Dana Moshkovitz and Ran Raz.
\newblock Two-query pcp with subconstant error.
\newblock {\em Journal of the ACM (JACM)}, 57(5):29, 2010.

\bibitem[NSZ18]{NageleSZ18}
Martin N{\"a}gele, Benny Sudakov, and Rico Zenklusen.
\newblock Submodular minimization under congruency constraints.
\newblock In {\em Proceedings of the Twenty-Ninth Annual ACM-SIAM Symposium on
  Discrete Algorithms}, pages 849--866. Society for Industrial and Applied
  Mathematics, 2018.

\bibitem[Pos41]{post}
Emil~L Post.
\newblock {\em The two-valued iterative systems of mathematical logic}.
\newblock Number~5 in Annals of Mathematics Studies. Princeton University
  Press, 1941.

\bibitem[Sch78]{Schaefer:1978}
Thomas~J. Schaefer.
\newblock The complexity of satisfiability problems.
\newblock In {\em Proceedings of the Tenth Annual ACM Symposium on Theory of
  Computing}, STOC '78, pages 216--226, New York, NY, USA, 1978. ACM.

\bibitem[TB98]{BarringtonT98}
G{\'a}bor Tardos and DA~Mix Barrington.
\newblock A lower bound on the mod 6 degree of the or function.
\newblock {\em Computational Complexity}, 7(2):99--108, 1998.

\bibitem[Uma17]{DBLP:conf/focs/2017}
Chris Umans, editor.
\newblock {\em 58th {IEEE} Annual Symposium on Foundations of Computer Science,
  {FOCS} 2017, Berkeley, CA, USA, October 15-17, 2017}. {IEEE} Computer
  Society, 2017.

\bibitem[Yek08]{Yekhanin08}
Sergey Yekhanin.
\newblock Towards 3-query locally decodable codes of subexponential length.
\newblock {\em Journal of the ACM (JACM)}, 55(1):1, 2008.

\bibitem[Zhu17]{DBLP:conf/focs/Zhuk17}
Dmitriy Zhuk.
\newblock A proof of {CSP} dichotomy conjecture.
\newblock In Umans \cite{DBLP:conf/focs/2017}, pages 331--342.

\end{thebibliography}

\end{document}